\documentclass[11pt]{article}
\usepackage{amsmath}
\usepackage{amssymb}
\usepackage{amsthm}
\usepackage{anysize}
\usepackage{rotating}

\usepackage{amsmath,amsfonts,amsthm,amssymb}
\usepackage{pdflscape}

\marginsize{25mm}{30mm}{30mm}{30mm}
\usepackage{nonfloat}

\newcommand{\beq}  {\begin{equation}}
\newcommand{\eeq}  {\end{equation}  }

\newcommand{\bit}  {\begin{itemize}}
\newcommand{\eit}  {\end{itemize}  }

\newcommand{\ben}  {\begin{enumerate}}
\newcommand{\een}  {\end{enumerate}  }

\newcommand{\bpr}  {\begin{proof}}
\newcommand{\epr}  {\end{proof}  }

\newcommand{\R}    {\mathbb{R}}

\newcommand{\la}{\lambda}
\newcommand{\ve}{\varepsilon}

\usepackage{amsmath}
\usepackage{amssymb}
\usepackage{amsthm}
\usepackage{graphicx}
\usepackage[round]{natbib}
\usepackage{todonotes}

\newtheorem{theorem}{Theorem}[section]
\newtheorem{lemma}[theorem]{Lemma}
\newtheorem{proposition}[theorem]{Proposition}

\newtheorem{assumption}[theorem]{Assumption}
\newtheorem{remark}[theorem]{Remark}

\title{ Fast Hybrid Schemes for Fractional Riccati Equations\\ (Rough is not so Tough)}

\author{Giorgia Callegaro \thanks{Department of Mathematics ``Tullio Levi Civita'',
University of Padova, via Trieste 63, 35121 Padova, Italy. Email: gcallega@math.unipd.it}\\
\and
Martino Grasselli \thanks{Department of Mathematics ``Tullio Levi Civita'',
University of Padova, via Trieste 63, 35121 Padova, Italy, and Devinci Research Center, P\^ole Universitaire L\'eonard de Vinci, Paris la Defense. Email: grassell@math.unipd.it}
\and
Gilles Pag\`es \thanks{Laboratoire de Probabilit\'es et Mod\`eles Al\'eatoires, Universit\'e Pierre et Marie Curie, Paris 75252 Paris Cedex 5. Email: gilles.pages@upmc.fr. Acknowledgments: we thank Omar El Euch and Mathieu Rosenbaum for useful discussions on the first part of the paper. We are also grateful to Elia Smaniotto and Giulio Pegorer for valuable comments.}
}

\begin{document}
\maketitle
\begin{abstract}
We solve a family of fractional Riccati equations with constant (possibly complex) coefficients. These equations arise, e.g., in fractional Heston stochastic volatility models, that have received great attention in the recent financial literature thanks to their ability to reproduce a rough volatility behavior.
		We first consider the case of a zero initial value corresponding to the characteristic function of the log-price. Then we investigate the case of a general starting value associated to a transform also involving the volatility process.
		The solution to the fractional Riccati equation takes the form of power series, whose convergence domain is typically finite. This naturally suggests a hybrid numerical algorithm to explicitly obtain the solution also beyond the convergence domain of the power series. Numerical tests show that the hybrid algorithm is extremely fast and stable. When applied to option pricing, our method largely outperforms the only available alternative, based on the Adams method. 
\end{abstract}

{\bf 2010 Mathematics Subject Classification}. 60F10, 91G99, 91B25.

{\bf Keywords}: fractional Brownian motion, fractional Riccati equation,
		rough Heston model, power series representation.

\section{Introduction and Motivation}
	
Stochastic volatility models have received great attention in the last decades in the financial community. The most celebrated model is probably the one introduced by~\cite{HESTON93}, where the asset price $S$ has a diffusive dynamics with a stochastic volatility following a square root process driven by a Brownian motion  partially correlated  with the one driving the underlying. This correlation is important in order to capture the leverage effect, a stylized feature observed in the option market that translates into a skewed  implied volatility surface. The Heston model is also able to reproduce other stylized facts, as fat tails for the distribution of the underlying and time-varying volatility. What is more, the characteristic function of the asset price can be computed in closed form, so that the Heston model turns out to be highly tractable insofar option pricing as well as calibration can be efficiently performed through Fourier methods. This analytical tractability is probably the main reason behind the success of the Heston model among practitioners.

\medskip Recently, there has been an increasing attention in the literature to some roughness phenomena observed in the volatility behaviour of high frequency data, which suggest that the log-volatility is very well modeled by a fractional Brownian motion with Hurst parameter of order $0.1$, see e.g.~\cite{Euch2017},~\cite{Mattieu2016},~\cite{gatheral2014}. 
From a practitioner's perspective, rough volatility models would in principle allow for a good fit of the whole volatility surface, in a parsimonious way. Nevertheless, being the fractional Brownian motion non-Markovian, mathematical tractability might be a challenge.
The idea of introducing a fractional Brownian motion in the volatility noise is not new and it goes back, to the best of our knowledge, to~\cite{Fabienne1998}, where the authors extend the~\cite{HullWhite87} stochastic volatility model to the case where the volatility displays long-memory, in order to capture the empirical evidence of persistence of the stochastic feature of the Black Scholes implied volatilities, when time to maturity increases. Long-memory is associated to a Hurst index greater than $0.5$, while the classic Brownian motion case corresponds to a Hurst parameter equal to 0.5. As the  debate on the empirical value for the Hurst index is still controversial in the literature, in our paper we will consider settings which include the complete range of the Hurst coefficient, namely $H\in (0,1)$.

 \medskip A fractional adaptation of the classical Heston model has come under the spotlight (see e.g. the papers \cite{Euch2017}, \cite{gatheral2014} and \cite{Mattieu2016}), since, in this case, pricing of European options is still feasible and hedging portfolios for vanilla options are explicit. In addition, the fractional version of the Heston model is able to reproduce the slope of the skew for short term expiring options without the need of introducing jumps as in the classical Heston model.

When extending the Heston model to the case where the volatility process is driven by a fractional Brownian motion, one faces some challenges due to the fact that the model is no longer Markovian, due to the presence of memory in the volatility process. On the one hand, the model keeps the affine structure, so that  the computation of the characteristic function of the log-price is still associated to the solution of a quadratic ODE as in the classic Heston case. On the other hand, such Riccati ODE involves fractional derivatives and their solution is no longer available in closed form.  The Adams discretization scheme (see e.g.~\cite{Diethelm2002},~\cite{Diethelm2004}) is the standard numerical method to deal with fractional ODEs. As an alternative, a rational approximation, based on Pad\'e approximants, has been recently proposed in \cite{GR2019}, who started from the short time expansion of the solution as developed in \cite{AGR2019}. Unfortunately, they considered only a very specific set of parameters (and only one value for the argument of the Fourier transform) and they did not provide a detailed analysis of the error and the computational time, so that it seems difficult to benchmark to their results. From a numerical viewpoint, algorithms based on the Adams method, which is basically an Euler scheme of the equation, are not well performing due to the presence of a discrete time convolution induced by the fact that the fractional derivative is not a local operator. In this respect, Runge-Kutta schemes do not seem to be appropriate. On the contrary, the Richardson-Romberg extrapolation method is easy to implement since it consists of a linear combination of solutions of Euler schemes with coarse and refined  steps, so that the same accuracy can be obtained with a dramatic reduction of the computation time (see~\cite{Tubaro90} and~\cite{Gilles07} who developed and popularized the same paradigm in a stochastic environment). 
One can reach and even outperform in the multi-step case the rate obtained by Euler schemes for regular ODEs, which is known to be proportional to the inverse of the complexity. 

 \medskip In this paper we study the efficient computation of the solution of the fractional Riccati ODEs arising from the (fractional) Heston model with constant coefficients for a general Hurst index $H$ ranging in $(0,1)$.  It is also worth mentioning paper \cite{Gerhold2018}, where the authors exploited the Volterra integral representation for the solution to the Riccati ODE, with null initial condition, in order to find upper and lower bounds for its explosion time. In the specific case of the rough Heston model, they tried a fractional power series ansatz for the solution to the Riccati and, via the Cauchy-Hadamard formula they proposed an approximation of the explosion time (see Theorems 7.5 and 7.6 therein).

\medskip We show that it is possible to represent the solution as a power series in a neigbourhood of $0$ and we determine upper and lower bounds for its convergence domain.  It is important to notice that the existence domain of the solution does not always coincide with the convergence domain of the power series (we will see that this typically happens when the coefficients of the fractional Riccati ODE have different signs), in analogy with the fact that the function $x/(x+1)$ is well defined on $(-1,+\infty)$, despite the convergence domain of its power series expansion is only defined for $\vert x\vert <1$. From a computational point of view, the expansion  we propose is extremely efficient compared with the Richardson-Romberg extrapolation method on its domain of existence. If the solution is needed at a date which is beyond the convergence interval, we propose a hybrid numerical scheme that combines our series expansion together with the Richardson-Romberg machinery. The resulting algorithm turns out to be flexible and still very fast, when compared with the benchmark available in the literature, based on the Adams method.

\medskip
The fractional Riccati ODE associated to the characteristic function of the log-asset price is very special insofar it starts from zero. More general transforms (including the characteristic function of the volatility process) lead to non zero initial conditions, see e.g. ~\cite{Pulido2017}, where the authors extend the results of ~\cite{Euch2017} to the case where the volatility is a Volterra process, which includes the (classic and) fractional Heston model for some particular choice of the kernel. The extension of our results to the case of a general (non-null) initial condition is not straightforward and requires additional care. 
Nevertheless, we will show that it is still possible to provide bounds for the convergence domain of the corresponding power series expansion, at the additional cost of extending the  implementation of the algorithm to a  doubly indexed series, in the spirit of ~\cite{jacquier2014}. 

\bigskip
\noindent {\sc Notation.}
\medskip

\noindent $\bullet$ $|z|$ denotes the modulus of the complex number $z\!\in\mathbb{C}$ and $\Re e(z)$ and $\Im m(z)$ its real and imaginary part respectively.

\medskip
\noindent $\bullet$ $x_{_\pm}= \max(\pm x,0)$, $x\!\in \R$.

\medskip
\noindent $\bullet$ $\Gamma(a)= \int_0^{+\infty} u^{a-1}e^{-u}du$, $a>0$  and $B(a,b) = \int_0^1 u^{a-1} (1-u)^{b-1}du$, $a,\, b>0$. We will use extensively  the classical identities $\Gamma(a+1)= a\Gamma(a)$ and $B(a,b) = \frac{\Gamma(a)\Gamma(b)}{\Gamma(a+b)}$.

\medskip
\noindent $\bullet$ $L_p([a,b])$ denotes the set of all Lebesgue measurable functions $f$ such that $\int_{[a,b]} | f(x)|^p dx  < +\infty$, for $1 \le p < \infty$. 

\medskip
\noindent $\bullet$ $\textrm{AC}([a,b])$ for $-\infty \le a < b \le +\infty$, denotes the space of absolutely continuous functions on $[a,b]$. A function $f$ is absolutely continuous if for any $\epsilon >0$ there exists a $\delta>0$ such that for any finite set of pairwise nonintersecting intervals $[a_k,b_k] \subset [a,b]$, $k=1,2 \dots$, such that $\sum_{k=1}^n (b_k - a_k) < \delta$ we have $\sum_{k=1}^n  |f(b_k) - f(a_k)| < \epsilon$. 

\medskip
\noindent $\bullet$ $\textrm{AC}^n([a,b])$, for $n=1,2,\dots$ and  for $-\infty \le a < b \le +\infty$, denotes the space of continuous functions $f$ which have continuous derivatives up to order $(n-1)$ on $[a,b]$, with $f^{(n-1)} \in \textrm{AC}([a,b])$. 

\section{The Problem}

We start by recalling the fractional version of the Heston model, where the pair $(S,V)$ of the stock (forward) price and its instantaneous variance has the dynamics
\begin{equation}\label{eq:frHeston}
\left\{
\begin{array}{rcl}
d S_t & = & S_t \sqrt{V_t} dW_t, \quad S_0=s_0 \in \mathbb R_+ \\
V_t &=& V_0 + \frac{1}{\Gamma(\alpha)}  \int_0^t (t-s)^{\alpha-1} \eta (m - V_s) ds + \frac{1}{\Gamma(\alpha)}  \int_0^t (t-s)^{\alpha-1} \eta \zeta \sqrt{ V_s} dB_s, V_0 \in \mathbb R_+,
\end{array}
\right.
\end{equation}
where $\eta, m, \zeta$ are positive real numbers and the correlation between the two Brownian motions $W$ and $B$ is $\rho \in (-1, 1)$. The parameter $\alpha \in (0,2)$ plays a crucial role (see Remark~\ref{rem:alpha} below). Notice that the classical~\cite{HESTON93} model corresponds to the case $\alpha=1$.
\begin{remark}\label{rem:alpha}
The smoothness of the volatility trajectories is governed by $\alpha$. Recall that the fractional Brownian motion $W^H$, where $H \in (0,1)$ is the Hurst exponent, admits $e.g.$ the Manderlbrot-Van Ness representation
$$
W_t^H = \frac{1}{\Gamma(H+\frac{1}{2})} \int_{- \infty}^0 \left( {(t-s)}^{H-\frac{1}{2}} - (-s)^ {H-\frac{1}{2}}  \right) dW_s + \frac{1}{\Gamma(H+\frac{1}{2})} \int_{0}^t {(t-s)}^{H-\frac{1}{2}}  dW_s 
$$
where the Hurst parameter plays a crucial role in the path's regularity of the kernel ${(t-s)}^{H-\frac{1}{2}}$. In particular, when $H<\frac{1}{2}$ the Brownian integral has Holder regularity and it allows for a rough behavior (see~\cite{Euch2017}).  So, defining $\alpha = H + \frac{1}{2}$ and taking $\alpha< 1$ in the dynamics ~\eqref{eq:frHeston} leads to a rough behavior of the trajectories of $V$.
\end{remark}

The starting point of our work  is the key Theorem 4.1 in~\cite{Euch2017}, which has been extended by~\cite{Pulido2017} (see Theorem~4.3 and Example~7.2 therein) to the class of affine Volterra processes.  
More precisely,~\cite{Euch2017} showed that the characteristic function of the log-price $X_T:= \log(S_T/S_0) $, for $T >0$ and $u_1\! \in \imath \mathbb{R}$, reads
\begin{equation}\label{charF}
 \mathbb E(e^{u_1 X_T}) = \exp \left[ \phi_1(T) + V_0 \ \phi_2(T) \right]
\end{equation}
where 
\begin{equation}\label{eq:charFb}
\phi_1(T) = m \,\eta \int_0^T   \psi(s) ds, \quad \phi_2(T)=I_{1-\alpha} \psi(T)
\end{equation}
and $\psi$ solves the fractional Riccati equation for $t\in[0,T]$
\begin{equation}\label{fracRiccati0}
D^\alpha \psi(t) = \frac{1}{2} (u_1^2 - u_1) + \eta (u_1 \rho \zeta - 1) \psi(t) + \frac{{(\eta \zeta)}^2}{2} \psi^2(t), \quad I_{1-\alpha}\psi(0)=0,
\end{equation}
where $D^{\alpha}$ and $I_{1 - \alpha}$ denote, respectively, the {\em Riemann-Liouville} fractional derivative of order $\alpha$ and the {\em Riemann-Liouville} integral of order $(1-\alpha)$. 

Here we briefly recall both definitions, inspired by \cite[Chapter 2]{samko93}. For any $\alpha >0$ and $f: (0,+\infty)\to \R$ in $L_1([0,T])$, the {\em Riemann-Liouville fractional integral of order $\alpha$} is defined as follows
\begin{equation}\label{eq:Integral}
I_{\alpha}f(t)= \frac{1}{\Gamma(\alpha)}\int_0^t (t-s)^{\alpha-1}f(s)ds.
\end{equation}

Note that we skip $0$ in the above fractional integral, thus avoiding the classical notation $I_{\alpha,0+}$.\\
For $\alpha\!\in(0,1)$, we now define the {\em Riemann-Liouville fractional derivative of order $\alpha$} of $f$ as follows:
\begin{align}
D^{\alpha} f(t) = \frac{1}{\Gamma (1 - \alpha)}\frac{d}{dt}\int_0^t (t-s)^{-\alpha}f(s) ds.
\end{align}

A sufficient condition for its existence is $f \in \textrm{AC}([0,T])$. In the case when $\alpha\!\in [1,2)$ we have
\begin{align}
D^{\alpha} f(t) = \frac{1}{\Gamma (2 - \alpha)}\frac{d^2}{dt^2}\int_0^t (t-s)^{1-\alpha}f(s) ds.
\end{align}

A sufficient condition for its existence is $f \in \textrm{AC}^1([0,T])$.

%

\begin{remark}
$(a)$  When $\alpha =1$, $D^{\alpha}$ obviously coincides with the regular differentiation operator and the above Riccati equation reduces to the classic one. 

\smallskip
\noindent $(b)$ Notice that the fractional derivative is also defined for a general $\alpha \ge 1$ as follows:
\begin{equation}\label{eq:Dalphagene}
D^{\alpha} f(t) = \frac{1}{\Gamma (n - \alpha)}\frac{d^n}{dt^n}\int_0^t (t-s)^{n-1-\alpha}f(s) ds\quad\mbox{ where }\quad n=\lfloor \alpha \rfloor+1.
\end{equation} 
A  sufficient condition for its existence is $f \in \textrm{AC}^{\lfloor \alpha \rfloor }([0,T])$.
\end{remark}

More generally,~\cite{Pulido2017} in their Example 7.2 proved that for $\Re e(u_1) \in [0,1], \Re e(u_2)\leq 0$,
\begin{equation}\label{charFbis}
\mathbb E \left( e^{u_1 X_T + u_2 V_T} \right) = \exp \left[ \phi_1(T) + V_0 \ \phi_2(T) \right],
\end{equation}
where 
$\phi_1,\phi_2$ are defined as before 
and $\psi$ solves the same  fractional Riccati equation ~\eqref{fracRiccati0} with a different initial condition:
\begin{equation}\label{intiCondPulido}
I_{1 - \alpha} \psi(0)=u_2 .
\end{equation}
This transform can be useful in view of pricing volatility products, as it involves the joint distribution of the asset price and the volatility. Obviously, once the characteristic function is known, option pricing can be easily performed through standard Fourier techniques.

\medskip Our first aim in this paper is  to solve the  fractional Riccati {\it ODE} ~\eqref{fracRiccati0} with constant coefficients  when $\alpha\!\in (0,2]$.  From now on, we relabel  the coefficients as follows 
\begin{equation} \label{eq:Riccatialpha}
(\mathcal{E}^{u,v}_{\lambda, \mu,\nu}) \,\equiv\, D^{\alpha} \psi = \lambda  \psi^2 +\mu \psi+\nu, \;\left\{\begin{array}{ll} I_{1 - \alpha} \psi(0)=u&\mbox{if }\; \alpha\!\in (0,1]\\
 I_{1 - \alpha} \psi(0)=u \mbox{ and }I_{2 - \alpha}\psi(0)= v &\mbox{if }\; \alpha\!\in (1,2],
\end{array}\right.
\end{equation}
where $ \lambda$, $\mu$, $\nu$ and $u$, $v$  are {\em complex} numbers (when $\alpha\!\in (0,1]$ we use $(\mathcal{E}^{u}_{\lambda, \mu,\nu})$). 

We will propose an efficient numerical method to compute the  solution, with a special emphasis on the case where $\alpha\!\in (0,1)$ and the initial condition $u$ is equal to zero, corresponding to the characteristic function of the log-asset price. 

 \begin{remark}
a) Being the Hurst coefficient $H=\alpha-\frac 12$, the case $\alpha\!\in (0,1)$  contains  the rough volatility modeling whereas the case $\alpha\!\in (1,2)$ contains the long memory modeling and corresponds to the framework of~\cite{Fabienne1998}.\\
b) We refer, respectively, to~\cite{Euch2017} and to~\cite{Pulido2017} for existence  and uniqueness of the solution to the Riccati equation ~\eqref{fracRiccati0}, respectively with null initial condition and with initial condition ~\eqref{intiCondPulido}. Our approach will prove the existence of  a solution in a (right) neighbourhood of $0$. 
\end{remark}

\medskip
One checks that, under appropriate integrability conditions on the function $f$, $(D^{\alpha} \circ  I_{\alpha}) f = f$, so that the Fractional Riccati equation $(\mathcal E^u _{\lambda,\mu,\nu})$ can be rewritten equivalently in a fractional integral form as follows
\begin{equation}\label{eq:Riccating}
\psi(t)=  \frac{u}{\Gamma(\alpha)}t^{\alpha-1}+I_{\alpha}(\lambda \psi^2+\mu \psi+\nu)\quad\mbox{when $\alpha\!\in (0,1]$,}
\end{equation}
with $u\!\in \mathbb{C}$,  and 
\begin{equation}\label{eq:Riccating2}
\psi(t)=  \frac{u}{\Gamma(\alpha)}t^{\alpha-1}+\frac{v}{\Gamma(\alpha-1)}t^{\alpha-2}+I_{\alpha}(\lambda \psi^2+\mu \psi+\nu)\quad \mbox{when $\alpha\!\in (1,2]$,}
\end{equation}
with $u$, $v\!\in \mathbb{C}$. The consistency of such initial conditions follows in both cases from the fact that   $I_{\beta}(t^{- \beta})= \Gamma(1 +\lfloor \beta\rfloor - \beta)$, for $0 < \beta \le 2$. 

\medskip The starting strategy of our approach is to establish the existence of formal solutions to   $(\mathcal E^{u,v} _{\lambda,\mu,\nu})$ as {\em fractional power series expansions} and  then prove by a propagation method of upper/lower bounds that the convergence radius of such series is non zero (and possibly  finite). Indeed, this is strongly suggested by  the elementary computation of the fractional derivative of  a power function $t^r$, $r\!\in \mathbb{R}$:
\begin{align}\label{eq:Dalphatr}
D^{\alpha} t^r = \frac{\Gamma(r+1)}{\Gamma (r+1-\alpha)}t^{r-\alpha}\quad\mbox{if}\; r>\alpha-1\quad \mbox{ and }\quad  D^{\alpha} t^{\alpha-1} = 0.
\end{align}
Similarly
\begin{equation}
\label{eq:Ialphatr}
I_{\alpha}  t^r =  \frac{\Gamma(r+1)}{\Gamma(r+\alpha+1)}t^{\alpha+r} \quad \mbox{if}\; r\neq-1.
\end{equation}

In particular, note that this last property justifies why a natural starting value for $(\mathcal E^u_{\lambda, \mu,\nu})$ is of the form $\frac{u}{\Gamma(\alpha)}t^{\alpha-1}$ since its $\alpha$-derivative is $0$ and its $(1-\alpha)$-integral antiderivative is $u$ owing to the above formulas. On the other hand, the fractional derivative of a constant is not zero and reads:
\begin{align}
D^{\alpha} c = \frac{c}{\Gamma (1 - \alpha)}t^{-\alpha}.
\end{align}
\begin{remark}
When $\alpha =1$, $D^{\alpha}$ obviously coincides with the regular differentiation operator and the above Riccati equation is simply the regular Riccati equation with quadratic right-hand side, for which a closed form solution is available.
\end{remark}

\medskip
 In the first part of this paper we will mostly distinguish two cases:
\begin{itemize}
\item the  case $u=0$ and $\alpha\!\in (0,1]$, which is closely connected with the pricing of options in a rough stochastic volatility model (see~\cite{Mattieu2016, Euch2017}) 
\item the case $u=v=0$ and $\alpha\!\in (1,2]$, which can be seen as a special case of the more general results presented in \cite{Pulido2017},
\end{itemize}
and in a second part, we well investigate the more general case where $u \neq 0$, which requires more care. 
 

\medskip

The property~~\eqref{eq:Dalphatr} shows that the $\alpha$-fractional differentiation preserves the fractional monomials $t^{r\alpha}$, $r\!\in \mathbb{Z}$. This property strongly suggests to solve the above equation as fractional power series, at least in the neighborhood of $0$.  Usually the fractional power series has a finite convergence radius but this does not mean that the solution does not exist outside the interval defined by this radius. This will lead us to  design a hybrid numerical scheme  to solve this equation.


\section{Solving $(\mathcal E^0_{\la,\mu,\nu})$ as a Power Series}\label{sec:Riccati0}

As preliminary remarks before getting onto technicalities, note that:

\begin{itemize}
\item   if $\nu =0$, then the solution to the equation is clearly $0$ by a uniqueness argument.
\item  If $\lambda =0$, the Equation $(\mathcal E^0_{\la,\mu,\nu})$ becomes linear and, as we will see on the way, the  unique solution is  expandable in a fractional power series with an infinite  convergence radius.
\end{itemize}

As a consequence, henceforth we will work, except specific mention, under the following
\begin{assumption}
We assume that $\lambda\nu\neq 0$.
\end{assumption}

   \subsection{The Algorithm}

 The starting idea is to proceed by verification: we  search for a solution as a fractional power series:
\begin{align}\label{eq:DSEalpha}
\psi (t) =\psi_{\lambda,\mu,\nu}(t):= \sum_{k\geq0}a_k t^{k\alpha} 
\end{align}
where the coefficients $a_k$, $k\ge1$,  are complex numbers. We will show that the coefficients $a_k$ are uniquely defined and we will establish that the convergence radius $R_{\psi}$ of $\psi$ is non-zero. 

Assume that $R_{\psi}>0$. First note that, for $0<t<R_{\psi}$ 
\begin{align*}
\psi^2 (t) = \sum_{k\geq0}{a_k^*}^2 \ t^{k\alpha} 
\end{align*}
with the Cauchy coefficients of the discrete time convolution given by
\[
{a_k^*}^2 =\sum_{\ell=0}^{k} a_{\ell} a_{k-\ell},\; k\ge 0.
\]

It follows from~~\eqref{eq:Dalphatr}  that 
\begin{align*}
\nonumber D^{\alpha} \psi (t)& = \sum_{k\geq0}a_k\frac{\Gamma(\alpha k+1)}{\Gamma (\alpha k+1-\alpha)}t^{\alpha (k-1)}
=\sum_{\ell\geq -1}a_{\ell+1}\frac{\Gamma(\alpha (\ell+1)+1)}{\Gamma ( \alpha (\ell+1)+1-\alpha)}\,t^{\alpha \ell}.
\end{align*}

On the other hand, from the Riccati equation we have
\begin{align}\label{eq:Ricatalpha}
D^{\alpha} \psi (t)& = \sum_{k\geq0}\big(\lambda {a_k^*}^2 +\mu a_k\big)t^{\alpha k} +\nu
\end{align}
so that the sequence $(a_k)_{k\ge 0}$ satisfies (by identification of the two expansions for $D^{\alpha} \psi$)   the {\em discrete time convolution equation}:
\begin{align}\label{eq:Eqak0}
 (A_{\lambda, \mu,\nu})\quad\equiv\quad a_{k+1}&= \big(\lambda {a_k^*}^2  +\mu a_k\big)\frac{\Gamma (\alpha k +1)}{\Gamma (\alpha k +\alpha +1 )},\;k\ge 1, \; a_1 =\frac{\nu}{\Gamma (\alpha +1)}, \; a_0=0.
\end{align}
\noindent\begin{remark} As a consequence of $a_0=0$, note that the discrete convolution $a^{*2}_k$ reads
\begin{equation}\label{eq:convoldef1}
{a_1^*}^2=0 \quad \mbox{ and }\quad {a_k^*}^2 =\sum_{\ell=1}^{k-1} a_{\ell} a_{k-\ell},\; k\ge 2.
\end{equation}
\end{remark}

   \subsection{The Convergence Radius}\label{subsec:convR}

Let us recall that the convergence radius  $R_{\psi}$ of the fractional power series~~\eqref{eq:DSEalpha} is given by Hadamard's formula:
\begin{equation}\label{eq:Hadamard}
R_{\psi}= \liminf_k \big|a_k\big|^{-\frac{1}{\alpha k}}\!\in [0, +\infty].
\end{equation}
The fractional power series is absolutely converging for every $t\!\in[0,R_{\psi})$ and diverges outside $[0, R_{\psi}]$. 
(We will not discuss the possible extension on the negative real line of the equation.)  It may also be semi-convergent at $R_{\psi}$ if the $a_k $ are real numbers with an alternate sign and decreasing in absolute value. 

The maximal solution of the equation  may exist beyond this interval: we will see that this occurs for example  when the parameters $\lambda$, $\mu$, $\nu$ satisfy $\lambda\nu>0$ and $\mu<0$.  The typical example being the function $t\mapsto \frac{t}{1+t}$ solution to  $\psi'= \psi^2-2\psi+1$, $\psi(0)=0$ defined on $(-1,+\infty)$ but only expandable (at $0$) on $(-1,1]$. This has to do with the existence of poles on the complex plane of the meromorphic extension of the expansion.

However, if the $a_k$, $k\ge 1$, are all non-negative, one at least being non zero, then the domain of existence of the maximal solution is exactly $[0, R_{\psi})$. Its proof is postponed to Section~\ref{sec:thm3.1}.

\smallskip
The theorem below, which is the first key result of this paper, provides explicit bounds for the convergence radius $R_{\psi}$ for the equation $(\mathcal E^0_{\la, \mu,\nu})$.
\begin{theorem} \label{thm:Radius} Let $\alpha \!\in [0,2]$ and let $\lambda$, $\mu$, $\nu\! \in \mathbb{C}$, $\lambda\neq0$. We denote by $\psi_{\la,\mu,\nu}$ the function defined by~~\eqref{eq:DSEalpha} where the coefficients $a_k$ satisfy $(A_{\la, \mu,\nu})$.

\medskip 
\noindent $(a)$ $[$General lower bound of the radius $]$ We have
\begin{equation}\label{eq:lowerbound}
  R_{\psi_{\lambda,\mu,\nu}} > \frac{2^{\frac{1}{\alpha} -(\frac{1}{\alpha}-2)^{+}}\alpha}{\Big(|\mu| + \sqrt{\mu^2 +c_{\alpha} \frac{ |\lambda|  |\nu|}{\Gamma(\alpha )} }\,\Big)^{\frac{1}{\alpha}}}:=\tau_*>0.
\end{equation}
where $c_{\alpha} =  2^{2-(1-2\alpha)^{+}-2(\alpha-1)^{+}}\alpha^{\alpha-1}B(\alpha\wedge 1,\alpha\wedge 1) >0$.

\medskip
\noindent $(b)$ $[$Upper-bound for the radius $]$ If $\la, \nu>0$ and $\mu\ge 0$ (resp.  $\lambda, \nu <0$ and $\mu\le 0$), then 
\[
  R_{\psi_{\lambda,\mu,\nu}} \le R_{\psi_{|\lambda|,0,|\nu|}}\le C_{\alpha}\left( \frac{\Gamma(\alpha+1)}{ \lambda \nu}\right)^{\frac{1}{2\alpha}}
\quad \mbox{where  }\quad C_{\alpha} =\left\{\begin{array}{ll} \displaystyle\left( 3.5^{\alpha-1}\right)^{\frac{1}{2\alpha}} \sqrt{\alpha}&\mbox{if }\alpha \!\in (0,1],\\
&\\
\displaystyle \frac{\sqrt{2\alpha}}{\widetilde B(\alpha)}&\mbox{if }\alpha \!\in (1,2],
\end{array}\right.  
\]
with $\widetilde B(\alpha)=  B(\alpha,\alpha)-2^{1-2\alpha}>0$.
Moreover, $\psi_{\la, \mu,\nu}$ is increasing  (resp. decreasing) and $\displaystyle \lim_{t\to +R_{\psi_{\la, \mu,\nu}}}\psi_{\la, \mu,\nu}(t) ={\rm sign}(\lambda).\infty$ so that the existence domain of $\psi_{\la, \mu,\nu}$ is $[0, R_{\psi_{\la, \mu,\nu}})$.

\medskip
\noindent $(c)$ If $\la,\nu>0$ and $\mu\le 0$, then (with obvious notations) $a^{(\la,\mu,\nu)}_k= (-1)^{k} a^{(\lambda,-\mu,\nu)}_k$, $k\ge 1$, so that $R_{\psi_{\lambda,\mu,\nu}}=   R_{\psi_{\lambda,-\mu,\nu}}$. Moreover if the sequence $a_k^{(\lambda,-\mu,\nu)}$   decreases for $k$ large enough, then the expansion of $\psi_{\la,\mu,\nu}$  converges at $R_{\psi_{\lambda,-\mu,\nu}}$.

\medskip
\noindent $(d)$ If $\mu=0$, then $a_{2k}= 0$ for every $k\ge 1$  and the sequence $b_k=a^0_{2k-1}$, $k\ge 1$,  is solution to  the recursive equation
\begin{equation}\label{eq:convb}
b_1= \frac{\nu}{\Gamma(\alpha+1)} \quad \mbox{ and  } \quad b_{k+1} = \lambda\, \frac{\Gamma(2\alpha k+1)}{\Gamma((2k+1)\alpha +1)}b^{*2}_{k+1},\; k\ge 1,
\end{equation}
where the squared convolution is  still defined by~~\eqref{eq:convoldef1} (the equation is consistent since  $b^{*2}_{k+1}$ only involves terms $b_\ell$, $\ell\le k$).
\end{theorem}


\begin{remark}
(a) The lower bound is not optimal since, if $\lambda =0$ and $\mu\neq 0$, it is straightforward that 
 \[
 \frac{a_{k+1}}{a_k}\sim  \frac{\Gamma (\alpha k +1)}{\Gamma (\alpha k +\alpha +1 )}\mu \to 0\quad \mbox{ as }\quad k\to +\infty \quad \mbox{so that $R_{\psi_{0,\mu,\nu}}=+\infty$}. 
 \]
(b) In particular the theorem shows  that, if $\lambda$, $\nu >0$,  there exist real constants $0<K_1(\alpha)<K_2(\alpha)$, only depending on $\alpha$, such that 
 \[
  \frac{K_1(\alpha)}{(\lambda\nu)^{\frac{1}{2\alpha}}}\le R_{\psi_{\lambda,0,\nu}}\le \frac{K_2(\alpha)}{(\lambda\nu)^{\frac{1}{2\alpha}}}.
 \]
(c) When $\la$, $\nu>0$ and $\mu\le 0$, the maximal solution of $(\mathcal E^0_{\la,\mu,\nu})$  lives on the whole positive real line, even if its expansion only converges for $t\!\in [0, R_{\psi_{\la,\mu,\nu}}]$. \\
(d) In Table \ref{ValuesTau*} in Section \ref{sec:testFrac} we will provide the values of the convergence domain in some realistic scenarios.
 \end{remark}
 
 \bigskip As already mentioned in the introduction, the domain of existence of the solution to $(\mathcal E^0_{\la, \mu, \nu})$ may be strictly wider than that   of the fractional power series. Hence, it is not possible to rely exclusively on this expansion of the solution to propose  a fast numerical method for solving the equation. The aim is to take optimally advantage of this expansion to devise a {\em hybrid} numerical scheme which works to approximate the solution of the equation everywhere on its domain of existence. 
 
   \subsection{Controlling the ``Remainder'' Term}
In order to control the error induced by truncating the fractional series expansion \eqref{eq:DSEalpha} at any order $n_0$,  we need some errors bounds. In practice we do not know the exact value of the radius $R_{\psi}$. However, we can rely on our theoretical lower bound $\tau_*$
given by the right-hand side of \eqref{eq:lowerbound}. 

An alternative to this theoretical choice is  to compute  $R^{(n)}:= |a_n|^{-\frac {1}{\alpha n}}$ for $n$ large enough   where $(a_n)$ satisfies $(A_{\la,\mu,\nu})$.  The value turns out to be a good approximation of $R_{\psi}$, but may of course overestimate it, which suggests to consider   $\tau_{\psi} = p R^{(n)}$ with $p\!\in [0,0.90]$.

In both cases,  in what follows we assume that $t\!\in (0, \tau_*)$. 

In the proof of Theorem~\ref{thm:Radius}$(a)$, see Section~\ref{sec:thm3.1}, we will show by induction that the sequence $(a_n)_{n\ge 1}$ satisfies
$$
|a_k|\le C_*(\rho_*)^k k^{\alpha-1},\;k\ge 1.
$$
where $\rho_*= (\tau_*)^{-\alpha}$ or, equivalently, $ \tau_*= (\rho_*)^{-\frac{1}{\alpha}}$ and $C_*$ is given by $C_* =  \frac{|\nu|}{\Gamma(\alpha +1)\rho_*}$.


 Then 
\begin{eqnarray*}
\forall\, t\!\in \big(0, \tau_*\big),\quad \left|\psi_{\lambda,\mu,\nu}(t)- \sum_{k=1}^{n_0} a_k t^{k\alpha}\right|&\le &\sum_{k\ge n_0+1} |a_k| t^{k\alpha}\le C_*\sum_{k\ge n_0+1}k^{\alpha-1}(\rho_*t^{\alpha})^k\\
& =& C_*\sum_{k\ge n_0+1}k^{\alpha-1}(t/\tau_*)^{\alpha k}
\end{eqnarray*}
owing to~~\eqref{eq:UpperBoundak} (and its counterpart for $1<\alpha\le 2$ with $\rho_*$ given by ~\eqref{eq:rho_*2}), where $\theta= \rho_*t^{\alpha}= (t/\tau_*)^{\alpha} \!\in (0,1)$.

\medskip
\noindent {\em Case $\alpha\!\in (0,1]$:} The function $\xi \mapsto \xi^{\alpha-1}\theta^\xi$ is decreasing on the positive real line. 
\begin{eqnarray}
\label{eq:ErrorBound}\sum_{k\ge n_0+1} k^{\alpha-1} \theta^k   \le  \int_{n_0}^{+\infty} \xi^{\alpha-1}  \theta^{\xi}\,d\xi 
= \big(\log(1/\theta)\big)^{-\alpha}\int_{n_0\log(1/\theta)}^{+\infty} u^{\alpha-1} e^{-u} du.
\end{eqnarray}

\smallskip
Note that $u^{\alpha-1}\le x^{\alpha-1}$ for every $u\ge x$ since $0<\alpha\le 1$ so that  $ \int_x^{+\infty} u^{\alpha-1}e^{-a u}du\le x^{\alpha-1} \int_x^{+\infty}e^{-au}du = \frac{x^{\alpha-1}e^{-ax}}{a}$. Hence, we deduce that 
\begin{equation*}
\forall\, t\!\in \big(0, \tau_*\big),\quad  \left|\psi_{\lambda,\mu,\nu}(t)- \sum_{k=1}^{n_0} a_k t^{k\alpha}\right| \le \frac{C_*}{\alpha\log(\tau_*/t)} \,\frac{(t/\tau_*)^{n_0\alpha}}{n_0^{1-\alpha}}.
\end{equation*}


\medskip
\noindent 
{\em Case $\alpha\!\in [1,2]$.} Note that the function $\xi \mapsto \xi^{\alpha-1}\theta^\xi$  is now only decreasing over $\left[\displaystyle \frac{\alpha-1}{\log(\vartheta)}, +\infty\right)$ so that~~\eqref{eq:ErrorBound} only holds for $n_0\geq \displaystyle \frac{\alpha-1}{\log(1/\theta)}$(which can be very large if $\theta= (t/\tau_*)^{\alpha}$ is close to $1$). In practice this means that, to compute $\psi$ one should at least consider $n_0$ terms!

\smallskip
To get an upper-bound we perform an integration by part which shows that
\begin{align*}
\int_{x}^{+\infty} u^{\alpha-1} e^{-u} du&= x^{\alpha-1}e^{-x}+(\alpha-1)\int_{x}^{+\infty} u^{\alpha-2}e^{-u}du\\
&\le   x^{\alpha-1}e^{-x}+(\alpha-1) x^{\alpha-2}e^{-x}=  x^{\alpha-1}e^{-x}\Big(1+\frac{\alpha-1}{x}\Big)
\end{align*}
 where in the second line we used that $u^{\alpha-2}\le x^{\alpha-2}$ since $u\ge x$ and  $1\le \alpha\le 2$.
Plugging this in~~\eqref{eq:ErrorBound} with $x =n_0\log(1/\theta)$ and $\theta= t/\tau_*$ yields  the following formula which holds true for every $\alpha\!\in (0,2]$,
\[
\forall\, t\!\in (0, \tau_*),\qquad \left|\psi_{\lambda,\mu,\nu}(t)- \sum_{k=1}^{n_0} a_k t^{k\alpha}\right| \le    \frac{C_*n_0^{\alpha-1} }{\alpha\log(\tau_*/t)}\,(t/\tau_*)^{n_0\alpha }\left(1+\frac{(\alpha-1)_{_+}}{\alpha\, n_0\log(\tau_*/t)}\right).
\]
 
If $\tau_*$ is estimated empirically, the propagation property can be no longer used. Similar bounds, though less precise, can be obtained using that $\tau_*<R_{\psi_{\lambda,\mu,\nu}}$.

In practice,  we will favor this second approach over the use of $\rho_*$, as  $\rho_*$ provides a too conservative lower estimate of $R_{\psi_{\la,\mu,\nu}}$.

 \section{Hybrid Numerical Scheme for $(\mathcal E^0_{\lambda,\mu,\nu})$, $0<\alpha \le 1$}
 \label{sec:hybrid}

The idea now is to mix two approaches to solve the above fractional Riccati equations~~\eqref{eq:Riccatialpha} 
on an interval $[0,T]$, $T>0$, supposed to be included in the domain $D_{\psi}$ on which $\psi$  is defined. We will focus on the first  equation (with $0$ as initial value) for convenience.

 The aim here is to describe a hybrid algorithm to compute the triplet
 \[
\Psi(t)=  \Big(\psi(t), I_1(\psi)(t), I_{1-\alpha}(\psi)(t)\Big)
\]
 at a  given time  $t=T$ where $\psi=\psi_{\lambda,\mu,\nu}$ is solution to $({\cal E}^0_{\lambda,\mu,\nu})$ (see Equation~~\eqref{eq:Riccatialpha}). By hybrid we mean that we will mix (and merge) two methods, one based on the fractional power series expansion of $\psi$ and its two integrals and one based on a time discretization of the equation satisfied by $\psi$ and the integral operators. 
 
 On  the top of that, we will introduce a Richardson-Romberg extrapolation method based  on a conjecture on the existence of  an expansion of the time discretization error. We refer to ~\cite{Tubaro90} and~\cite{Gilles07} for a full explanation of the Richardson-Romberg extrapolation method and its multistep refinements.
 
 As established in Section~\ref{sec:Riccati0}, the   solution $\psi$ can be expanded as a (fractional) power series on $[-R_{\psi},R_{\psi}]$, $R_{\psi}>0$. Namely, for every $t\!\in (-R_{\psi},R_{\psi})$
 \begin{equation}\label{eq:Psibis}
 \psi(t) = \sum_{r\ge 1} a_rt^{\alpha r}.
 \end{equation}
 As a consequence, it is straightforward that 
 \begin{equation}\label{eq:I1Psi}
 I_1(\psi)(t)= \int_0^t \psi(s)ds = t\sum_{r\ge 1}\frac{a_r}{\alpha r+1}t^{\alpha r}
 \end{equation}
 and, using~~\eqref{eq:Ialphatr}, that
 \begin{align}\
\nonumber  I_{1-\alpha}(\psi)(t) &= \sum_{r\ge 1} a_r\frac{\Gamma(\alpha r+1)}{\Gamma(\alpha (r-1)+1)} \frac{t^{r\alpha+1-\alpha}}{\alpha r+1-\alpha}\\
 \label{eq:I1-alphapsi}&= t^{1-\alpha}\left( \nu t^{\alpha}+ \sum_{r\ge 2} a_r\Big(1-\frac 1r\Big)^{-1} \frac{\Gamma(\alpha r)}{\Gamma(\alpha(r-1)) }\frac{ t^{r\alpha}}{\alpha r+1-\alpha}\right).
 \end{align}
 
 We will now proceed in four steps.
 
\subsection{Step~1: radius of the power series expansion}
A preliminary step consists in computing enough coefficients $a_r$ of  the fractional power extension of $\psi$, say $r_{\max}$, and estimating  its convergence radius by
 \[
 R_{\psi} = \liminf_{r}|a_r|^{-\frac{1}{\alpha r}}\simeq  R_{\psi,r_{\max}} := |a_{r_{\max}}|^{-\frac{1}{\alpha r_{\max}}}.
 \]
The radius  $R_{\psi} $ is also that of the two other components of $\Psi(T)$, so
 a more  conservative approach in practice  is to estimate $R_{\psi}$ using the larger coefficients $a'_r := a_r\frac{\Gamma(\alpha r+1)}{\Gamma(\alpha (r-1)+1)}\frac{1}{\alpha r+1-\alpha}$, $r\ge 2$,  coming out in~~\eqref{eq:I1-alphapsi}, {\it i.e.} consider
 \[
  R_{\psi}\simeq  \widehat R_{\psi}:= |a'_{r_{\max}}|^{-\frac{1}{\alpha r_{\max}}}.
 \]
This estimate of the radius  is lower than what would be obtained with the sequence $(a_r)$, which is in favor of a better accuracy of the scheme (see further on).  
 
 Then we decide  the accuracy level we wish for the approximation of these series: let $\ve_0$ denote this level, typically $\ve_0=0.01$ or   $0.001$. If we consider some $t$ close to $R_{\psi}$ (or at least its estimate), we will need to compute too many terms of the series to achieve the prescribed accuracy, so we define a threshold $\vartheta\!\in (0,1)$ and we decide that the  above triplet will be  computed by their series expansion only on $[0,\vartheta \widehat  R_{\psi}]$. Then the prescribed accuracy is satisfied if the above fractional power series expansions are truncated into  sums from $r=1$ up to $r_0$ with
$$
 r_0=r_0(\ve_0, \vartheta)=\left\lceil\frac{\log(\ve_0(1-\vartheta))}{\alpha\log(\vartheta)}-1\right\rceil
 $$
 provided $r_0\le r_{\max}$. If $r_0> r_{\max}$, it suffices to invert the above formula where $r_0$ is replaced by $r_{\max}$ to determine the resulting accuracy of the computation.

\subsection{Step~2: hybrid expansion-Euler discretization scheme} 
We assume in what follows that  $\vartheta R^{r_{\max}}_{\psi}<R_{\psi}$ to preserve the accuracy of the computations of the values of $\psi$ by the fractional power  expansion.

\medskip
$\blacktriangleleft$  {\em  Case $T<\vartheta R^{r_{\max}}_{\psi}$}.  One computes the triplet  $\Psi(T)$ by truncating the three  fractional power extensions as explained above.
%
 
\medskip
$\blacktriangleleft$   {\em Case  $T>\vartheta R^{r_{\max}}_{\psi}$}. This is the case where we need to introduce the hybrid feature of the method.  

\medskip
\noindent  {\em Phase~I: Power series computation.} We will use the power series expansion until $\vartheta R_{\psi}$ and then an Euler scheme {\em with memory} (of course) of the equation in its integral form
  \[
( \mathcal{E}^0_{\lambda,\mu,\nu})\;\equiv \;\psi = I_{\alpha}\Big(\nu +\mu \psi+\lambda \psi^2\Big).
 \]
First we consider a time step of the form $h= \frac Tn$   where $n\ge 1$ is an integer (usually a power of $2$). We denote by $\bar \psi^n$  the Euler discretization  scheme with step $h$. Set 
\[
t_k=t^n_k=\frac{kT}{n},\; k=0:n \quad \mbox{ and let }\quad k_0 =  \Big\lfloor \frac{ n\vartheta \widehat  R_{\psi}}{T}\Big\rfloor
\] 
so that $t_{k_0} \le \vartheta \widehat R_{\psi}<t_{k_0+1}$. Note that $t_{k_0}$ may be equal to $0$.

\noindent\begin{remark} The values $\bar \psi^n (t_{k})$ $k=0,\ldots, k_0$ are not computed as an Euler scheme (in spite of the notations) but  using the fractional power expansion~~\eqref{eq:Psibis} truncated at $r_0$. 
\end{remark}

\medskip
\noindent {\em Phase~II: Plain Euler discretization.}  Then, given   the definition~~\eqref{eq:Integral}  of the fractional integral operator $I_{\alpha}$, one has, for every $t\!\in [0,T]$
 \[
 \psi(t) = \frac{\nu t^{\alpha}}{\Gamma(\alpha+1)}+ \frac{1}{\Gamma(\alpha)}\int_0^t \psi(s)\big(\mu+\lambda \psi(s)\big)(t-s)^{\alpha-1}ds
 \]
so that  the values of $\bar \psi^n(t_k)$ for $k=k_0+1,\ldots,n$ are computed by induction  for every $k\ge k_0+1$ by
 \begin{align}
 \nonumber \bar \psi^n (t_k)&=  \frac{\nu \,t_k^{\alpha}}{\Gamma(\alpha+1)}+\frac{1}{\Gamma(\alpha)}\sum_{\ell=1}^{k-1}\bar \psi^n(t_{\ell})\big(\lambda\bar\psi^n(t_{\ell})+\mu\big)\int_{t_{\ell}}^{t_{\ell+1}}(t_k-s)^{\alpha-1}ds\\
\label{eq:EulerIalpha} &=  \frac{1}{\Gamma(\alpha+1)} \left(\frac Tn\right)^{\alpha} \left( \nu k^{\alpha} +\sum_{\ell=1}^{k-1}c^{(\alpha)}_{k-\ell-1}\bar \psi^n(t_{\ell})\big(\lambda\bar\psi^n(t_{\ell})+\mu\big)\right)
 \end{align}
 where
 \[
 c^{(\alpha)}_0=1 \quad \mbox{ and }\quad c^{(\alpha)}_\ell= (\ell+1)^{\alpha} -\ell^{\alpha}, \; \ell=1:k-2.
 \]

To approximate the other two components $I_1(\psi)(t_k)= \int_0^{t_k}\psi(s)ds$ and $I_{1-\alpha}(\psi)(t_k)$ of $\Psi(t_k)$, we proceed as follows:

\begin{itemize}
\item  
For the regular antiderivative $I_1(\psi)$:  we first decompose the integral into two parts by additivity of regular integral
\[
I_1(\psi)(t_k)= I_1(\psi)(t_{k_0})+ \int_{t_{k_0}}^{t_k} \psi(s)ds.
\]
The first integral is computed by integrating the fractional power series expansion~~\eqref{eq:Psibis} $i.e.$
\[
I_1(\psi)(t_{k_0})=  \int_0^{t_{k_0}} \psi(s)ds\simeq t \sum_{r=1}^{r_0}\frac{a_r}{\alpha\, r +1} \,t_{k_0}^{\alpha r},
\]
while the second  one is computed using  a classical trapezoid   method, namely
\[
I_1(\psi)(t_k) \simeq \int_0^{t_{k_0}}\psi(s)ds + \frac Tn\sum_{\ell=k_0}^{k-1}\bar\psi^n(t_{\ell}) +\frac{T}{2n}\big(\bar\psi^n(t_{k})-\bar\psi^n(t_{k_0})\big).
\]
 \item  For the fractional antiderivative $I_{1-\alpha}(\psi)$: first note that we  could take advantage of the fact that $I_{1-\alpha}\circ I_{\alpha}=I_1$ leading to
\[
 I_{1-\alpha}(\psi)= I_1\big(\nu+\mu \psi+\lambda \psi^2\big)
 \]
 so that, for every $t$, 
 \[
  I_{1-\alpha}(\psi)(t) = \nu \, t +\int_0^t \psi(s) \left( \mu + \lambda \psi(s)\right)ds.
 \]
 This reduces the problem to the numerical computation of a standard integral, but with an integrand containing the square of the function $\psi$. 
 
 However, numerical experiments (not reproduced here) showed that a direct approach is much faster, especially when the ratio  $\nu/\lambda$ is large. This led us to conclude that a standard Euler discretization of the integral would be more satisfactory. Consequently, we have
\begin{equation}\label{eq:Ibar1-alpha1}
 \forall\, k\!\in \{0,\ldots,n\},\; I_{1-\alpha}(\psi) (t_k) \simeq \bar   I^n_{1-\alpha}(\psi) (t_k) =  \frac{1}{\Gamma(2-\alpha)} \left(\frac Tn\right)^{1-\alpha} \sum_{\ell=1}^{k-1}c^{(1-\alpha)}_{k-\ell-1}\bar \psi^n(t_{\ell})\
  \end{equation}
  where $c^{(1-\alpha)}_{0}=1 $ and $ c^{(1-\alpha)}_{\ell}= (\ell+1)^{1-\alpha}-\ell^{1-\alpha}$, $\ell=1:n$.
%
%
%
 \end{itemize}
\begin{remark}
The rate of convergence of the Euler scheme of the fractional Riccati equation with our quadratic right-hand side is not a  consequence of standard theorems on ODEs, even in the regular setting $\alpha=1$, since the standard Lipschitz condition  is not satisfied by the polynomial function $u\mapsto \la u^2+\mu u +\nu$. 
\end{remark}

 \subsection{Step~3: extrapolated Hybrid method} \label{sec:step3}
Let $\Psi(t)= \big(\bar \psi(t), I_1(\psi), I_{1-\alpha}(\psi)(t)\big)$ and 
 $$
 \bar \Psi^n(t_k) = \left(\bar \psi^n, \bar  I^n_1(\psi),\bar I^n_{1-\alpha}(\psi)\right)(t_k),\;k=0:n,
 $$ 
 with an obvious (abuse of) notation.   Numerical experiments~--~not yet confirmed by a theoretical analysis~--~strongly  suggest (see the Example \ref{sec:example} below) that the first component of the vector $\Psi$, {\it i.e.} the solution to the Riccati equation itself, satisfies
 \begin{equation}\label{eq:DevtErreur}
 \bar \psi^n(T)-\psi(T) = \frac{c_1}{n}+o(n^{-1}).
 \end{equation}
 
 Taking advantage of this error expansion~~\eqref{eq:DevtErreur}, one considers, for $n$ even,    the approximator~--~known as  Richardson-Romberg (RR) extrapolation~--~ defined by
\[
 \bar \psi_{_{RR,2}}^n(T) := 2\, \bar \psi^n(T) -\bar \psi^{n/2}(T),
\] 
which satisfies
\begin{align*}
 \bar \psi_{_{RR,2}}^n(T) -\psi(T)&= 2\,\Big( \bar \psi^n(T)-\psi(T)\Big) -(\bar \psi^{n/2}(T)-\psi(T)\Big) \\
& =  2\,\Big( \frac{c_1}{n}+o(n^{-1})\Big) - \Big( \frac{2\,c_1}{n}+o(n^{-1}) \Big) =o (n^{-1}).
\end{align*}

We analogously perform the same extrapolation with the two other components $ \bar  I^n_{1-\alpha}(\psi) (T)$ and $ \bar  I^n_{1}(\psi)(T) $ of $\Psi(T)$ and we may reasonably guess that 
\[
 \bar \Psi_{_{RR,2}}^n(T) -\Psi(T)= o(n^{-1}) \quad \mbox{ where }\quad  \bar \Psi_{_{RR,2}}^n(T) := 2\, \bar \Psi^n(T) -\bar \Psi^{n/2}(T) . 
\]

Note that if $o(n^{-1})= O(n^{-2})$, then $ \bar \Psi_{_{RR,2}}^n(T) -\Psi(T)= O(n^{-2})$,  which dramatically   reduces the complexity and makes the scheme rate of decay (inverse-)linear in the complexity. 

 \subsection{Step~4: multistep extrapolated Hybrid method} \label{sec:step4}
 Here we make the additional assumption that the following second-order expansion holds on the triplet\\
  \[
 \bar \Psi^n(T)-\Psi(T) = \frac{c_1}{n}+\frac{c_2}{n^2}+ o(n^{-2}).
 \]
 We define the weights $(w_i)_{i=1,2,3}$ by  (for a reference on this multistep extrapolation we may cite \cite{Gilles07})
 \[
  w_1 = \frac 13,\quad w_2= -2,\quad w_3= \frac 83.
 \]
and taking $n$ as a multiple of $4$ and set
\[
n_1= \frac n4,\quad n_2= \frac n2,\quad n_3 = n.
\]
So, we define the multistep extrapolation
  \begin{equation}\label{eq:4.32plus1}
\bar \Psi_{_{RR,3}}^n(T) = \frac 13\, \bar \Psi^{n/4}(T)  -2\,\bar \Psi^{n/2}(T) +\frac 83\, \bar \Psi^{n}(T).
 \end{equation}
An easy computation shows that  $\bar \Psi_{_{RR,3}}^n(T) $ satisfies
\[
\bar \Psi_{_{RR,3}}^n(T) -\Psi(T)= o(n^{-2}). 
\]

\subsection{Example} \label{sec:example}

$\rhd$ {{\em Testing the convergence rate.}} Here we test on an example whether our guess on the rate of convergence is true. 
We take the fractional Riccati equation~~\eqref{eq:Riccatialpha} with $\alpha=0.64$ ($i.e.$ Hurst coefficient $H=0.14$), with null initial condition $ I_{1 - \alpha} \psi(0)=0$ and  ({\em real valued}) parameters (these parameters are in line with the one in Section \ref{sec:numerics}, which were calibrated in \cite{Euch2017} using a real data set; here we humped the parameter $\lambda$ in order to gain convexity in the quadratic term of the Riccati, which is more challenging from the numerical point of view. 
)
\[
\lambda= 0.045,\quad\mu= -64.938,\quad \nu= 44\,850.
\]
We focus on short maturities, supposed to be numerically more demanding, and we set $T= 1/252$ (corresponding to one trading day).
Numerically speaking, one may proceed as follows (with the notations introduced for the Richardson-Romberg extrapolation): if the rate~~\eqref{eq:DevtErreur} is  true, it becomes clear that the sequence 
\begin{equation}\label{eq:ErrExpTest}
n\longmapsto \bar c^n_1 = 2n\big( \bar \psi^{n}(T) -\bar \psi^{2n}(T)\big) 
\end{equation}
converges to $ c_1$  as $n\to +\infty$.

 In Table~\ref{table:cn}  we display the values of the constant coefficient $c_1=\bar c^{n}_1$ appearing in~~\eqref{eq:DevtErreur} for different values of $n$ ranging in $[8,131 072]$.
\begin{table} 
\begin{tabular}{|c c|c c|c c|} 
\hline$[n:2n]$& $\bar c^n_1$ & $[n:2n]$& $\bar c^n_1$& $[n:2n]$& $\bar c^n_1$
  \\ \hline\hline
 [8--16]& 123.8478 & [256--512]& 103.8532 &    [8\,192--16\,384] &101.1105
    \\ \hline
  [16--32]& 118.0696 &  [512--1\,024]& 102.9883&    [ 16\,384--32\,768]  &100.9268
    \\ \hline
  [32--64]& 113.9827 &   [1\,024--2\,048]&102.5672&  [32\,768--65\,536 ]&     100.8097
    \\ \hline
  [64--128]& 108.7523 &     [2\,048--4\,096]& 101.8524&  [65\,536--131\, 072 ]&100.6396 
    \\ \hline
  [128--256]&104.8304& [4\,096--8\,192]&101.3989   & [131\,072--262144 ] & 100.5652
  \\ \hline
\end{tabular} \caption{\em Values of $\bar c_1^n$ in Formula~~\eqref{eq:DevtErreur} for $n$ ranging from $8$ to $2^{17}=131 072$. The last value, obtained for $n=n_{\max} = 131\,072$, is taken as reference value.}
\label{table:cn} 
\end{table}
 We take $c_1\simeq \bar c^{n_{\max}}_1=100.5652=c_1^{ref}$ in Formula~~\eqref{eq:DevtErreur} as a reference value, obtained with an accuracy level $\varepsilon_0=0.005$ and $n_{\max} = 2^{17} = 131\,072$.
%
%
 \normalsize
%
%
%
%
%
%

\noindent  
\begin{figure}[h!] 
\begin{tabular}{cc}
\hskip -1.5cm  \includegraphics[width=10cm]{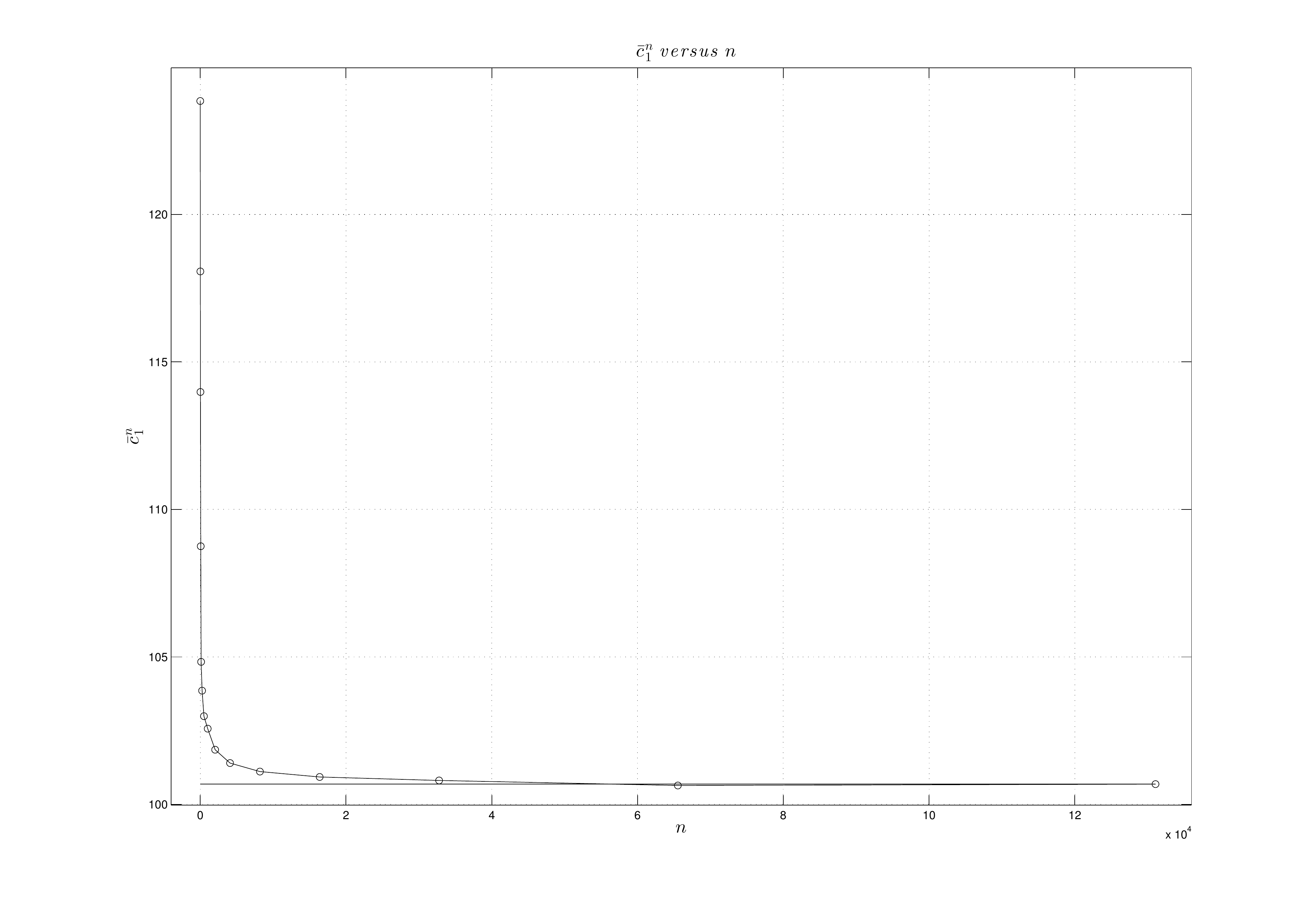}&\hskip -1.25cm \includegraphics[width=10cm, height =6.75cm]{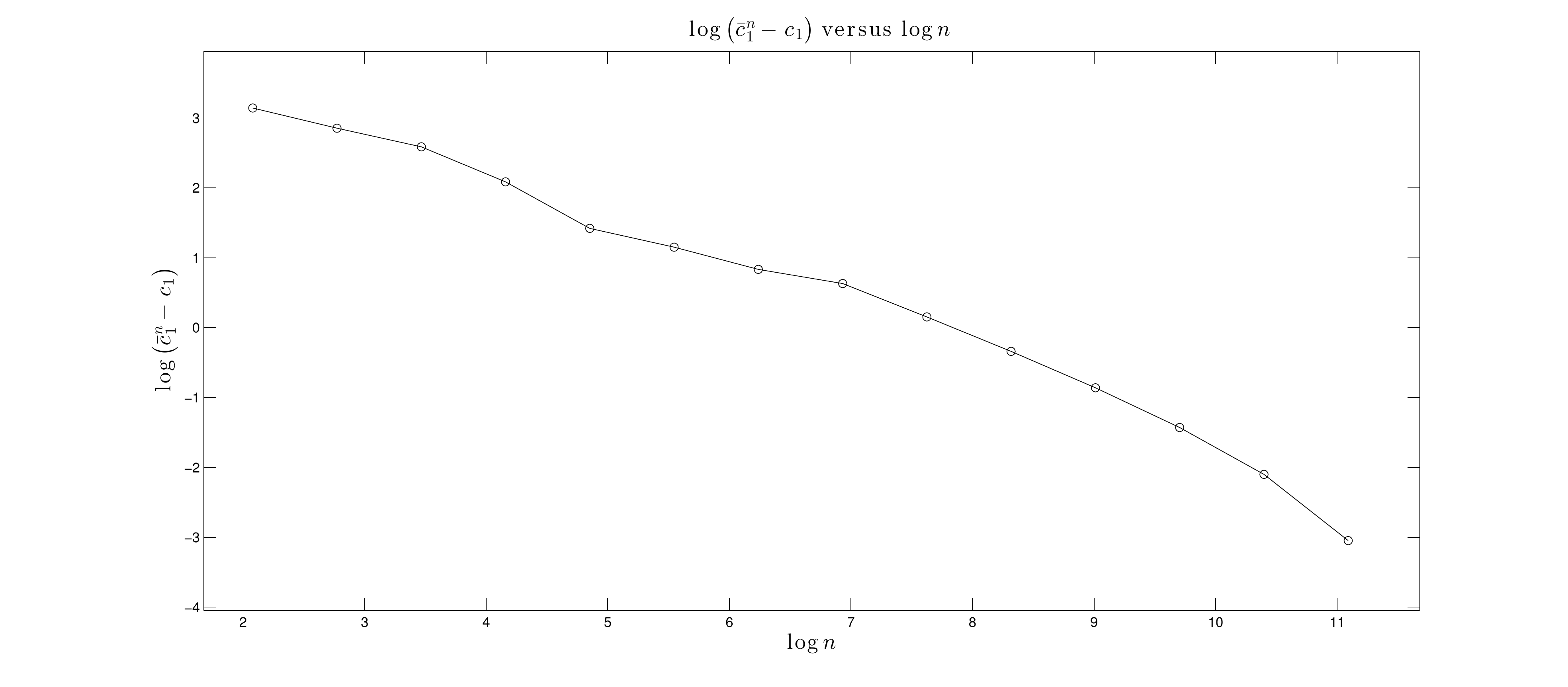}
 \end{tabular}
 \caption{\em Empirical illustration/confirmation of the first order expansion. Left: plot $\bar c_1^n$ versus $n$. Right:  $\log$-$\log$ plot of $\bar c_1^n-c_1^{ref}$ versus $n$.}
\label{fig:ErrExpOrd1}
 \end{figure}

Figure~\ref{fig:ErrExpOrd1} (left hand side) strongly supports the existence of a first order expansion, whereas Figure~\ref{fig:ErrExpOrd1} (right hand side)  is quite consistent with the existence of a second order expansion. 
Unfortunately,  Figure~\ref{fig:ErrExpOrd1} (right hand side) suggests that the higher order expansion does exist but rather of the form 
\[
 \bar \Psi^n(T)-\Psi(T) = \frac{c_1}{n}+\frac{c_2}{n^{2-\beta}}+ o\big(n^{\beta -2}\big)\quad \mbox{with } \beta \!\in (0,1)
\]
where $\beta$  seems to depend on the value of the parameters $\lambda$, $\mu$ and $\nu$.
When we consider the regression coefficient  in the $\log$-$\log$ plot of $\bar c_1^n-c_1^{ref}$ versus $n$, we find  a slope of $-0.5999\simeq -0.6$. Hence, we do not find an expansion of the form $c_2n^{-2}+o(n^{-2})$, since $\log(\bar c_1^n- \bar c^{n_{\max}}_1)\simeq -0.6. \log n + b$ which suggests a second term $c_2n^{-1.52}+ o(n^{-1.52})$. The numerical test seems to suggest that the exponent of this second term of the expansion varies as $n$ increases. In order to avoid the calibration of this additional parameter, we set $\beta = 0$, which  numerically yields by far the most stable and accurate results (see also  Section~\ref{sec:numerics}).
Hence, in all our numerical tests we use the ``regular'' extrapolation formula of order three \eqref{eq:4.32plus1}. 

$\rhd$ {{\em Testing the efficiency of the Richardson-Romberg meta-schemes.}  Let us now turn our attention to the convergence of the hybrid scheme. To evaluate  its efficiency we proceed as follows: we 
%
artificially introduce the hybrid scheme by setting
\[
k_0 =  \Big\lfloor \frac{ 0.5\times n\vartheta \widehat  R_{\psi}}{T}\Big\rfloor
\]
which differs from the original $k_0$ by the $0.5$ factor. As a consequence, the series expansion is only used approximately between $0$ and $0.5.\vartheta \widehat R_{\psi}$ and the time discretization scheme is used between $t_{k_0}$ and $T$. This artificial switch is applied to each of the three scales $T/n$, $T/(2n)$ and $T/(4n)$ of  the extrapolated meta-scheme implemented.

As a benchmark for the triplet we use the value obtained via the fractional power series expansion with $r_0= 200$, namely 
\begin{align*}
\Psi(T)&=  \Big(\psi(T), I_1(\psi)(T), I_{1-\alpha}(\psi)(T)\Big)\\
 & = (165.7590,   21.2394,    0.4409).
\end{align*}

In the numerical test reproduced in Figure~\ref{fig:ErrRR2et3} below,  the convergence of both the $RR2$ and the ``regular'' (namely, associated to error expansion in the scale $n^{-k}$, $k=0,1,2$) $RR3$ Richardson-Romberg meta-schemes (which were introduced, respectively, in Sections~\ref{sec:step3} and~\ref{sec:step4}) is tested, by plotting $\bar \Psi_{_{RR,2}}^n$ and $\bar \Psi_{_{RR,3}}^n$ as functions of $n$ and of the computational time.    
 

Although we could not exhibit through numerical experiments the existence of a third order expansion of the error at rate $c_2 n^{-2}$~--~corresponding to $\beta =0$~--~as mentioned above, it turns out that the weights  resulting from this value of $\beta$, i.e. the ``regular'' extrapolation formula~~\eqref{eq:4.32plus1} in the  third order Richardson-Romberg extrapolation (RR3) yields by far the most stable and accurate results (see also  Section~\ref{sec:numerics}).
\vspace{-1.5 cm}
\noindent  
\begin{figure}[h!] 
\begin{tabular}{cc}
\includegraphics[width=8cm, height=9cm]{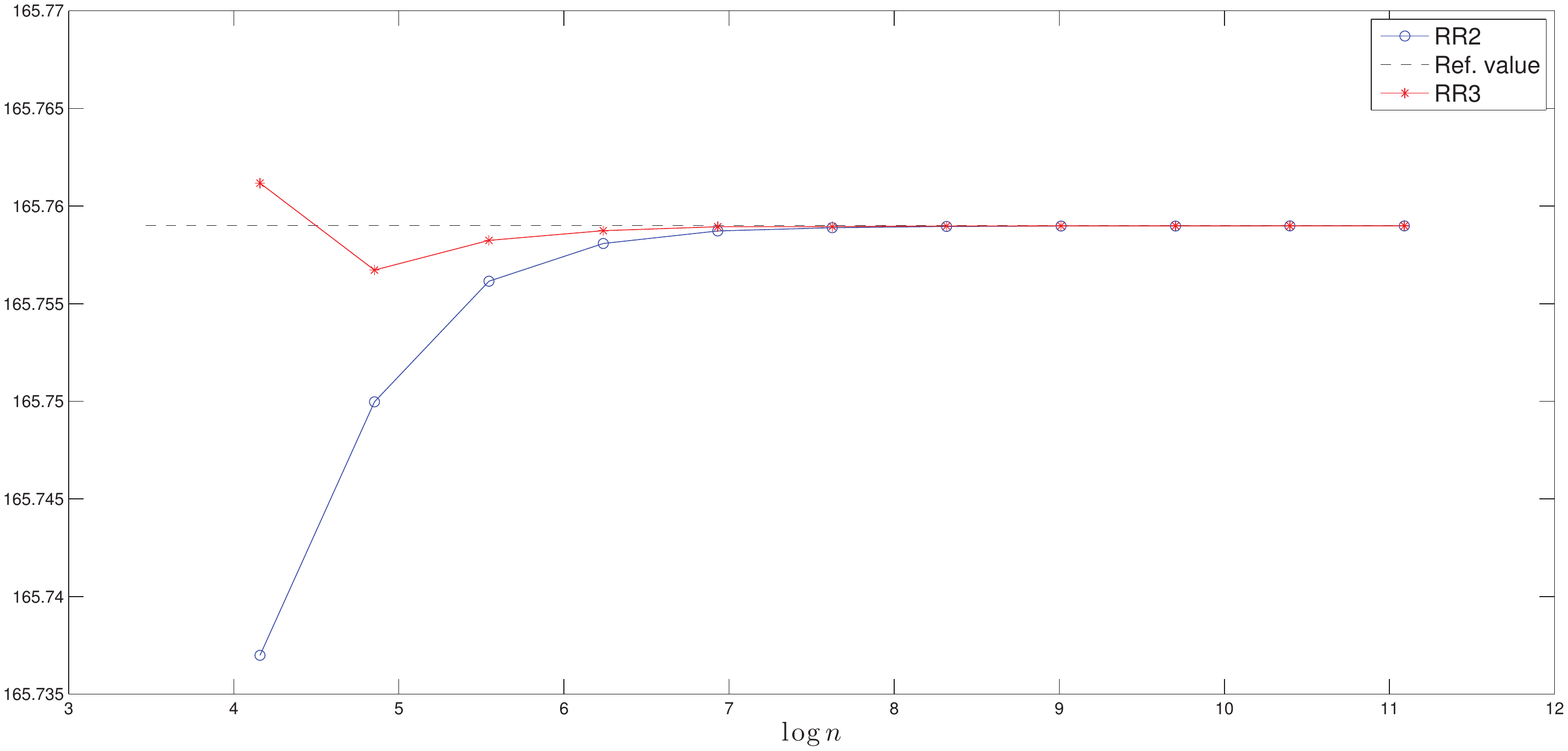}&\hskip -0.25cm \includegraphics[width=8cm, height=9cm]{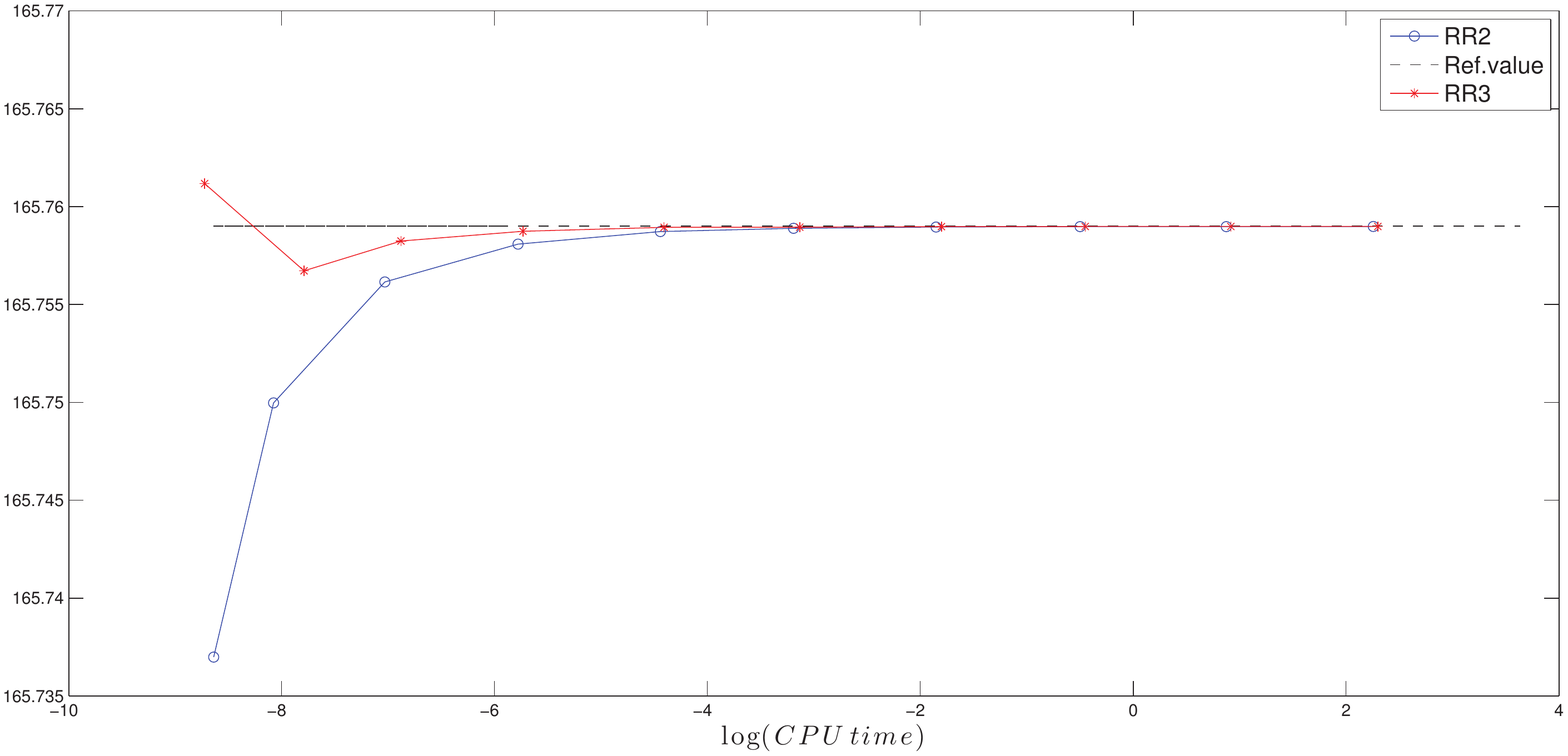}
 \end{tabular}
 \caption{\em Function $\psi$: $RR2$ (--$o$--) and $RR3$ (--$*$--) meta-schemes versus $\log n$ (left) and $\log(CPU\, time)$ (right), $n=2^5,\ldots,2^{16}$.}
\label{fig:ErrRR2et3}
 \end{figure}

\begin{remark} 
When $t<\vartheta \widehat R_{\psi}$, the computation is performed exclusively via the series expansion and is  extremely faster than that involving the meta-schemes.
\end{remark}

We end this subsection by highlighting that in all numerical experiments that follow we will adopt, in our hybrid algorithm, the ``regular" third order extrapolation meta-scheme \eqref{eq:4.32plus1}.

\medskip
 We provide in the short section that follows further technical specifications of the hybrid scheme.

 \subsection[Practitioner's Corner]{Practitioner's Corner}\label{practitioner}

\begin{itemize}
\item[$\rhd$]  {\em Complexity reduction}.  To significantly reduce the complexity of the computations, one may note that if $\phi $ is solution to 
$ (\mathcal E^0_{1,\,\mu,\nu/\lambda})$ then $ \psi = \frac{\phi}{\lambda}$ is solution to $ (\mathcal E^0_{\lambda,\mu,\nu})$. Solving directly $(\mathcal E^0_{1,\,\mu,\,\nu/\lambda})$ allows to cancel all multiplications by $\lambda$ throughout  the numerical scheme, at the price of  a unique division by $\lambda$ at the end.
\item[$\rhd$] {\em Calibration of $\vartheta$}.  Numerical experiments carried out with $T = 1/252$ (one trading day) suggest that, at least for
small values of $t$, the optimal threshold $\vartheta = \vartheta(h)$ is a function of the time discretization step $h$.
The coarser $h$ is,  the lower $\vartheta$ should be to minimize the execution time. It seems that for $h = T/16$,
$\vartheta(h) \simeq 0.5$ whereas for $h = T/4096$, $\vartheta(h) \simeq 0.925$. This leads us to set, when $h=\frac Tn$, 
\[
\vartheta(h) = \min\left(0.65 +0.3\left(\frac{n - 32}{4064}\right)^{0.25},0.925\right), \quad 32\le n \le 4096.
\]
\end{itemize}


%
%
%

\section{Numerical performance in the rough Heston model}\label{sec:numerics}

In this section we test the performance of our results in solving the homogeneous fractional Riccati equation $(\mathcal E^0_{\lambda,\mu,\nu})$ in  ~\eqref{eq:Riccatialpha} that we recall for the reader's convenience:
$$
D^\alpha \psi(t) = \lambda \psi^2(t) +  \mu \psi(t)  + \nu , \quad I_{1-\alpha}\psi(0)=0.
$$
We apply our methodology to the fractional Riccati equation arising in the Rough Heston pricing model considered in~\cite{Euch2017}. 

We test the series approximation and the hybrid procedure we introduced in two steps: first  we consider the power series approximation to the Riccati solution, which reveals to be extremely fast. Then we consider the hybrid method, i.e., the series combined with the Richardson-Romberg extrapolation method, in order to allow for horizons beyond the convergence radius of the power series representation. Remarkably, we find that also the hybrid method is very fast and stable when compared with the only competitor in the literature, represented by the Adams method. 
All the test have been obtained in C++ using a standard laptop endowed with a 3.4GHz processor.

\subsection{Testing the Fractional Power  Series Approximation}\label{sec:testFrac}

Let us consider the fractional power series expansion representation ~\eqref{eq:DSEalpha} $ \psi_{\lambda,\mu,\nu}(t)= \sum_{k\geq0}a_k t^{k\alpha}$.

We would like to use the calibrated parameters in  \cite{Euch2017}, namely we would like to work in the following setting (for clarity, we add a subscript $R$ in the parameters below)
\begin{equation}\label{eq:frHestonR}
\left\{
\begin{array}{rcl}
d S_t & = & S_t \sqrt{V_t} dW_t, \quad S_0=s_0 \in \mathbb R_+ \\
V_t &=& V_0 + \frac{1}{\Gamma (\alpha)}  \int_0^t (t-s)^{\alpha-1} \gamma_R (\theta_R - V_s) ds + \frac{1}{\Gamma(\alpha)}  \int_0^t (t-s)^{\alpha-1} \gamma_R \nu_R \sqrt{ V_s} dB_s , V_0 \in \mathbb R_+,
\end{array}
\right.
\end{equation}
where $\gamma_R, \theta_R, \nu_R$ are positive real numbers. We take the following values
\begin{equation}\label{eq:param}
\alpha=0.62 \quad \gamma_R = 0.1 \quad \rho = -0.681 \quad V_0=0.0392 \quad \nu_R=0.331 \quad \theta_R=0.3156,
\end{equation}
where $\rho $ denotes the correlation between the two Brownian motions $W$ and $B$.


The  fractional Riccati equation to be solved in the setting of \cite{Euch2017} is:
\begin{equation}\label{fracRiccatiR}
D^\alpha \psi(t) = \frac{1}{2} (u_1^2 - u_1) + \gamma_R (u_1 \rho \nu_R - 1) \psi(t) + \frac{{(\gamma_R \nu_R)}^2}{2} \psi^2(t), \quad I_{1-\alpha}\psi(0)=0,
\end{equation}
so that the correspondence for the fractional Riccati coefficients is as follows:
\medskip

\begin{center}
	\begin{tabular}{|c c c| c |}
		\hline
		Riccati ~\eqref{eq:Riccatialpha} & \cite{Euch2017} & Our eq. ~\eqref{fracRiccati0} & Value\\
		\hline
		& & & \\
		$\lambda$ & $\displaystyle \frac{{(\gamma_R \nu_R)}^2}{2}$ & $\displaystyle \frac{{(\eta \zeta)}^2}{2}$ & $0.000547805$ \\
		& & & \\
		$\mu$ & $\gamma_R (u_1 \rho \nu_R - 1)$ & $\eta (u_1 \rho \zeta - 1)$ & $0.1(- u_1 \ 0.225411 - 1)$ \\
		& & & \\
		$\nu$ & $\frac{1}{2} (u_1^2 - u_1)$ & $\frac{1}{2} (u_1^2 - u_1)$ & $\frac{1}{2} (u_1^2 - u_1)$ \\
		& & & \\
		\hline
	\end{tabular}
\end{center}
\medskip

Of course, the parameters for the fractional Riccati  will depend on the frequency $u_1 \in \mathbb C$ of the Fourier-Laplace transform. 
To provide more insight into the convergence domain of the series solution studied in Section \ref{subsec:convR}, we show in Table \ref{ValuesTau*} the value $\tau_{*}$ introduced in Equation \eqref{eq:lowerbound} for the specified parameters and for different values of $u_1$ (recall that this bound is conservative, in that it is not optimal).
\begin{center}
	\begin{table}[h!]
\begin{center}
		\begin{tabular}{|c| c c c c c c |}
			\hline
			$u_1$ &$0.5$ & $5$  & $10$  & $50$ & $100$ & $500$ \\
			\hline
			$\tau_*$ & $21.0481$  & $5.6586$  & $2.3846$ & $0.2201$ & $0.0739$ & $0.0056$ \\
			\hline
		\end{tabular}  \caption{Values of the convergence radius $\tau_*$ for different values of $u_1$.} \label{ValuesTau*} 
\end{center}
	\end{table} 
\end{center}

In order to give an idea of the computational time required by the fractional power series solution, we set $ \Re e (u_1)=100$ (in the pricing procedure we will  consider several values for $ \Re e (u_1)$, as we shall integrate over this parameter in order to compute the inverse Fourier transform) and we fix the dampening factor $\Im m(u_1) = -2.1 $, in line with the Fourier approach of \cite{article_Carr99}. 
Finally, we  focus on short term maturities (namely $T\leq 1$ month), in order to test the pure series expansion. 

In Figure
~\ref{fig:compTimePsi} we plot the computational time required to obtain the solution  $\psi_{\lambda,\mu,\nu}(T)$ when $T \in \{ 1 \ day, 2 \ days, 3 \ days, 4 \ days, 1 \ week, 2 \ weeks, 3 \ weeks, 1 \ month \}$. We stress the fact that this time is expressed in microseconds, i.e., in $10^{-6}$ seconds. The corresponding convergence radius is equal to $0.198036$, which is beyond $T$.  Namely, with the notation of Section~\ref{sec:hybrid}, $T < \vartheta R^{r_{\max}}_{\psi}$ and so $\psi_{\lambda,\mu,\nu}(T)$  is approximated via the fractional power series ~\eqref{eq:Psibis} truncated into sums from $r=1$ to $r=250$. 
Figure~\ref{fig:compTimePsi} confirms that the power series representation is extremely fast and the computational time is basically constant with respect to small maturities $T$.  

\begin{figure}[h!] 
		\includegraphics[scale=0.5]{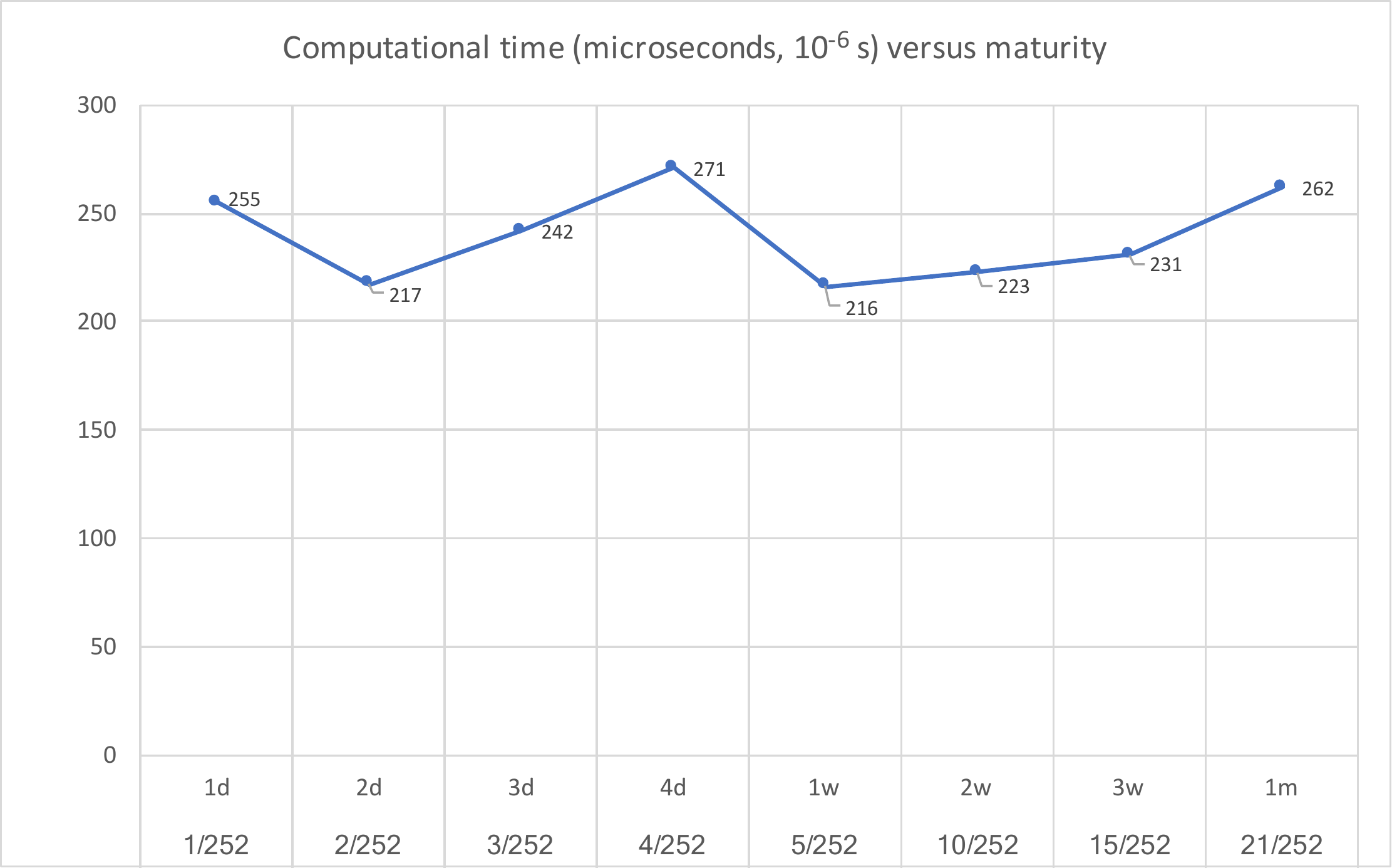} 
	\caption{\em Computational time, in microseconds, required to obtain $\psi_{\lambda,\mu,\nu}(T)$ varying as a function of $T$. Here $\lambda = 0.000547805; \mu=0.1(- u_1 \ 0.225411 - 1); \nu=\frac{1}{2} (u_1^2 - u_1)$ with  $\Re e (u_1)=100, \Im m(u_1) = -2.1$. }
\label{fig:compTimePsi}
\end{figure}

\subsection{The Hybrid Algorithm at Work}

We now test the performance of our hybrid method presented in Section~\ref{sec:hybrid} when applied to the option pricing problem.   We consider a book of European Call options on $S$, in a Heston-like stochastic volatility framework, where the volatility exhibits a rough behaviour (the so-called fractional Heston model) as in \cite{Euch2017}, with maturities ranging from $1$ day till $2$ years and strikes in the interval $[80\% ; 120\%]$ in the moneyness.  
The model parameters are those specified in Section~\ref{sec:testFrac}, i.e., the same parameters calibrated by \cite{Euch2017}.
In order to obtain the prices we use the Carr-Madan inversion technique, namely the one presented in \cite[Section 3]{article_Carr99}, with dampening factor $\alpha_{CM} = 1.1$.  In the numerical inversion of the Fourier transform in Equation ~\eqref{charF}, we integrate with respect to the real part of the frequency parameter  $u_1$. Here, we consider $u_1$ varying as follows (the choice of $250$ turns out to be a good tradeoff between the stability of the results and the computational time)
$$
 \Re e (u_1) \in \{ 0.1 , 0.2, \dots,249.9, 250 \} \quad \textrm{and} \quad  \Im m(u_1) = -2.1.
$$

That is, we compute $2500$  times a Fourier transform.
Notice that, in general, the maturities of the options can go beyond the convergence radius of the power series representation. In other words, when computing the triplet 
\[
\Psi(T)=  \Big(\psi(T), I_1(\psi)(T), I_{1-\alpha}(\psi)(T)\Big)
\]
one can be forced to switch to the hybrid method in order to get the solution of the fractional Riccati, since $T > \vartheta R^{r_{\max}}_{\psi}$. In this case, we set $n=128$ (recall Section~\ref{sec:hybrid}, Step 2, Phase I), which turns out to be a value leading to stable results, in the sense that for larger values of $n$ the prices do not change.


{\small
\begin{table}[h!]
	\begin{center}
			\rotatebox{90}{
			\begin{tabular}{|c|c|c|c|c|c|c|c|}
			\hline
	\multicolumn{1}{|c|}{} &		\multicolumn{1}{|c|}{} & \multicolumn{1}{|c|}{$1$ day } & \multicolumn{1}{|c|}{$1$ week} & \multicolumn{1}{|c|}{$1$ month } & \multicolumn{1}{|c|}{$6$ months} & \multicolumn{1}{|c|}{$1$ year} & \multicolumn{1}{|c|}{$2$ years} \\
			\hline
	\multicolumn{1}{|c|}{Strike} &		\multicolumn{1}{|c|}{} & \multicolumn{1}{|c|}{price(CT)} & \multicolumn{1}{|c|}{price(CT)} & \multicolumn{1}{|c|}{price(CT)} & \multicolumn{1}{|c|}{price(CT)} & \multicolumn{1}{|c|}{price(CT)} & \multicolumn{1}{|c|}{price(CT)} \\
			\hline
\multicolumn{1}{|c|}{$80\%$} &			\multicolumn{1}{|l|}{Hybrid}&$ 20(180) $&    $ 20(154)                              $&$ 20.0005(410) $&$ 20.6112(672) $&$ 22.1366(553) $&$ 25.4301(667) $\\
\multicolumn{1}{|c|}{} &			\multicolumn{1}{|l|}{Adams}&$ 19.9988(108)**   $&$  20(108)                        $&$  20.0005(108) $&$  20.6095(16095)$&$ 22.1331(37768)$&$  25.4258(244381) $\\
			\hline
\multicolumn{1}{|c|}{$85\%$} &			\multicolumn{1}{|l|}{Hybrid}& $ 15(170) $&$ 15(164)                        $&$ 15.0108(421) $&$ 16.2807(689) $&$ 18.3529(553) $&$ 22.2091(596) $\\
\multicolumn{1}{|c|}{} &			\multicolumn{1}{|l|}{Adams}&$  15(107)           $&$ 15(108)                         $&$  15.0108(9282) $&$  16.2783(16658) $&$  18.3486(39557)*$&$ 22.2044(244056) $\\
			\hline
\multicolumn{1}{|c|}{$90\%$} &			\multicolumn{1}{|l|}{Hybrid}& $  10(155) $&$ 10.0002(155)              $&$ 10.1144(423) $&$ 12.3948(671) $&$ 14.9672(549) $&$ 19.2898(594) $\\
\multicolumn{1}{|c|}{} &			\multicolumn{1}{|l|}{Adams}&$  9.9985(109)**    $&$ 10.0002(109)                $&$  10.1141(9598) $&$  12.3924(38252) $&$ 14.9623(38236)$&$  19.2847(248059) $\\
			\hline
\multicolumn{1}{|c|}{$95\%$} &			\multicolumn{1}{|l|}{Hybrid}&$5.0003(156) $&$5.0491(156)            $&$ 5.6723(410) $&$ 9.0636(676) $&$ 12.0059(557) $&$ 16.6676(594) $\\
\multicolumn{1}{|c|}{} &			\multicolumn{1}{|l|}{Adams}&$  4.9967(112)**      $&$ 5.0489(2178)             $&$  5.6712(2359) $&$  9.0609(37826) $&$  12.0006(38672)$&$  16.6622(248779) $\\
			\hline
\multicolumn{1}{|c|}{$100\%$} &			\multicolumn{1}{|l|}{Hybrid}& $ 0.5012(154)$&$1.1347(156)         $&$ 2.3896(416) $&$6.3497(672)  $&$ 9.4737(548) $&$ 14.3319(596) $\\
\multicolumn{1}{|c|}{} &			\multicolumn{1}{|l|}{Adams}&$ 0.5071(108)* $&$  1.1339(108)                    $&$  2.3885(112) $&$  6.3461(16076) $&$ 9.4683(38909)$&$  14.3264(249641) $\\
			\hline
\multicolumn{1}{|c|}{$105\%$} &			\multicolumn{1}{|l|}{Hybrid}& $6.39$E-05$(159)$&$ 0.04113(125) $&$ 0.6809(416) $&$4.2550(666)  $&$ 7.3563(551) $&$ 12.2676(601) $\\
\multicolumn{1}{|c|}{} &			\multicolumn{1}{|l|}{Adams}&$ 6.33$E-04$(121)* $&$  0.04118 (108)           $&$  0.6804(2141) $&$  4.2516(16539)$&$  7.3510(38256)$&$ 12.2621(258262)$\\
			\hline
\multicolumn{1}{|c|}{$110\%$} &			\multicolumn{1}{|l|}{Hybrid}& $ 2.37$E-05$(163) $&$ 9.22$E-05$(125) $&$ 0.1205(414) $&$2.7251(681)  $&$ 5.6234(562) $&$ 10.4562(599) $\\
\multicolumn{1}{|c|}{} &			\multicolumn{1}{|l|}{Adams}&$ 2.06$E-03$(109)* $&$  9.28$E-05$(71309)*       $&$  0.01205(2259) $&$ 2.7223(16670) $&$  5.6195(73084)$&$ 10.4508(248696) $\\
			\hline
\multicolumn{1}{|c|}{$115\%$} &			\multicolumn{1}{|l|}{Hybrid}& $ 1.51$E-05$(155) $&$ 6.82$E-09$(155) $&$ 0.0124(410) $&$1.6680(683)  $&$ 4.2343(582) $&$ 8.8773(593) $\\
\multicolumn{1}{|c|}{} &			\multicolumn{1}{|l|}{Adams}&$  1.49$E-03$(107)* $&$ 8.85$E-07$(70685)*       $&$ 0.0125(2194) $&$  1.6658(16302) $&$  4.2307(73086)$&$ 8.8720(248391) $\\
			\hline
\multicolumn{1}{|c|}{$120\%$} &			\multicolumn{1}{|l|}{Hybrid}& $ 1.14$E-05$(196) $&$ 1.80$E-13$(156) $&$ 7.32$E-04$(414) $&$0.9761(667)  $&$ 3.1424(577) $&$ 7.5093(593) $\\
\multicolumn{1}{|c|}{} &			\multicolumn{1}{|l|}{Adams}&$ 1.08$E-03$(110)* $&$  3.70$E-7$(108)*              $&$  7.37$E-04$(2316) $&$  0.9674(16269) $&$  3.1393(73184)$&$ 7.5042(245999) $\\
			\hline
		\end{tabular}
}	\end{center}
	\caption{{\small Call option pricing with the Hybrid and Adams methods. The parameters are as in  \cite{Euch2017}. Maturities range from 1 day till 2 years, strikes range between $80\%-120\%$ of the moneyness. The computational time $(CT)$ is in milliseconds (i.e. $10^{-3}$s). For the hybrid method we fix $n=128$, while for the Adams method the discretization step is chosen in order to satisfy $\vert \sigma_{IMP}(hybrid)-\sigma_{IMP}(Adams)\vert \leq 10^{-2}$. When this is not possible for any discretization step, we put $(*)$ besides the values, while $(**)$ are associated to prices that lead to arbitrage opportunities. }}
\label{prezzi}
\end{table}
}

In  Table~\ref{prezzi} we display the prices together with the  computational times (in milliseconds) obtained by our hybrid method for the entire book of options. 
A quick look at the Table~\ref{prezzi} shows that our method is extremely fast. In fact, all prices are computed in less than one second. Moreover, one can easily verify that using a larger value for $n$ (which here is set to $n=128$) does not change the prices. Therefore, our hybrid method is also very stable and can be used as the benchmark. 
\medskip

Now we  compare the performance of our hybrid algorithm  with the (only) other competitor present in the literature, namely the fractional Adams method,  a numerical discretization procedure described e.g. in~\cite{Euch2017} (see their Section 5.1). As for any discretization algorithm, also for the Adams method one should select the discretization step, and according to this choice the corresponding price can be different. Of course, the smaller the time step in the discretization procedure, the longer will take the pricing procedure.  As we consider  our hybrid method as the benchmark, now we look for the discretization step for the Adams method that leads to prices that are close enough to ours, according to a given tolerance. Here, the error  is measured in terms of the difference of the corresponding implied volatilities associated to the prices generated by the two methods. We fix for example a maximal difference  of $1\%$. Notice that this maximal error is very large, as  for the calibration of the classic Heston model one can typically reach an average for the  RESNORM (sum of the squares of the differences) around E-05$=10^{-5}$. 
\medskip

First of all, let us focus our attention on the very short term maturities in Table~\ref{prezzi}, namely 1 day and 1 week. It turns out that, apart from a couple of situations for 1 week,  the Adams method is also very fast. However, we notice that Adams prices can lead to some arbitrage opportunities. In fact, for example, for the maturity of 1 day and strike $80\%$, Adams method leads to a price smaller than the intrinsic value of the Call (recall that the interest rate here is set to be zero). We put $(**)$ in the table when this situation occurs.  Also, one can check that in many cases it is not possible to find the discretization step for the Adams method in order to generate a price within  the tolerance. We put $(*)$ for the cases where the error is greater than the tolerance, regardless the choice of the discretization step (we pushed the discretization till 150 steps without observing any relevant change). The reason for this phenomenon  is quite intuitive: if the maturity is very short, adding  discretization steps in the procedure does not necessarily produce different prices because the process has not enough time to move. On the other hand, our hybrid method takes benefit of the fractional power series expansion that works extremely well mostly for very short maturities. In conclusion, our method is very fast, stable and accurate in the short maturities when compared to the Adams method.
\medskip

Now let us consider the other maturities till 2 years. Here, in order to get prices close to our benchmark, we are forced to choose an {\em ad hoc} discretization grid for the Adams method, including a number of steps ranging from 10 to 150, depending on the particular  maturity and strike. As a consequence, the corresponding  computational time turns out to be much higher than ours, to the point that for 2 years the computation of prices for the  Adams method require about 4 minutes with 150 discretization steps, while our hybrid algorithm still takes less than one second. In conclusion, we can state that our hybrid algorithm dominates the Adams method for all maturities. 

\medskip

We end this subsection by  reproducing the analogue of Figure 5.2 in~\cite{Euch2017}, namely the term structure of the at-the-money skew, that is the derivative of the implied volatility with respect to the log-strike for at-the-money Calls.  
\begin{figure}[h!]
		\includegraphics[scale=0.5]{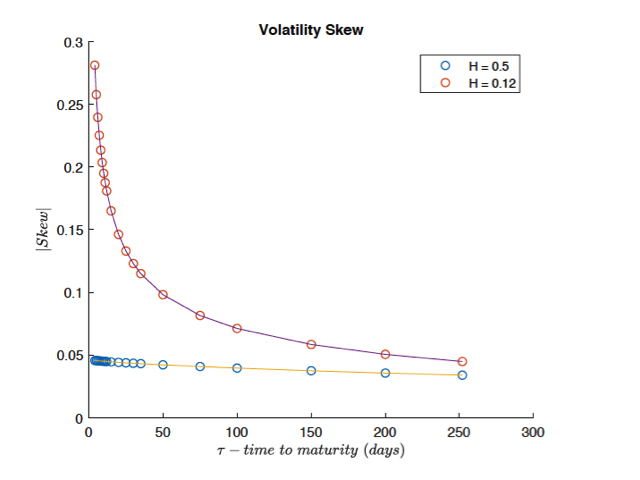} 
	\caption{\em At-the-money skew as a function of the maturity, ranging from 1 day till 1 year, for $\alpha =1$ (corresponding to the classic Heston model with $H=0,5$) and $\alpha =0,62$ (Rough  Heston model with $H=0,12$). The other parameters of the model are as in  \cite{Euch2017}. }
 \label{fig:skew}
\end{figure}
Figure~\ref{fig:skew} confirms the results in~\cite{Euch2017}:  in particular, for the Rough Heston model we see that the skew explodes for short maturities, while it remains quite flat for the classic Heston model, which is well known to be unable to reproduce the slope of the skew for short-maturity options.

\section{The Non-Homogeneous Fractional Riccati Equation}

In this section, we consider the more general case where the fractional Riccati equation has a non zero starting point.
The fractional Riccati ODE arising in finance and associated to the characteristic function of the log-asset price is very special insofar it starts from zero. Recently, \cite{Pulido2017} extended the results of ~\cite{Euch2017} to the case where the volatility is a Volterra process, which includes the (classic and) fractional Heston model for some particular choice of the kernel. Such extension leads to a fractional Riccati ODE but with a general (non-zero) initial condition. This case  is mathematically more challenging and requires additional care. 
Nevertheless, in this section we prove that it is still possible to provide bounds for the convergence domain of the corresponding power series expansion, at the additional cost of extending the  implementation of the algorithm to a  doubly indexed series, in the spirit of ~\cite{jacquier2014}.

For the reader's convenience, first of all we recall the  general Riccati equation (see~~\eqref{eq:Riccatialpha})
\begin{equation*}
(\mathcal E^{u,v}_{\lambda, \mu,\nu})\;\equiv\; D^{\alpha} \psi = \lambda  \psi^2 +\mu \psi+\nu \quad \mbox{and } \left\{\begin{array}{ll}\quad I^{1 - \alpha} \psi(0)=u&\mbox{if }\alpha\!\in (0,1]\hskip 1.5cm\\
I^{1 - \alpha} \psi(0)=u,\;  I^{2 - \alpha}\psi(0)= v &\mbox{ if }\; \alpha\!\in (1,2],\end{array}\right.
\end{equation*}
where $ \lambda$, $\mu$, $\nu$ and $u$, $v$  are {\em complex} numbers. Our aim is to find a solution as a fractional power series as we did in  Section~\ref{sec:Riccati0}  in the case $\alpha\!\in (0,1]$, $u=0$ and $\alpha\!\in (1,2]$, $u=v=0$. This leads us to deal with  the integral form~~\eqref{eq:Riccating2} of this equation. However, the solution will be in this more general setting a doubly index series based on the fractional monomial  functions $t^{\alpha k-\ell}$. Of course, we will again take advantage of the fact that for every $\alpha>0$ and every $r>-\alpha$, 
\[
I_{\alpha}(t^r)= \frac{\Gamma(r+1)}{\Gamma(r+\alpha+1)}t^{r+\alpha}, \; t\ge 0.
\]
 
We now state the result of this section, which represents the main mathematical contribution of the paper. 

\begin{theorem}\label{thm:non-homog} Equation $(\mathcal E^{u,v}_{\lambda, \mu,\nu})$ admits a  solution expandable on a non-trivial interval $[0, R_{\psi})$,  $R_{\psi}>0$, as follows 
\begin{align}\label{eq:defpsi}
\psi(t) =&\sum_{\ell\ge 0}\psi_{\ell}(t) = \sum_{\ell\ge 0} \sum_{k\ge  k(\ell)}a_{k,\ell} t^{\alpha k-\ell}, \quad t\!\in (0, R_{\psi}), 
\end{align}
where the coefficients $a_{k,\ell}\!\in \mathbb{C}$ and, for every $\ell\ge 0$,  $k(\ell)= \min\big\{k\ge 1 : a_{k,\ell} \neq 0 \big\}$ denotes the valuation of $(a_{k,\ell})_{k\ge 1}$. Moreover, the above doubly indexed series is normally convergent on any compact interval of $(0, R_{\psi})$.

\smallskip
\noindent $(a)$ \textbf{Case $\alpha\!\in (\frac 12,1)$:} we have $k(\ell) = (2\ell-1)\vee 1$ if $\nu$, $u\neq 0$, $k(0)=+\infty$ if $\nu=0$ and $k(\ell)=+\infty$ for $\ell\ge 1$ if $u=0$. In particular, one always has $k(\ell) \ge  (2\ell-1)\vee 1$.

The coefficients $a_{k,\ell}$ are recursively defined as follows:  $a_{1,0}= \frac{\nu}{\Gamma(\alpha+1)}$, $a_{1,1}= \frac{u}{\Gamma(\alpha)}$, and, for every $\ell\ge 0$ and every $k\ge k(\ell)\vee 2$,
\begin{equation}\label{eq:akell2}
a_{k,\ell}= \big(\mu a_{k-1,\ell} + \lambda a^{*2}_{k-1,\ell}\big)\frac{\Gamma\big(\alpha(k-1)-\ell+1\big)}{\Gamma(\alpha k-\ell+1)}
\end{equation}
where
\begin{equation}\label{eq:Convol2}
a_{k,\ell}^{*2} = \sum_{\begin{smallmatrix}k_1+k_2= k, \,k_i\geq k(\ell_i)\\ \ell_1+\ell_2=\ell,\, \ell_i\ge 0, i=1, 2\end{smallmatrix}} a_{k_1, \ell_1}a_{k_2,\ell_2}
\end{equation}
Note that  $a_{1,\ell}= 0$, $\ell\ge 2$, and $a_{1,\ell}^{*2}= 0$, $\ell\ge0$.

\smallskip
\noindent $(b)$ \textbf{Case $\alpha\!\in (1,2]$:} we have  $k(\ell)\ge 1+ (\ell-1)\mbox{\bf 1}_{\{\ell\ge 3\}}$ with equality if $\nu$, $u$, $v\neq 0$. 
The coefficients $a_{k,\ell}$ still satisfy~~\eqref{eq:akell2} and~~\eqref{eq:Convol2}  with $a_{1,0}= \frac{\nu}{\Gamma(\alpha+1)}$, $a_{1,1}= \frac{u}{\Gamma(\alpha)}$, $a_{1,2}= \frac{v}{\Gamma(\alpha-1)}$   (and $a_{1,\ell}= 0$, $\ell\ge 3$).
\end{theorem}

The constructive proof of this result is divided into several steps and is provided in full details in Appendix~\ref{sec:thm5.1}. 
A full numerical illustration of the general case is beyond the scope of our paper. Here, we just mention that the computation of the solution through the doubly index fractional power series representation turns out to be still extremely fast as in the previous single-index case. In practice very few $\ell$-layers are needed to compute $\psi(t)$ for a standard accuracy, say $10^{-3}$ or $10^{-4}$. Also, a  hybrid algorithm based on  Richardson-Romberg extrapolation  
can also be  devised in order to allow for maturities longer than the convergence radius of the above double index series. We skip the details for sake of brevity. 

\section[Conclusion]{Conclusion}
In this paper, motivated by recent advances in Mathematical Finance, we solved a family of fractional Riccati differential equations, with constant (and possibly complex) coefficients, whose solution is the main ingredient of the characteristic function of the log-spot price in the fractional Heston stochastic volatility model.   
We first considered the case of a zero initial condition  and we then analyzed the case of a general starting value, which is closely related to the theory of affine Volterra processes.

The solution to the fractional Riccati equation with null initial condition takes the form of power series, whose coefficients satisfy a convolution equation. We showed that this solution has a positive convergence domain, which is typically finite. In order to allow for maturities that are longer than the convergence domain of the fractional power series representation, we provide a hybrid algorithm based on a Richardson-Romberg extrapolation method that reveals to be very powerful. Our theoretical results naturally suggest an efficient numerical procedure to price vanilla options in the Rough Heston model that is quite encouraging in terms of computational performance, when compared with the usual benchmark, represented by the Adams method. 

In the case of a non-null initial condition, the solution takes the form of a double indexed series and, working with additional technical care, we provided error bounds for its convergence domain.

In the last years, Deep Neural Networks(DNN)-based algorithms have been massively adopted in finance in order to solve long-standing problems like robust calibration and hedging of large portfolios, see e.g. \cite{STONE2019}, \cite{HMT2019}, \cite{BS2019} and references therein. Thanks to their flexibility and extremely fast execution time, DNN represent nowadays the standard in the banking industry. However, any deep neural network requires a learning phase, which typically takes a lot of time and needs the pricing technology to feed the network.  
Our methodology, in view of pricing, is useful in order to speed up the learning phase. In this sense, our results should not be seen in competition with DNN since, on the contrary, they are a useful and efficient ingredient to feed the DNN with. So our work is a prerequisite for any DNN based algorithm.

 \begin{appendix}
  \section[Toolbox: Riemann Sums, Convexity  and Kershaw's Inequalities]{Toolbox: Riemann Sums, Convexity  and Kershaw's Inequalities}
  
\subsection{Riemann sums, convexity}
We will extensively  need the following elementary lemma on Riemann sums.

\begin{lemma} \label{lem:f}
Let $f: (0,1)\rightarrow \mathbb{R}_+$ be a function, non increasing on $(0,1/2]$ and symmetric, i.e. such that $f(1-x)=f(x), \,x\!\in (0,1)$, hence  convex. Assume that
$\int_0^1f(u)du<+\infty$. Then,  $\inf_{x\in(0,1)}f(x)= f(\frac 12)$ and:
\begin{align}
\label{fboundeven}
\left(1-\frac 1k \right) f\left(\frac 12\right)
\leq \frac{1}{k}\sum_{\ell =1}^{k-1}f\left(\frac{\ell}{k}\right)\leq\int_0^1\hskip-0.25cmf(u)du -\mbox{\bf 1}_{\{k\, \tiny even\}} \int_{1/2}^{1/2+1/k}\hskip-0.5cm f(u)du-\mbox{\bf 1}_{\{k\,\tiny odd\}}  \int_{(k-1)/2k}^{(k+1)/2k}\hskip-0.5cm f(u)du. 
\end{align}

In particular,  it follows that, for every $k\ge 2$, 
\begin{align*}
\frac 12 f \left(\frac 12\right) \le \frac 1k \sum_{\ell =1}^{k-1}f\left(\frac{\ell}{k}\right)\leq\int_0^1f(u)du\quad \mbox{ and }\quad \lim_{k} \frac{1}{k}\sum_{\ell =1}^{k-1}f\left(\frac{\ell}{k}\right) = \int_0^1f(u)du
\end{align*}
so that 
\begin{equation}\label{eq:upperf}
 \sup_{k\ge 1} \frac{1}{k}\sum_{\ell =1}^{k-1}f\left(\frac{\ell}{k}\right) = \int_0^1f(u)du\quad\mbox{ and }\quad 
 \min_{k\ge 2}\frac 1k \sum_{\ell =1}^{k-1}f\!\left(\frac{\ell}{k}\right)=  \frac 12\, f\!\left(\frac 12\right).
\end{equation}
\end{lemma}

\begin{proof} 
The lower bound is a straightforward consequence of the convexity of  the function $f$ since for every $k\ge 2$, 
\begin{align*}
\frac{1}{k}\sum_{\ell =1}^{k-1}f\left(\frac{\ell}{k}\right)&=\left(1-\frac 1k\right)\frac{1}{k-1}\sum_{\ell =1}^{k-1}f\left(\frac{\ell}{k}\right)\\
&\geq \left(1-\frac 1k\right)f\left(\frac{1}{k-1} \frac{k(k-1)}{2k}\right) =\left(1-\frac 1k\right)f\left(\frac 12\right)
\geq \frac 12f\left(\frac 12\right)
 \end{align*}
with equality if and only if $k=2$. Let us consider now  the upper bounds. 

\medskip
\noindent $\rhd$ {\em Case~$k$ even}. We consider separately the  half sums from $1$ to $\frac{k}{2}$ and from $  \frac{k}{2}+1$ to $k-1$.

For $\ell \in \{ 1,\ldots,\frac{k}{2}\}$, $f\left( \frac{\ell}{k}\right)\leq f(u)$ for $u\in \left( \frac{\ell -1}{k},\frac{\ell}{k}\right)$, while  
for $\ell \in \{ 1,\ldots,\frac{k}{2}-1\}$, $f\left( \frac{\ell}{k}\right)\geq f(u)$ for $u\in \left( \frac{\ell }{k},\frac{\ell+1}{k}\right)$. 
Therefore 
\begin{align*}
\frac{1}{k}\sum_{\ell =1}^{\frac{k}{2}}f\left(\frac{\ell}{k}\right)&\leq \int_0^{1/2}f(u)du. 
 \end{align*}
  
On the other hand, for the second half sum, 
if $\ell \in \{ \frac{k}{2}+1,\ldots,k-1\}$, $f\left( \frac{\ell}{k}\right)\leq f(u)$ for $u\in \left( \frac{\ell }{k},\frac{\ell+1}{k}\right)$, while  
for $\ell \in \{ \frac{k}{2}+1,\ldots,k-1\}$, $f\left( \frac{\ell}{k}\right)\geq f(u)$ for $u\in \left( \frac{\ell-1 }{k},\frac{\ell}{k}\right)$.

Therefore 
\begin{align*}
\frac{1}{k}\sum_{\ell =\frac{k}{2}+1}^{k-1}f\left(\frac{\ell}{k}\right)&\leq \int_{1/2+1/k}^{1}f(u)du. 
 \end{align*}

Summing up both sums  yields ``even part'' of~~\eqref{fboundeven}.

\medskip
\noindent $\rhd$  {\em Case~$k$ odd}. One shows likewise that 
\begin{align*}
\frac{1}{k}\sum_{\ell =1}^{(k-1)/2}f\left(\frac{\ell}{k}\right)&\leq \int_0^{1/2-1/2k}f(u)du 
 \end{align*}
for the first half sum, while 
\begin{align*}
\frac{1}{k}\sum_{\ell =(k+1)/2}^{k-1}f\left(\frac{\ell}{k}\right)&\leq \int_{(k+1)/2k}^1f(u)du 
 \end{align*}
for the second one. Summing up gives the``odd part'' of~~\eqref{fboundeven}.
\end{proof} 

\subsection{\em Kershaw inequalities}
We also rely on these inequalities (see~\cite{KER}) controlling ``ratios of close terms'' of the Gamma function. For $x>0$, and every $s\!\in (0,1)$, 
\begin{align}\label{eq:Kershaw}
\left(x+\frac{s}{2}\right)^{1-s} < \frac{\Gamma (x+1)}{\Gamma (x+s)}< \left(x-\frac{1}{2}+\sqrt{s+\frac{1}{4}}\right)^{1-s}.
 \end{align}
For $s=1$ this double  inequality becomes an  equality.

\section[Proof of Theorem~3.1]{Proof of Theorem~\ref{thm:Radius}}\label{sec:thm3.1}

\subsection{Proof when  $\alpha\!\in (0,1]$}
\subsubsection{Lower Bound for the Radius by Upper-Bound Propagation (Claim~$(a)$)} \label{subsec:upperpropag0}
We first  want to  prove  by induction  the following upper-bound of the coefficients $a_k$, namely  
\begin{align}\label{eq:propaga1}
\forall\, k\geq  1, \quad \vert a_{k}\vert&\leq C k^{\alpha- 1}\rho^k
\end{align}
for some $C$ and $\rho >0$ (note that $ \alpha- 1\!\in (-1, 0]$). We assume the  $|a_1|  \le C\rho$ (with $a_1=  \frac{|\nu}{\Gamma(\alpha+1}$: this condition will be double-checked later) and we want to propagate this inequality by induction. Assume that~~\eqref{eq:propaga1} holds for some  $k\ge 1$. Plugging this bound in~~\eqref{eq:Eqak0} yields
\begin{align}\label{eq:Ineqak}
\vert a_{k+1}\vert&\leq \frac{\Gamma (\alpha k +1)}{\Gamma (\alpha k +\alpha +1 )} \big( \lambda \vert {a_k^*}^2\vert +\vert \mu\vert  \vert a_k\vert\big).
  \end{align}

As ${a_k^*}^2=\sum_{\ell =1}^{k-1}a_\ell a_{k-\ell}$ (see~~\eqref{eq:convoldef1}),  we have
\begin{align}
 \nonumber 
 \vert {a_k^*}^2\vert &\leq C^2 \rho^k \sum_{\ell =1}^{k-1}\ell^{\alpha- 1} (k-\ell)^{\alpha- 1} =  C^2 \rho^k k^{2\alpha- 2} \sum_{\ell =1}^{k-1}\left(\frac{\ell}{k}\right)^{\alpha- 1}\left( 1-\frac{\ell}{k}\right)^{\alpha- 1}.
\end{align}

Applying Inequality~~\eqref{eq:upperf}  from Lemma~\ref{lem:f} to the function $f_{\alpha}$ defined by  $f_{\alpha}(x) = x^{\alpha-1}(1-x)^{\alpha-1}$, $\alpha\!\in(0,1]$, yields for every $k\ge 1$, 

\begin{equation}\label{eq:upperconva_k}
 \vert {a_k^*}^2\vert 
                                     \le C^2 \rho^k k^{2\alpha-1} \int_{0}^{1}u^{\alpha-1} (1-u)^{\alpha-1} du
 = C^2 \rho^k k^{2\alpha-1} B(\alpha,\alpha)
\end{equation}
where $B(a,b)$ denotes the Beta function (note that $a^{*2}_1=0$).

\smallskip From the Kershaw  inequality~~\eqref{eq:Kershawb}, we obtain in particular that, for every $x>0$ and every $s\!\in (0,1)$, 
\begin{align}\label{eq:Kershawb}
 \frac{\Gamma (x+s)}{\Gamma (x+1)}< \left(x+\frac{s}{2}\right)^{s-1}.
 \end{align}

Now set $x=\alpha (k+1)$ and  $s=1-\alpha$. We get
\begin{align*}
 \frac{\Gamma (\alpha k +1)}{\Gamma (\alpha k +\alpha +1 )}&< \left(\alpha (k+1)+\frac{1-\alpha}{2}\right)^{-\alpha} = (k+1)^{-\alpha}\alpha^{-\alpha}\left(1+\frac{1-\alpha}{2\alpha (k+1)}\right)^{-\alpha}< (k+1)^{-\alpha}\alpha^{-\alpha}
                      \end{align*}
since $\left(1+
\frac{1-\alpha}{2\alpha (k+1)}\right)^{-\alpha}<1$. Plugging successively  this inequality and~~\eqref{eq:upperconva_k} into~~\eqref{eq:Ineqak}  yields for every $k\ge 1$, 
\begin{align}
\nonumber \vert a_{k+1}\vert&\leq \frac{\Gamma (\alpha k +1)}{\Gamma (\alpha k +\alpha +1 )} \big( |\lambda| \vert {a_k^*}^2\vert +\vert \mu\vert  \vert a_k\vert  \big) < (k+1)^{-\alpha}\alpha^{-\alpha}
   \big( |\lambda| \vert {a_k^*}^2\vert +\vert \mu\vert  \vert a_k\vert \big)\\
  \nonumber    &\leq  (k+1)^{-\alpha}\alpha^{-\alpha}
   \big( |\lambda| C^2 \rho^k k^{2\alpha-1} B(\alpha,\alpha) +C\vert \mu\vert  \rho^k k^{\alpha-1} \big)\\            
\label{eq:majorfonda}  &\leq  C\alpha^{-\alpha}  (k+1)^{-\alpha}\rho^{k+1} \frac{|\lambda| 
 C k^{2\alpha-1} B(\alpha,\alpha)  +\vert \mu\vert    k^{\alpha-1}}{\rho  }.        
 \end{align}

Keeping in mind that we want to get $|a_{k+1}|\le C(k+1)^{\alpha-1}\rho^{k+1}$, we rearrange the terms as follows
\begin{align}
\nonumber \vert a_{k+1}\vert 
 & \le   C  (k+1)^{\alpha-1}\rho^{k+1}\left(\frac{k}{k+1}\right)^{2\alpha-1} \alpha^{-\alpha}\frac{|\lambda| C B(\alpha,\alpha) + |\mu| k^{-\alpha}}{\rho}\\
\label{eq:akplusun}  &\le C  (k+1)^{\alpha-1}\rho^{k+1} 2^{(1-2\alpha)_+} \alpha^{-\alpha}\frac{|\lambda| C B(\alpha,\alpha) + |\mu|  }{\rho}
 \end{align}
 where we used that 
$ \left(\frac{k}{k+1}\right)^{2\alpha-1}\le 2^{(1-2\alpha)_+}$ and we recall the notation $x_+=\max\{x,0\}$.

\bigskip
 Finally, the propagation of Inequality~~\eqref{eq:propaga1} is satisfied for every $k\ge 1$ by  any couple $(C, \rho)$ satisfying
 \[
|a_1|=  \left|\frac{\nu}{\Gamma(\alpha +1)}\right| \le C\rho\quad \mbox{ and }\quad  2^{(1-2\alpha)_+}\alpha^{-\alpha}\big(|\lambda| C B(\alpha,\alpha) + |\mu|\big) \le \rho.
 \]
It is clear that that, the  lower $\rho$ is,  the higher  our lower bound for the convergence radius of the series will be. Consequently, we need to saturate both inequalities which leads to the system
\[
\rho= \frac{|\nu|}{\Gamma(\alpha +1)C} \quad \mbox{ and }\quad \rho = 2^{(1-2\alpha)_+} \alpha^{-\alpha}\Big( \frac{|\lambda||\nu|}{\Gamma(\alpha +1)\rho}  B(\alpha,\alpha) + |\mu| \Big)
\]
or, equivalently, using both identities $B(\alpha, \alpha)= \frac{\Gamma(\alpha)^2}{\Gamma(2\alpha)}$ and  $\Gamma(\alpha+1)= \alpha\Gamma(\alpha)$, 
\[
C= \frac{|\nu|}{\Gamma(\alpha +1)\rho} \quad \mbox{ and }\quad 2^{- (1-2\alpha)_+} \alpha^{\alpha}\rho^2-|\mu|\rho -\frac{|\lambda|   |\nu| \Gamma(\alpha) }{\alpha\Gamma(2\alpha )}=0.
\]
The positive  solution $\rho_*= \rho_*(\alpha, |\lambda|, \mu, \nu)$ of the above  quadratic equation in $\rho$ is given by 
\begin{equation}\label{eq:rho_*}
\rho_* = \frac{|\mu| + \sqrt{\mu^2 + 2^{2 - (1-2\alpha)_+}\frac{ \alpha^{\alpha-1} \Gamma(\alpha) }{\Gamma(2\alpha )}|\lambda|  |\nu|  }}{2^{1 - (1-2\alpha)_+}\alpha^{\alpha}}.
\end{equation}
Consequently, setting $C_* =  \frac{|\nu|}{\Gamma(\alpha +1)\rho_*}=  \frac{|\nu|}{\alpha\Gamma(\alpha )\rho_*}$, we finally find that  
\begin{equation}\label{eq:UpperBoundak}
\forall\, k\ge 1,\quad |a_k| \le C_*k^{\alpha-1}\rho_*^k
\end{equation}
so that   the convergence radius $R_{\psi} = \liminf_{k}|a_k|^{-\frac{1}{\alpha k}}$ of the function $\psi$  satisfies
\[
R_{\psi_{\lambda,\lambda,\nu}} \ge  \rho_*^{-\frac{1}{\alpha}}=    \frac{2^{\frac{1}{\alpha} -(\frac{1}{\alpha}-2)^+}\alpha}{\Big(|\mu| + \sqrt{\mu^2 + 2^{2-(1-2\alpha)^+} \frac{\alpha^{\alpha-1} \Gamma(\alpha) }{\Gamma(2\alpha )}|\lambda|  |\nu| }\Big)^{\frac{1}{\alpha}} }.
\]

\bigskip
\noindent {\bf Remarks.} $\bullet$ Note that, when $\lambda\neq 0$, one deduces from~~\eqref{eq:rho_*}
\[
\rho_*\ge 2^{(1-2\alpha)_+} \alpha^{-\alpha}\max\left(|\mu|, \alpha^{\frac{\alpha-1}{2}}\,2^{-(\frac 12-\alpha)^+} \Big(\frac{B(\alpha,\alpha) }{\Gamma(\alpha)}\Big)^{\frac 12}|\lambda|  |\nu| \right).
\]

\smallskip
\noindent  $\bullet$ A slight improvement of the theoretical lower bound is possible by imposing the constraints $|a_1|\leq C\rho$ and $|a_2|\leq C\rho^2 2^{\alpha-1}$ and using that $k^{-\alpha}\le 2^{-\alpha}$ when $k\ge 2$ in~~\eqref{eq:akplusun}.
  
\subsubsection{Upper-Bound for the Radius via Lower Bound  Propagation, $\la$, $\mu$, $\nu\!\in \R_+$, (Claims~$(b)$, $(c)$, $(d)$)}
In this subsection, we assume that the parameters $\lambda$, $\mu$, $\nu$ are real numbers. We will prove a comparison result between the case $\mu \ge 0$ and $\mu=0$. 

The case of $\mu \le 0$ can be reduced to the case $\mu\geq 0$ owing to the next Section~\ref{subsubsec:muneg}: we will see that the triplets   $(\lambda,\mu,\nu)$ ($\mu \ge 0$) and $(\lambda,-\mu,\nu)$ lead to solutions as fractional power series having the same convergence radius.

\begin{proposition}\label{prop:lowerb} Let $\alpha>0$. Let $(a_k)_{k\ge 0}$ and $(a^0_k)_{k\ge0}$ be solutions to $(A_{\lambda,\mu,\nu})$ and $(A_{\lambda,0,\nu})$ respectively, where $\lambda, \mu,\,  \nu$ are real numbers.

 \smallskip
\noindent $(a)$ For every $k\ge 1$, $a^0_{2k}  =0$ and $(a^0)^{*2}_{2k-1}=0$.
Moreover, the sequence defined  for every $k\ge 1$ by  $b_k=a^0_{2k-1}$ is solution of the recursive equation
\begin{equation}\label{eq:convb}
b_1= \frac{\nu}{\Gamma(\alpha+1)} \quad \mbox{ and  } \quad b_{k+1} = \lambda \frac{\Gamma(2\alpha k+1)}{\Gamma((2k+1)\alpha +1)}b^{*2}_{k+1},\; k\ge 1,
\end{equation}
where the squared convolution is  still defined by~~\eqref{eq:convoldef1} (the equation is consistent since  $b^{*2}_{k+1}$ only involves terms $b_\ell$, $\ell\le k$).

\smallskip
\noindent $(b)$ Assume $\alpha\!\in (0,2]$ and $\lambda$, $\mu$, $\nu \ge 0$. Then for every $k\ge 1$, $a_k\ge a^0_k\ge 0$,  so that $  R_{\psi_{\lambda,\mu,\nu}}\le  R_{\psi_{\lambda,0,\nu}}$.

\smallskip
\noindent $(c)$  Assume $\alpha\!\in (0,2]$. If $\la,\nu\ge 0$ and $\mu\le 0$, then (with obvious notations) $a^{(\la,\mu,\nu)}_k= (-1)^{k} a^{(\lambda,-\mu,\nu)}_k$, $k\ge 1$ so that $R_{\psi_{\lambda,\mu,\nu}}=   R_{\psi_{\lambda,-\mu,\nu}}$. Moreover if the non-negative sequence $a_k^{(\lambda,-\mu,\nu)}$   decreases for large enough $k$, then the expansion of $\psi_{\la,\mu,\nu}$  converges at $R_{\psi_{\lambda,-\mu,\nu}}$.
\end{proposition}

Note that  claim~$(c)$ is that of Theorem~\ref{thm:Radius} and  claim~$(a)$ is claim~$(d)$.

\bigskip
\noindent {\bf Proof.}  $(a)$ We proceed again by induction on $k$. If $k=1$, $(a^0)^{*2}_1=0$ and $a^0_{2} = \frac{\nu \mu }{\Gamma(2\alpha +1)} \frac{\Gamma(\alpha+1)}{\Gamma(2\alpha+1)}= 0$. 
 Assume $a^0_{2\ell}=0$, $1\le \ell \le k-1$, then 
 \[
 (a^0)^{*2}_{2k+1} = \sum_{\ell=1}^{2k} a^0_\ell a^0_{2k+1-\ell}.
 \]
 It is clear that either $\ell$ or   $ 2k+1-\ell$ is even. Consequently, $a_{\ell}a_{2k+1-\ell}=0$ so that  $ (a^0)^{*2}_{2k+1} =0$ and 
 \[
 a^0_{2(k+1)} =a^0_{2k+1+1}= \lambda \,\frac{\Gamma( \alpha (2k+1)+1)}{\Gamma(2\alpha(k+1)+1)} (a^0)^{*2}_{2k+1} =0.
 \]
 Let us look first  at the convolution at an odd even index. As $a^0_{\ell}=0$ for even index $\ell$, one has
 \begin{align*}
 (a^0)^{*2}_{2k} &= \sum_{\ell=1}^{2k-1} a^0_\ell a^0_{2k-\ell}  = \sum_{r=1}^{k}a^0_{2r-1} a^0_{2(k-r+1)-1} = \sum_{r=1}^kb_{r}b_{k+1-r} = b^{*2}_{k+1}.
 \end{align*}
 Plugging this in $(A_{\lambda,0,\nu})$ at index $2k+1$ yields~~\eqref{eq:convb}.

Notice that, by induction,  $a_k\ge 0$ for every $k\ge 1$ if $\lambda$, $\mu$, $\nu\ge 0$ (in particular  $a^0_k\ge 0$ as well). 

\smallskip
\noindent $(b)$ We proceed by induction on $k$.  It holds as an equality for $k=1$: $a_1=a^0_1 =\frac{\nu}{\Gamma(\alpha+1)}$. Assume $a_\ell\ge a^0_\ell\ge 0$, $1\le \ell \le k$. Then, using~~\eqref{eq:convoldef1},  
\[
a^{*2}_k = \sum_{\ell=1}^{k-1} a_{\ell}a_{k-\ell}\ge \sum_{\ell=1}^{k-1} a^0_{\ell}a^0_{k-\ell} = (a^0)^{*2}_k
\]
so that, using that $\mu \ge 0$, 
\[
a_{k+1} = \frac{\Gamma(\alpha k+1)}{\Gamma(\alpha(k+1)+1)}\big(\lambda a^{*2}_k+\mu a_k \big)\ge  \lambda\frac{\Gamma(\alpha k+1)}{\Gamma(\alpha(k+1)+1)} a^{*2}_k\ge  \lambda\frac{\Gamma(\alpha k+1)}{\Gamma(\alpha(k+1)+1)} (a^0)^{*2}_k= a^0_{k+1}.
\]
 
 \noindent $(c)$  
 \label{subsubsec:muneg}  
 Let $\tilde a_k = (-1)^{k-1} a_k$. It is clear that 
 $$
 \tilde a^{*2}_k = \sum_{\ell=1}^{k-1} (-1)^{\ell-1} a_\ell (-1)^{k-\ell-1}a_{k-\ell}= (-1)^k a^{*2}_k
 $$
(also obvious by setting $\rho=-1$ and replacing $\alpha-1$ by $0$ in former computations).
Consequently, $\tilde a_1 = a_1= \frac{\nu}{\Gamma (\alpha+1)}$ and 
\[
\tilde a_{k+1} =\frac{\Gamma(\alpha k +1)}{\Gamma(\alpha (k+1) +1)} (-1)^{k} \big(\lambda a^{*2}_k +\mu a_k\big)= \frac{\Gamma(\alpha k +1)}{\Gamma(\alpha (k+1) +1)}(\lambda\tilde a^{*2}_k-\mu \tilde a_k),
\]
so that $(\tilde a_k)_{k\ge 1}$ is solution to $(A_{\lambda, -\mu,\nu})$.  In particular, if we set formally
$$
\widetilde \psi_{\lambda, \mu,\nu}(u)= \sum_{k\ge1}a_ku^k
$$
then 
$$
\psi_{\lambda, \mu,\nu}(t)= \widetilde \psi_{\lambda, \mu,\nu}(t^{\alpha})
\qquad \mbox{ and }\qquad
\psi_{\lambda, -\mu,\nu}(t)= - \widetilde \psi_{\lambda, \mu,\nu}(-t^{\alpha})
$$
so that both expansions of  $\psi_{\lambda, \mu,\nu}$ and $\psi_{\lambda, -\mu,\nu}$ have the same convergence radius $R_{\lambda, \mu,\nu}= R_{\lambda, -\mu,\nu}$. See also the comments further on.$\qquad_{\diamondsuit}$

\begin{remark} Note that when $\lambda,\, \mu,\ \nu>0$ the coefficients $a_k>0$ so that $\lim_{t\to R_{\lambda, \mu,\nu}}(t)=+\infty$. As a consequence,  the  definition domain of the solution $\psi_{\lambda, \mu,\nu}$ of the Riccati equation on the positive real line is $[0,R_{\lambda, \mu,\nu})$. 

By contrast,  the series with terms $(-1)^kR_{\lambda, \mu,\nu}^ka_k$  is most likely  alternate  ($i.e.$ the absolute value of the generic term decreases toward $0$ for $k$ large enough). This implies  that the series will still converge at $t= R_{\psi_{\lambda,0,\nu}}$ {\it i.e.}
\[
\lim_{t\to R_{\lambda, -\mu,\nu}(t)}\psi_{\lambda, - \mu,\nu}(t) = \sum_{k\ge 1} (-1)^{k-1}R_{\lambda, -\mu,\nu}^k a_k\!\in \R.
\]
This explains the highly unstable numerical behavior observed near the explosion time compared to the case where all $a_k>0$, but also that the solution of the Riccati equation may be defined beyond $R_{\psi_{\lambda,0,\nu}}$, as already mentioned in the introduction.
\end{remark}

\bigskip
 Now, we are in position to prove claim~$(b)$ (lower bound of the radius). In the same manner as we proceed  for upper bound, we aim this time at propagating   a lower bound for  the non-zero subsequence of $(a^0_k)_{k\ge 0}$ {\it i.e.} the sequence $(b_k)_{k\ge1}$, namely 
 $$b_k \ge c\,\rho^k k^{\alpha-1},\; k\ge 1.
 $$
 Keeping in mind that the function $f_{\alpha}(x)= \big(x(1-x)\big)^{\alpha-1}$ is convex since $0<\alpha\le 1$,
 \begin{align*}
 b^{*2}_{k+1} = \sum_{\ell=1}^kb_{\ell}b_{k+1-\ell}  &\ge c^2\rho^{k+1}\sum_{\ell=1}^k  \ell^{\alpha-1} (k+1-\ell)^{\alpha-1} = c^2\rho^{k+1} (k+1)^{2(\alpha-1)}\sum_{\ell=1}^kf\left(\frac{\ell}{k+1}\right)\\
 &\ge c^2\rho^{k+1} (k+1)^{2(\alpha-1)} k \,f_{\alpha}\!\left(\sum_{1\le \ell\le k}\frac{\ell}{k(k+1)}\right)\\
 &= c^2\rho^{k+1} (k+1)^{2\alpha-1} \frac{k}{k+1} f_{\alpha}\left(\frac 12\right)= c^2\rho^{k+1} (k+1)^{2\alpha-1}\left(1+\frac 1k\right)^{-1} 2^{-2(\alpha-1)}.
 \end{align*}

Using Kershaw's Inequality with $x=2\alpha k$ and $s= \alpha$
 \begin{align*}
 \frac{\Gamma(2 \alpha k+1)}{\Gamma((2k+1)\alpha+1)} &= \frac{1}{\alpha(2k+1)}\frac{\Gamma(2 \alpha k+1)}{\Gamma(2\alpha k+ \alpha)} \ge  \frac{1}{\alpha(2k+1)}\left( 2\alpha k+ \frac{\alpha}{2}\right)^{1-\alpha}\\
 &= (2\alpha k)^{-\alpha} \left(1+\frac{1}{2k}\right)^{-1} \left(1+\frac{1}{4k}\right)^{1-\alpha}\\
 &=  (2\alpha )^{-\alpha} (k+1)^{-\alpha} \left(1+\frac1k\right)^{\alpha}\left(1+\frac{1}{2k}\right)^{-1} \left(1+\frac{1}{4k}\right)^{1-\alpha}.
 \end{align*}

 Plugging the above  two lower bounds for $ b^{*2}_{k+1}$ and $ \frac{\Gamma(2 \alpha k+1)}{\Gamma((2k+1)\alpha+1)}$  into~~\eqref{eq:convb} yields
 \[
 b_{k+1}\ge \lambda c^2\rho^{k+1}(k+1)^{\alpha-1} (2\alpha)^{-\alpha} 2^{-2(\alpha-1)}\widetilde b_k
 \] 
 where
 \[
 \widetilde b_k= \left(1+\frac 1k\right)^{\alpha-1}\left(1+\frac{1}{2k}\right)^{-1} \left(1+\frac{1}{4k}\right)^{1-\alpha}, \; k\ge 1.
 \]
 Consequently the propagation holds if 
 \[
 b_1=\frac{\nu}{\Gamma(\alpha+1)}\ge c\rho \quad \mbox{ and }\quad  (2\alpha)^{-\alpha} 2^{-2(\alpha-1)}\lambda \,c\, \widetilde b_k \ge 1, \; k\ge 1.
 \]
If we saturate the left inequality by setting $c= \frac{\nu}{\rho\alpha \Gamma(\alpha)}$, then the right  condition boils  down to 
$
\rho \le  2^{2-3\alpha}\frac{\alpha^{-(1+\alpha)}}{\Gamma(\alpha)}\lambda \nu \, \widetilde b_k, \; k\ge 1.
$
One checks that $\displaystyle \min_{k\ge1} \widetilde b_k = \widetilde b_1= \frac 23\left(\frac 58\right)^{1-\alpha}=2^{3\alpha-2}\frac{5^{1-\alpha}}{3}$
 which yields
 \[
 \rho^* = \frac{5^{1-\alpha}}{3}\frac{\alpha^{-\alpha}}{\Gamma(\alpha+1)\lambda\nu }\quad \mbox{ and }\quad c^*=  \frac{\nu}{\rho^*\alpha \Gamma(\alpha)}.
 \]
 Now,
 \begin{align*}
 R_{\psi_{\lambda,0,\nu}}^{-1} & =\limsup_k |a_k|^{\frac{1}{\alpha k}}\ge \limsup_k |a^0_k|^{\frac{1}{\alpha k}} = \limsup_k |a^0_{2k+1}|^{\frac{1}{(2k+1)\alpha }}\quad \mbox{ since }\quad a^0_{2k}=0,\, k\ge 0,\\
 &= \limsup_k |b_k|^{\frac{1}{(2 k+1)\alpha}} = \Big(\limsup_k |b_k|^{\frac{1}{k}}\Big)^{\frac{1}{2\alpha}} \ge\Big( \limsup_k \big(c^*k^{\alpha-1}(\rho^*)^k \big)^{\frac{1}{k}}\Big)^\frac{1}{2\alpha }= (\rho^*)^\frac{1}{2\alpha }
 \end{align*}
 which finally leads to the announced upper-bound 
 \[
  R_{\psi_{\lambda,\mu,\nu}}\le  R_{\psi_{\lambda,0,\nu}}\le \left(\frac{3}{5^{1-\alpha}}\right)^{\frac{1}{2\alpha}} \frac{ \alpha^{\frac 12 +\frac{1}{2\alpha}}\Gamma(\alpha)^{\frac{1}{2\alpha}}}{(\lambda\nu)^{\frac{1}{2\alpha}}}. \qquad_{\diamondsuit}
 \]
 
\noindent {\bf  Remark.} From the upper-bound result, we know that 
\[
 R_{\psi_{\lambda, 0,\nu}}  \ge  \frac{\alpha^{\frac12+\frac{1}{2\alpha}}\Gamma(2\alpha )^{\frac{1}{2\alpha}}}{2^{(\frac{1}{2\alpha}-1)^+}\big(  \Gamma(\alpha) |\lambda|  |\nu| \big)^{\frac{1}{2\alpha}} }.
\]

In particular we have established that, if $\lambda$, $\nu >0$,  there exist real constants $0<c_1(\alpha)<c_2(\alpha)$, only depending on $\alpha$, such that 
 \[
  \frac{c_1(\alpha)}{(\lambda\nu)^{\frac{1}{2\alpha}}}\le R_{\psi_{\lambda,0,\nu}}\le \frac{c_2(\alpha)}{(\lambda\nu)^{\frac{1}{2\alpha}}}
 \]
 with $c_1(\alpha) =  \frac{\alpha^{\frac12+\frac{1}{2\alpha}}\Gamma(2\alpha )^{\frac{1}{2\alpha}}}{2^{(\frac{1}{2\alpha}-1)^+} \Gamma(\alpha)^{\frac{1}{2\alpha}} }$ and 
 $c_2(\alpha) = \left(\frac{3}{5^{1-\alpha}}\right)^{\frac{1}{2\alpha}} \alpha^{\frac 12 +\frac{1}{2\alpha}}\Gamma(\alpha)^{\frac{1}{2\alpha}}$.

\subsection{Proof when $\alpha\!\in (1,2]$}
\subsubsection{Upper-Bound of the Radius by Lower Bound Propagation (Claim~$(a)$)} We start from the same the equation 
$(E_{\lambda, \mu,\nu})$  (see~~\eqref{eq:Eqak0}). If $\alpha\!\in (1,2]$, then we may write
\[
\frac{\Gamma (\alpha k +1)}{\Gamma (\alpha k +\alpha +1 )} = \frac{\Gamma (\alpha k +1)}{\Gamma (\alpha (k +1))} \frac{1}{\alpha(k+1)}.
\]

By Kershaw's Inequality we have, by setting $x= \alpha(k+1)-1$ and $s= 2-\alpha\!\in [0,1)$,

\begin{align*}
\frac{\Gamma (\alpha k +1)}{\Gamma (\alpha (k +1))} &\le \Big(\alpha(k+1)-\frac{\alpha}{2}\Big)^{1-\alpha} = \Big(\alpha(k+1/2)\Big)^{1-\alpha}
\end{align*}
so that 
\begin{align*}
\frac{\Gamma (\alpha k +1)}{\Gamma (\alpha k +\alpha +1 )} &\le \frac{1}{(\alpha\, k)^\alpha}\frac{k}{k+1}\left(\frac{2k}{2k+1}\right)^{\alpha-1}\le  \frac{1}{(\alpha\, k)^\alpha}
\end{align*}
since $(\alpha-1)\ge 0$. Now, using the concavity of the function $f(x)= x^{\alpha-1}(1-x)^{\alpha-1}$ over $[0,1]$ since $\alpha\ge 1$, we derive by Jensen's Inequality that 
\begin{align*}
\sum_{\ell=1}^{k-1}\ell^{\alpha-1}(k-\ell)^{\alpha-1}
&= k^{2(\alpha-1)}\sum_{\ell=1}^{k-1}f_{\alpha}\Big(\frac{\ell}{k}\Big) \le k^{2(\alpha-1)}(k-1)f_{\alpha}\!\left(\frac{1}{k-1}\sum_{\ell=1}^{k-1}\frac{\ell}{k}\right) \\
&=  k^{2\alpha-1}\left(1-\frac 1k\right) f_{\alpha}(1/2) =  k^{2\alpha-1}\left(1-\frac 1k\right) 2^{-2(\alpha-1)}\\
&\le k^{2\alpha-1} 2^{-2(\alpha-1)}.
\end{align*}
Consequently, assuming that $a_\ell\le C\rho^\ell \ell^{\alpha-1}$ for every $\ell=1,\ldots,k$, we derive that 
\begin{eqnarray*}
|a_{k+1}| &\le & C\rho^{k}  \alpha^{-\alpha}\left[|\lambda| Ck^{\alpha-1}2^{-2(\alpha-1)}+\frac{|\mu|}{ k}\right] = C\rho^{k+1}k^{\alpha-1}\frac{\alpha^{-\alpha}}{\rho}\left[|\lambda| C2^{-2(\alpha-1)}+ \frac{|\mu|}{k^\alpha} \right]\\
&\le & C\rho^{k+1}(k+1)^{\alpha-1}\frac{\alpha^{-\alpha}}{\rho}\left[|\lambda| C2^{-2(\alpha-1)}+ |\mu| \right],
\end{eqnarray*}
where we used that  $\alpha$ and  $\alpha-1\ge0$.  Hence, the propagation of the upper-bound holds if and only if 
\[
\frac{|\nu|}{\Gamma(\alpha+1)}\le C\rho \quad\mbox{ and }\quad |\lambda| C 2^{-2(\alpha-1)}+|\mu| \le \alpha^{\alpha}\rho.
\]


Following the lines of the case $\alpha\!\in (0,1]$, we derive that propagation does hold when  
\begin{equation}\label{eq:rho_*2}
\rho=\rho_*= \frac{|\mu| +\sqrt{\mu^2 +\frac{2^{2(2-\alpha)}\alpha^{\alpha-1}|\lambda||\nu|}{\Gamma(\alpha)}} }{2\,\alpha^{\alpha}} \quad\mbox{ and }\quad C=C_* = \frac{|\nu|}{\Gamma(\alpha+1)\rho_*},
\end{equation}
so that the convergence radius of $\psi$ satisfies 
$
R_{\psi} \ge (\rho_*)^{-\frac{1}{\alpha}}.
$

\noindent {\bf Remark.} It is the same formula as~~\eqref{eq:rho_*} except for the term $2^{2(2-\alpha)}$ which replaces $4B(\alpha,\alpha)$ since $2^{2(2-\alpha)}= 4\cdot2^{2(1-\alpha)} =  4f(1/2)$. This is due to the inversion of the convexity of the function $f_{\alpha}$ when $\alpha$ switches from $(0,1]$ to $[1,2)$.

\subsubsection{Upper Bound of the Radius by Lower Bound Propagation (Claim~$(b)$)} 

\medskip As a preliminary task, we note that the function  $f_{\alpha}(x)=\big(x(1-x)\big)^{\alpha-1}$ defined on $[0,1]$ is strictly concave when $\alpha\ge 1$, is symmetric with respect to $\frac 12$ and attains its maximum at $1/2$. Hence, $\tilde f_{\alpha}(x)=f_{\alpha}(1/2)-f_{\alpha}(x)$ satisfies the assumptions of Lemma~\ref{lem:f}, so that
\begin{equation*}
\frac 1k \sum_{\ell=1}^{k-1}\tilde f_{\alpha}\Big(\frac {\ell}{k}\Big)\geq \int_0^1\tilde f_{\alpha}(u)du
\end{equation*}
which finally yields, after easy manipulations, that, for every $k\ge 2$, 
\begin{equation}\label{eq:minorfconvex}
\frac 1k \sum_{\ell=1}^{k-1}f_{\alpha}\Big(\frac {\ell}{k}\Big)\geq \widetilde B(\alpha):= \int_0^1f_{\alpha}(u)du-\frac 12 f_{\alpha}\Big(\frac 12\Big) = B(\alpha,\alpha)-2^{1-2\alpha}>0.
\end{equation}
Notice that the positivity of $\widetilde B(\alpha)$ simply  follows from the strict concavity of $f_{\alpha}$.

\medskip
We assume that $\lambda$, $\nu>0$, $\mu \geq 0$ and that, for $\ell = 1, \ldots, k$, $a_k \ge c\rho^k k^{\alpha-1}$ for some real constant $c>0$. 

\medskip As in the case $\alpha\!\in (\frac 12,1]$, we will focus on the sequence $(b_k)$ since  Lemma~\ref{prop:lowerb} still applies. 

\medskip
As for the factor $\frac{\Gamma(2\alpha k+1)}{\Gamma(\alpha(2k+1)+1)} $, we may proceed as follows, still using Kershaw's Inequality, this time with $x=2\alpha k$ and $s= \alpha-1\!\in [0,1]$:
\begin{eqnarray*}
 \frac{\Gamma(2\alpha k+1)}{\Gamma(\alpha(2k+1)+1)}&=& \frac{1}{\alpha(2k+1)(\alpha(2k+1)-1)}\frac{\Gamma(2\alpha k+1)}{\Gamma(\alpha(2k+1)-1)}\\
\nonumber&\geq & \frac{1}{\alpha(2k+1)(\alpha(2k+1)-1)}\left(2\alpha k+ \frac{\alpha-1}{2}\right)^{2-\alpha}\\
\label{eq: minorGammaalphage1}
&\geq& \frac{(2\alpha)^{-\alpha}}{k^{\alpha}}\frac{2k}{2k+1}\frac{2\alpha k}{2\alpha k+\alpha-1}\left(1+\frac{\alpha-1}{4\alpha k}\right)^{2-\alpha} =  (2\alpha)^{-\alpha}(k+1)^{-\alpha}\,\widetilde b_k
\end{eqnarray*}
with
\[
\widetilde b_k
 = \left(1+\frac 1k\right)^{\alpha} \left(1+\frac{1}{2k}\right)^{-1}\left(1+\frac{\alpha-1}{2\alpha k}\right)^{-1}\left(1+\frac{\alpha-1}{4\alpha k}\right)^{2-\alpha}, \; k\ge 1.
\]
One checks that this sequence decreases toward $1$, so that $\inf_{k\ge 1} \widetilde b_k \ge 1$. Following the lines of the case $\alpha\!\in (0,1]$ yields
\[
b_{k+1} \ge  c^2 \lambda (2 \alpha)^{-\alpha}\widetilde B(\alpha) \rho^{k+1}(k+1)^{\alpha-1}, \; k\ge 0
\]
whereas  $\displaystyle b_1 = \frac{\nu}{\Gamma(\alpha+1)}$.  Hence, the propagation of the lower bound is satisfied if 
\[
\frac{\nu}{\Gamma(\alpha+1)}\ge c\rho \quad \mbox{ and }\quad c \lambda (2 \alpha)^{-\alpha}\widetilde B(\alpha) \ge 1.
\]
Finally, the lowest solution $\rho$ to this system is
$
\rho_* = \lambda \nu \frac{(2 \alpha)^{-\alpha}\widetilde B(\alpha) }{\Gamma(\alpha+1)},
$
corresponding to $C_*= \frac{(2 \alpha)^{\alpha}}{\lambda \widetilde B(\alpha)}$.

\section[Proof of Theorem~5.1]{Proof of Theorem~\ref{thm:non-homog}}\label{sec:thm5.1}
We now focus on the two separate cases on the next two subsections.
\subsection{Proof of Theorem~\ref{thm:non-homog}$(a)$ (Case $\alpha\!\in (\frac 12,1)$)} 
\medskip
\noindent {\sc \textbf{Step~1}} {\em Induction formula and propagation principle}. Let $\psi$ be formally defined by~~\eqref{eq:defpsi} and let $k(\ell)$ be   the valuation of the sequences $a_{.\ell}$.

The induction equation~~\eqref{eq:akell2} is obvious by identification. Note that $(a_{k,0})_{k\ge1}$ satisfies the  recursion~~\eqref{eq:Eqak0} of the case $u=0$, so that $k(0)= +\infty$ if $\nu=0$ and $k(0)=1$ otherwise. The main point is to determine the valuation $k(\ell)$ when $\nu\neq 0$. 

We start with the fact that $k(0)=1$ (corresponding to the expansion when $u=0$) and $k(1)=1$ due to the presence of the fractional monomial
$\frac{u}{\Gamma(\alpha)}t^{\alpha-1}$.

Let $\ell\ge 1$. A term $t^{\alpha k-\ell}$ comes in~~\eqref{eq:defpsi} for the $\alpha$-fractional integration of a term $t^{(k-1)\alpha -\ell}$, which itself comes either directly, either from the expansion at the same level $\ell$ of $\mu \psi$ or from a product $t^{\alpha k_1-\ell_1} \cdot t^{\alpha k_2-\ell_2}$ with $\ell_1+\ell_2 = \ell$ and $k_1+k_2 = k-1$ induced by the convolution term. Hence, $k(0)=1$ and, for every $\ell\ge 1$, 
\[
k(\ell) =\min\left[ \min_{\ell_1+\ell_2=\ell}\big[k(\ell_1)+k(\ell_2)\big], k(\ell-1)\right]+1.
\] 
It is clear that this minimum cannot be attained at $\ell_1$ or $\ell_2 =0$, since it leads to a non-sense. Then we can check that the formula
\[
k(\ell)= 2\ell-1, \quad \ell\ge 1,
\]
is solution to the above minimization problem. Finally, $k(\ell)= (2\ell-1)\vee 1$, $\ell\ge 0$. This justifies the definition of~~\eqref{eq:defpsi} and the double index discrete convolution~~\eqref{eq:akell2}.

\medskip To show the existence of a positive convergence radius $R_{\psi}$ shared by all the fractional series at all the  levels, we will prove that he coefficients $a_{k,\ell}$ satisfy the following upper bound for every level $\ell$ and every $k\ge k(\ell)$:
\begin{equation}\label{eq:ubuneq0}
|a_{k,\ell}| \le C \theta^{\ell}\rho^{k} \big(k-k(\ell)+1\big)^{\frac{\alpha}{2}-1}(\ell\vee1 )^{\frac{\alpha}{2}-1},\; k\ge k(\ell), \; \ell\ge 0
\end{equation}
(with $k(\ell)= (2\ell-1)\vee 1$). The method of proof consists in propagating this bound by a nested induction on the index $k$   and  through the levels $\ell$.

\smallskip

\medskip
 \noindent {\sc \textbf{Step~2}} {\em Propagation of the initial value across the levels $\ell\ge 0$}.  Following~~\eqref{eq:ubuneq0}, we want to propagate by induction the bound
 \begin{equation}\label{eq:ubuneq0start}
 |a_{k(\ell),\ell}|\le C \rho ^{k(\ell)} \theta^{\ell}(\ell\vee 1)^{\frac{\alpha}{2}-1}, \quad \ell\ge 0,
 \end{equation}
keeping in mind that $k(\ell)= (2\ell-1)\vee 1$, $a_{1,0}= \frac{\nu }{\Gamma(\alpha+1)}$ and $a_{1,1}= \frac{u}{\Gamma(\alpha)}$. The levels $\ell=0,1$ yield direct conditions to be used later. Let $\ell\ge 2$.

Applying the induction formula~~\eqref{eq:akell2}  with $k=k(\ell)= 2\ell-1$, we obtain
 \[
 a_{2\ell-1,\ell} =\Big(\mu \,a_{2(\ell-1),\ell}+ \lambda a^{*2}_{2(\ell-1), \ell}\Big)\frac{\Gamma((2\alpha-1)(\ell-1))}{\Gamma((2\alpha-1)(\ell-1)+\alpha)}.
 \] 
First note that  $a_{2(\ell-1),\ell}=0$ since $2(\ell-1)\le 2\ell-1$ and $\ell\ge 2$. Moreover, 
\[
a^{*2}_{2(\ell-1), \ell}= \widetilde a^{*2}_{2(\ell-1), \ell}= \sum_{\begin{smallmatrix}k_1+k_2= 2(\ell-1)\\ \ell_1+\ell_2=\ell\\ k_i\geq 2\ell_i-1, \,  \ell_i\ge 1\end{smallmatrix}}a_{k_1,\ell_1}a_{k_2,\ell_2} = \sum_{\ell'=1}^{\ell-1}a_{2\ell'-1,\ell'}a_{2(\ell-\ell')-1,\ell-\ell'}
\]
so that we get the following induction formula for the starting values $a_{2\ell-1, \ell}$, $\ell\ge 1$: 
\[
 a_{2\ell-1,\ell} = \lambda  \left(\sum_{\ell'=1}^{\ell-1}a_{2\ell'-1,\ell'}a_{2(\ell-\ell')-1,\ell-\ell'}\right)\frac{\Gamma((2\alpha-1)(\ell-1))}{\Gamma((2\alpha-1)(\ell-1)+\alpha)}, \quad\ell\ge2,\quad a_{1,1}= \frac {u}{\Gamma(\alpha)}.
\]
Let $\ell\ge 2$. Assume that ~~\eqref{eq:ubuneq0start} is satisfied by $a_{k(\ell'), \ell'}$  for every lower level $\ell'\!\in \{0,1,\ldots, \ell-1\}$.  Then
\begin{align*}
|a_{2\ell-1, \ell}| &\le \lambda\left[\sum_{\ell'=1}^{\ell-1} C^2 \rho^{2\ell'-1}\theta^{\ell'} (\ell')^{{\frac{\alpha}{2}-1}}\ \rho^{2(\ell-\ell')-1}\theta^{\ell-\ell'} (\ell-\ell')^{{\frac{\alpha}{2}-1}}\right]\frac{\Gamma((2\alpha-1)(\ell-1))}{\Gamma((2\alpha-1)(\ell-1)+\alpha)}\\
&\le \lambda C^2\rho^{2(\ell-1)}\theta^{\ell} \left[\sum_{\ell'=1}^{\ell-1}  (\ell')^{{\frac{\alpha}{2}-1}} (\ell-\ell')^{{\frac{\alpha}{2}-1}}\right]\frac{\Gamma((2\alpha-1)(\ell-1))}{\Gamma((2\alpha-1)(\ell-1)+\alpha)}\\
&\le C\rho^{2\ell-1}\theta^{\ell}\, \frac{\lambda C}{\rho}B\Big(\frac{\alpha}{2},\frac{\alpha}{2}\Big) \ell^{\alpha-1}\frac{\Gamma((2\alpha-1)(\ell-1))}{\Gamma((2\alpha-1)(\ell-1)+\alpha)}.
\end{align*}
It follows from Kershaw's Inequality~~\eqref{eq:Kershawb}, applied with $x=(2\alpha-1)(\ell-1)+\alpha$ and $s= 1-\alpha$, and the elementary identity $\Gamma(z+1)=z\Gamma(z)$  that 
\begin{align*}
\frac{\Gamma((2\alpha-1)(\ell-1))}{\Gamma((2\alpha-1)(\ell-1)+\alpha)}&= \frac{(2\alpha-1)(\ell-1)+\alpha}{(2\alpha-1)(\ell-1)}\frac{\Gamma((2\alpha-1)(\ell-1)+1)}{\Gamma((2\alpha-1)(\ell-1)+\alpha+1)}\\
&\le  \Big(1+ \frac{\alpha}{(2\alpha-1)(\ell-1)}\Big)\Big((2\alpha-1)(\ell-1) +\alpha+\frac{1-\alpha}{2}\Big)^{-\alpha}\\
&= \ell^{-\alpha}(2\alpha-1)^{-\alpha}\Big(1+ \frac{\alpha}{(2\alpha-1)(\ell-1)}\Big)\Big(1+ \frac{3(1-\alpha)}{2(2\alpha-1)\ell}\Big)^{-\alpha}.
\end{align*}

Hence
\begin{equation}\label{eq:ubuneq0final}
|a_{2\ell-1, \ell}| \le  C\rho^{2\ell-1}\theta^{\ell}\ell^{{\frac{\alpha}{2}-1}}\, \frac{\lambda C}{\rho}\kappa^{(1)}_{\alpha}
B\Big(\frac{\alpha}{2},\frac{\alpha}{2}\Big) 
\end{equation}

\begin{align*}
\mbox{where  }\qquad \kappa^{(1)}_{\alpha}&= (2\alpha-1)^{-\alpha}\sup_{\ell\ge 2}\left[\Big(1+ \frac{\alpha}{(2\alpha-1)(\ell-1)}\Big)\Big(1+ \frac{3(1-\alpha)}{2(2\alpha-1)\ell}\Big)^{-\alpha}\ell^{-\frac{\alpha}{2}}\right]\qquad\qquad\\
&= \left(1+\frac{\alpha}{2\alpha-1}\right)2^{-\frac{\alpha}{2}}(2\alpha-1)^{-\alpha}\left(1+\frac{3(1-\alpha)}{4(2\alpha-1)}\right)^{-\alpha} 
\end{align*}
since one easily checks (e.g. with the use of Mathematica) that the maximum is achieved at $\ell=2$. The condition on $\rho$ and $\theta$ for propagation hence reads
\begin{equation}\label{eq:propag-Init}
 \frac{|\nu|}{\Gamma(\alpha+1)}\le C\rho,\quad  \frac{|u|}{\Gamma(\alpha)} \le C\rho\, \theta  \quad\mbox{ and }\quad \frac{\lambda C}{\rho} \kappa^{(1)}_{\alpha} B\Big(\frac{\alpha}{2},\frac{\alpha}{2}\Big) \le 1,
\end{equation}
where the first two inequalities come from the initial values at levels $\ell=0,1$ and the third one ensures the propagation of the upper-bound in~~\eqref{eq:ubuneq0final}. These three inequalities are in particular satisfied if 
\[
\frac{\lambda |u|}{\rho^2\theta\,\Gamma(\alpha)} \kappa^{(1)}_{\alpha}
B\Big(\frac{\alpha}{2},\frac{\alpha}{2}\Big) \le 1 \quad \mbox{ and }\quad C =  \frac{C_0}{\rho}
\] 
where 
\begin{equation}\label{eq:fC0}
C_0=C_0(\alpha,\theta)= \Big[ \frac{|\nu| }{\Gamma(\alpha+1)}\vee \frac{|u|}{\theta \Gamma(\alpha)}\Big],
\end{equation} 
so that 
\begin{equation}\label{eq:r1c1}
\rho\ge \rho_1=\rho_1(\theta) = \sqrt{\frac{\lambda |u|B\big(\frac{\alpha}{2},\frac{\alpha}{2}\big) \kappa^{(1)}_{\alpha}}{\theta\, \Gamma(\alpha)}}\quad \mbox{ and }\quad C\ge  \frac{|u|}{\Gamma(\alpha)\rho\,\theta}.
\end{equation}
\medskip
\noindent {\sc \textbf{Step~3}} {\em Propagation through a level $\ell\ge 0$}.  Let $k\geq k(\ell)+1$. Assume that the above bound holds for every  couple $(k',\ell')$ such that level $\ell'<\ell$ and  $k'\ge 2\ell'-1$ or $\ell'=\ell$ and  $2\ell-1\le k'\le k-1$.  
\begin{align*}
|a^{*2}_{k-1,\ell} |&\le \;C^2 \sum_{\begin{smallmatrix}k_1+k_2= k-1\\ \ell_1+\ell_2=\ell\\ k_i\geq k(\ell_i), \,  \ell_i\ge 0\end{smallmatrix}} |a_{k_1, \ell_1}| |a_{k_2,\ell_2}|\\
& = C^2 \theta^{\ell}\rho^{k-1}\hskip-0.25cm \sum_{\begin{smallmatrix} \ell_1+\ell_2=\ell\\   \ell_i\ge 0\end{smallmatrix}} \big((\ell_1\vee 1) (\ell_2\vee 1)\big)^{\frac{\alpha}{2}-1}\hskip-0.25cm \sum_{\begin{smallmatrix}k_1+k_2= k-1\\   k_i\geq k(\ell_i)\end{smallmatrix}} (k_1-k(\ell_1)+1)^{\frac{\alpha}{2}-1}  (k_2-k(\ell_2)+1)^{\frac{\alpha}{2}-1}\\
&=  C^2 \theta^{\ell}\rho^{k-1}\hskip-0.25cm \sum_{\begin{smallmatrix} \ell_1+\ell_2=\ell\\   \ell_i\ge 0\end{smallmatrix}}\big((\ell_1\vee 1) (\ell_2\vee 1)\big)^{\frac{\alpha}{2}-1}\hskip-0.25cm 
\sum_{\begin{smallmatrix}k_1+k_2= k-(k(\ell_1)+k(\ell_2))+1\\   k_i\ge 1 \end{smallmatrix}} k_1^{\frac{\alpha}{2}-1}  k_2^{\frac{\alpha}{2}-1}\\
&\le C^2 \theta^{\ell}\rho^{k-1}\hskip-0.25cm \sum_{\begin{smallmatrix} \ell_1+\ell_2=\ell\\   \ell_i\ge 0\end{smallmatrix}}\hskip-0.1cm \big((\ell_1\vee 1) (\ell_2\vee 1)\big)^{\frac{\alpha}{2}-1}\hskip-0.1cm 
 B\Big(\frac{\alpha}{2},\frac{\alpha}{2} \Big) \big(k-(k(\ell_1)+k(\ell_2))+1\big)^{2(\frac{\alpha}{2}-1)+1}
 \end{align*}
owing (twice) to Lemma~\ref{lem:f}$(a)$ since $\alpha/2- 1 <0$. Now, as $k(\ell_1)+k(\ell_2) \ge k(\ell) $ by defintion of the valuation, we deduce that
\[
\big(k-(k(\ell_1)+k(\ell_2))+2\big)^{2(\frac{\alpha}{2}-1)+1}\le \big(k-k(\ell)+1\big)^{2(\frac{\alpha}{2}-1)+1}\le \big(k-1-k(\ell)+1\big)^{2(\frac{\alpha}{2}-1)}
\]
so that 
\[ 
|a^{*2}_{k-1,\ell} |\le C^2 \theta^{\ell}\rho^{k-1} B\Big(\frac{\alpha}{2},\frac{\alpha}{2} \Big) \big(k-1-k(\ell)+1\big)^{2(\frac{\alpha}{2}-1)+1}\hskip-0.25cm \sum_{\begin{smallmatrix} \ell_1+\ell_2=\ell\\   \ell_i\ge 0\end{smallmatrix}}\hskip-0.1cm \big((\ell_1\vee 1) (\ell_2\vee 1)\big)^{\frac{\alpha}{2}-1}.
\]

Now note that, if $\ell\ge 1$, 
\begin{align*}
\sum_{\begin{smallmatrix}  \ell_1+\ell_2=\ell\\   \ell_i\ge 0\end{smallmatrix}}( \ell_1\vee 1)^{\frac{\alpha}{2}-1}(\ell_2\vee1)^{\frac{\alpha}{2}-1} &\le \ell^{\frac{\alpha}{2}-1}+\sum_{\ell_1=1}^{\ell-1}\ell_1^{\frac{\alpha}{2}-1} (\ell-\ell_1)^{\frac{\alpha}{2}-1}\le \ell^{\frac{\alpha}{2}-1}+
B\Big(\frac{\alpha}{2},\frac{\alpha}{2}\Big)\ell^{\alpha-1}
\end{align*}
owing to Lemma~\ref{lem:f}$(a)$. 

If $\ell=0$, $\displaystyle \sum_{\begin{smallmatrix}  \ell_1+\ell_2=\ell\\   \ell_i\ge 0\end{smallmatrix}}( \ell_1\vee 1)^{\frac{\alpha}{2}-1}(\ell_2\vee1)^{\frac{\alpha}{2}-1} =1$ so that the above right inequality  still holds by replacing $\ell $ by $\ell\vee 1$. 
Now,  combining these inequalities yields

\begin{align*}
|a^{*2}_{k-1,\ell} |&\le  C^2 \theta^{\ell}\rho^{k-1} B\Big(\frac{\alpha}{2},\frac{\alpha}{2}\Big) 
\big(k-k(\ell)+1\big)^{\alpha-1}\left(\ell^{\frac{\alpha}{2}-1}+B\Big(\frac{\alpha}{2},\frac{\alpha}{2}\Big)\ell^{\alpha-1}\right).
\end{align*}

As for the ratio of Gamma functions, one has, using $\Gamma(x+1)=x\Gamma(x)$ and Kershaw's inequality~~\eqref{eq:Kershawb}, 
\begin{align*}
\frac{\Gamma(\alpha (k-1)-\ell+1)}{\Gamma(\alpha k-\ell+1)} &= \frac{\alpha k-\ell+1}{\alpha k-\ell+1-\alpha}\frac{\Gamma(\alpha k-\ell+2-\alpha)}{\Gamma(\alpha k-\ell+2)}\\
&\le  \frac{\alpha k-\ell+1}{\alpha k-\ell+1-\alpha}\Big(\alpha k -\ell +1 +\frac{1-\alpha}{2}\Big)^{-\alpha}.
\end{align*}
For every $\ell\ge 1$ and  $k\ge k(\ell)+1\ge 2\ell$, 
\[
\frac{\alpha k-\ell+1}{\alpha k-\ell+1-\alpha}= 1+ \frac{\alpha}{\alpha k-\ell+1-\alpha} \le 1+ \frac{\alpha}{(2\alpha-1)\ell+1-\alpha}\le 2.
\]
Now, we note that 
\[
\alpha k -\ell +1 +\frac{1-\alpha}{2}= \alpha (k -2(\ell-1)) +(2\alpha-1)(\ell-1) +\frac{1-\alpha}{2}.
\]
Combining the above inequality with this identity and the  elementary inequality between non-negative real numbers
\begin{equation}\label{eq:elemab}
(a+b)^{-\alpha}\le (2ab)^{-\frac{\alpha}{2}}, \; a,\, b\ge 0,
\end{equation}
 we obtain (once noted that $2(\ell-1)=k(\ell)-1$) 
\[
\frac{\Gamma(\alpha (k-1)-\ell+1)}{\Gamma(\alpha k-\ell+1)} \le \underbrace{\frac{\alpha}{3\alpha-1}\big(2\alpha(2\alpha-1)\big)^{-\frac{\alpha}{2}}}_{=:  \kappa^{(1)}_{\alpha}}\big(k-k(\ell)+1\big)^{-\frac{\alpha}{2}}\Big(\ell-1 +\frac{1-\alpha}{2(2\alpha-1)}\Big)^{-\frac{\alpha}{2}}.
\]
Now, still using that  $\ell\ge 1$, 
\begin{align*}
\Big(\ell-1 +\frac{1-\alpha}{2(2\alpha-1)}\Big)^{-\frac{\alpha}{2}}&= \ell^{-\frac{\alpha}{2}}\Big(\frac{\ell}{\ell-1 +\frac{1-\alpha}{2(2\alpha-1)}}\Big)^{\frac{\alpha}{2}}\le   \Big(\frac{2(2\alpha-1)}{1-\alpha}\vee 1\Big)^{\frac{\alpha}{2}}
\ell^{-\frac{\alpha}{2}}.
\end{align*} 
However, if $\ell= \alpha=1$, this bound is infinite. Coming back to the original ratio yields, for every $k\ge k(1)+1=2$, 
\[
\frac{\Gamma(\alpha (k-1)-\ell+1)}{\Gamma(\alpha k-\ell+1)} =\frac{\Gamma( k-1 )}{\Gamma(  k ) } =\frac{1}{k-1}=\frac{\sqrt{k}}{k-1}  k^{-\frac 12}\le \sqrt{2}\, k^{-\frac 12}.
\]

If $\ell=0$, Kershaw's Inequality~~\eqref{eq:Kershawb} yields 
\[
\frac{\Gamma(\alpha(k-1)+1)}{\Gamma(\alpha k+1)}\le \Big(\alpha k +\frac{1-\alpha}{2}\Big)^{-\alpha}.
\]
This implies,
\[
\frac{\Gamma(\alpha(k-1)+1)}{\Gamma(\alpha k+1)}\le \left(\alpha \sqrt{k} + \frac{1-\alpha}{2}\right)^{-\alpha}\le   \alpha ^{-\alpha}  k^{-\frac{\alpha}{2}}=  \alpha ^{-\alpha} (k-k(0)+1)^{-\frac{\alpha}{2}}(0\vee 1)^{-\frac{\alpha}{2}}.
\]

Finally, collecting the above ienqualities, we obtained that, for every $\ell\ge 0$ and every $k\ge k(\ell)+1$,
\[
\frac{\Gamma(\alpha (k-1)-\ell+1)}{\Gamma(\alpha k-\ell+1)}\le \kappa_{\alpha}^{(2)} (k-k(\ell)+1)^{-\frac{\alpha}{2}}\ell^{-\frac{\alpha}{2}}
\]
with $\kappa_{\alpha}^{(2)} =  \frac{\alpha}{3\alpha-1}\big(2\alpha(2\alpha-1)\big)^{-\frac{\alpha}{2}} \Big[\Big(\frac{2(2\alpha-1)}{1-\alpha}\vee 1\Big)^{\frac{\alpha}{2}}\mbox{\bf 1}_{\{\alpha\neq 1\}}+\sqrt{2}\,
\mbox{\bf 1}_{\{\alpha =  1\}} \Big]$.

\smallskip Collecting all these partial results and plugging them in~~\eqref{eq:Convol2} yields 
\[
|a_{k,\ell}|\le C\theta^{\ell} \rho^{k}\big(k-k(\ell)+1\big)^{\frac{\alpha}{2}-1}\ell^{\frac{\alpha}{2}-1}\frac{ \kappa^{(2)}_{\alpha}}{\rho}\left(|\lambda| C  B\Big(\frac{\alpha}{2},\frac{\alpha}{2}\Big)^2 +|\lambda|CB\Big(\frac{\alpha}{2},\frac{\alpha}{2}\Big) +  2^{1-\frac{\alpha}{2}}|\mu|   \right)
\]
where we used $\ell^{-\frac{\alpha}{2}}\le 1$ and $\Big(\frac{k-k(\ell)}{k-k(\ell)+1}\Big)^{\frac{\alpha}{2}-1}\le 2^{1-\frac{\alpha}{2}}$.

%
Consequently,  propagation inside a  level   boils down to
\begin{align}\label{eq:propaglk2a}
\kappa^{(2)}_{\alpha}\left(|\lambda| CB\Big(\frac{\alpha}{2},\frac{\alpha}{2}\Big)\Big[   B\Big(\frac{\alpha}{2},\frac{\alpha}{2}\Big)   +1\Big]  +2^{1-\frac{\alpha}{2}}|\mu|   \right)\le \rho.
\end{align}


We combine now this constraint on $\rho $   with those on $C$ and $\rho$ coming 
the propagation across initial values, that is~~\eqref{eq:r1c1} {\it i.e.} $C=\frac{C_0(\alpha,\theta)}{\rho}$, where $ C_0(\alpha,\theta)$ is given by~~\eqref{eq:fC0} and~~\eqref{eq:r1c1}.
The constraint~~\eqref{eq:propaglk2a} reads
\begin{equation}\label{eq:rhoalpah0.5}
\rho^2-2^{1-\frac{\alpha}{2}}\kappa^{(2)}_{\alpha}|\mu|\rho  - \kappa^{(2)}_{\alpha}C_0(\alpha,\theta)\bar B(\alpha/2)|\lambda| \ge 0,
\end{equation}
where 
\begin{equation}\label{eq:Bbaralpha/2}
 \bar B(\alpha/2)= B\Big(\frac{\alpha}{2},\frac{\alpha}{2}\Big)\Big[   B\Big(\frac{\alpha}{2},\frac{\alpha}{2}\Big)   +1\Big].
 \end{equation}
Then all the constraints are fulfilled by parameters $(\rho, C)$ satisfying
\[
\rho\ge \rho_* (\alpha,\theta)=   \rho_2(\theta)
 \vee \rho_1(\theta)\quad \mbox{ and }\quad C = \frac{C_0(\alpha,\theta)}{\rho}.
\]
where $\rho_1(\theta)$ is given by~~\eqref{eq:r1c1} and 
\[
\rho_2(\theta)= 2^{-\frac{\alpha}{2}}\kappa^{(2)}_{\alpha}\left(|\mu|+\sqrt{|\mu|^2+\frac{2^{\alpha-2}|\lambda| C_0(\alpha,\theta)\bar B(\alpha/2)}{\kappa^{(2)}_{\alpha}}}\right)
\]
is the positive solution of the equation associated to Inequation~~\eqref{eq:rhoalpah0.5}.
In what follows we consider the admissible pair $(\rho, C)= \big( \rho_* (\alpha,\theta), \frac{C_0(\alpha,\theta)}{\rho_* (\alpha,\theta)}\big)$. We  derive that the   convergence radii $R_{{\psi}_{\ell}}$ of the functions $\psi_{\ell}$  all satisfy 
\[
R_{{\psi}_{\ell}}\ge \rho_*(\alpha,\theta)^{-\frac{1}{\alpha}}.
\]
Now we have to ensure the summability of the functions $\psi_{\ell}$. It suffices to consider the   levels $\ell\ge 1$, so that $k(\ell)= 2\ell-1$. Let $t\!\in I_*= \big[0, \rho_*^{-\frac{1}{\alpha}}\big)$, where $\rho_*= \rho_*(\alpha, \theta)$. Elementary computations using the upper bound~~\eqref{eq:ubuneq0} for the coefficients $a_{k,\ell}$ and a change of index $k-2(\ell-1)\to k$ show that
\[
|\psi_{\ell}(t)| \le C_* \theta^{\ell} \ell^{\frac{\alpha}{2}-1} \rho_*^{2(\ell-1)} t^{(2\alpha-1)\ell-2\alpha}\widetilde \psi(t),
\]
where $\widetilde \psi(t)= \sum_{k\ge 1} \rho_*^k k^{\frac{\alpha}{2}-1} t^{\alpha k}<+\infty$ does not depend neither on $\ell$ nor on $\theta$ and is uniformly bounded on any compact interval of $I_*$.

To ensure the summability  of the functions $\psi_{\ell}$ for every $t\!\in \big[0, \rho^{-\frac{1}{\alpha}}\big)$, we note that 
\begin{align*}
\sum_{\ell\ge 1}  \theta^{\ell} \ell^{\frac{\alpha}{2}-1} \rho_*^{2(\ell-1)}t^{(2\alpha-1)\ell-2\alpha} & < t^{-2\alpha} \sum_{\ell\ge 1}  \theta^{\ell} \ell^{\frac{\alpha}{2}-1} \rho_*^{2(\ell-1)}\rho_*^{-\frac{(2\alpha -1)\ell-1}{\alpha}}  = t^{-2\alpha} \rho_*^{-2-\frac{1}{\alpha}}\sum_{\ell\ge 1}  (\theta\rho_*^{\frac{1}{\alpha}})^{\ell} \ell^{\frac{\alpha}{2}-1} .
\end{align*}

Since $\rho_*(\alpha, \theta)= O\big(\theta^{-\frac 12}\big)$ as $\theta \to 0$, it is clear that $\displaystyle \lim_{\theta\to 0}\theta \rho_*(\alpha, \theta)= 0$ so there exists $\theta>0$ such that $0<\theta \rho_*(\alpha, \theta)<1$. Hence
\[
\theta_* = \min\big\{\theta : \theta \rho_*(\alpha, \theta)\ge 1\big\}<+\infty \quad \mbox{ and }\quad  \theta_* \rho_*(\alpha, \theta_*)=1
\]
owing to the continuity of $\theta \mapsto \theta \rho_*(\alpha, \theta)$. As $\rho_*(\alpha,\theta)$ is non-increasing in $\theta$, $\theta_*$ yields the highest admissible value for $\rho_*(\alpha, \theta)$ so that 
\[
R_{\psi}\ge  \rho_*(\alpha, \theta_*)^{-\frac{1}{\alpha}}
\]
in the sense that  {\em the doubly indexed series~~\eqref{eq:defpsi}  that defines the function $\psi$ is normally converging on any compact interval of $(0, R_{\psi})$.}

The following proposition establishes a semi-closed form for the starting values $a_{k(\ell), \ell}=a_{2\ell-1, \ell}$ at each level $\ell\ge 1$.
\begin{proposition}[Closed form for the starting values] For every $\ell\ge 1$,
\[
a_{2\ell-1, \ell}=\lambda^{\ell -1}\left(\frac{ u}{\Gamma(\alpha)}\right)^{\!\!\ell}c_\ell,
\]
with $c_1=1$ and for $\ell\geq2$, 
\begin{equation}
\label{eq:c-ell1}
c_\ell=\frac{\Gamma ((2\alpha -1)(\ell -1))}{\Gamma ((2\alpha -1)\ell+1-\alpha)}\sum_{j=1}^{\ell-1}c_j c_{\ell-j}.
\end{equation}
\end{proposition}
\begin{proof}  We prove the identity by induction: for $\ell=1$ it is obvious. Assume now it holds for $\ell\ge 1$. Then 
\[
a_{2(\ell+1)-1,\ell+1}= \lambda   a^{*2}_{2\ell, \ell+1}\frac{\Gamma(\ell(2\alpha-1))}{\Gamma((\ell+1)(2\alpha-1)+1-\alpha)}
\]
since $a_{2\ell, \ell+1}=0$ (keep in mind that $2\ell < k(\ell+1)= 2\ell+1$). Now, we rely on~~\eqref{eq:Convol2}. First note that $k_i(\ell_i)= 2\ell_i-1$, $i=1,2$ if both $\ell_i\ge1$. Hence $k_i\ge k_i(\ell_i)$ implies $k_1+k_2\ge 2(\ell_1+\ell_2)-1= 2\ell+1$ and consequently $k_i= 2\ell_i-1$, $i=1,2$. If $\ell_1=0$, $\ell_2 = \ell+1$ so that $k_2\ge k_2(\ell_2)= 2(\ell+1)-1$ which implies $k_1=0$. As $a_{0,0}$ is always $0$ by construction, we finally obtain 
\[
 a^{*2}_{2\ell, \ell+1}=\sum_{j=1}^\ell a_{2j-1, j}a_{2(\ell+1-j)-1, \ell+1-j}
=\lambda^{\ell-1}\left(\frac{u}{\Gamma(\alpha)}\right)^{\ell+1}\sum_{j=1}^\ell c_j c_{\ell+1-j}.
\]
Therefore 
\[
a_{2(\ell+1)-1,\ell+1}= \lambda^{\ell}\left(\frac{u}{\Gamma(\alpha)}\right)^{\ell+1}c_{\ell+1},
\]
where $c_{\ell+1}$ is given by~~\eqref{eq:c-ell1} (at level $\ell+1$).
The conclusion follows by induction.
\end{proof}

\subsection{Proof of Claim$(b)$ of Theorem~\ref{thm:non-homog} (Case $\alpha\!\in 
(1,2]$)}

We still search for a function of the form~~\eqref{eq:defpsi}, more precisely
\begin{align*}\label{eq:defpsi}
\psi(t) =&\sum_{\ell\ge 0} \psi_{\ell}(t) =\sum_{\ell\ge 0} \sum_{k\ge k(\ell)}a_{k,\ell} t^{\alpha k-\ell}
\end{align*}
where the valuation $k(\ell)$ is specified in the Lemma below. We set for convenience $a_{k,\ell}= 0$ for $k<k(\ell)$.
\begin{lemma}\label{lem:recakl2} $(a)$ {\em  Valuation}.  When $\nu, u,v\neq0$, then $k(\ell):=\min\{k\ge 1: a_{k,\ell}\neq 0\}$ satisfies  
\[
k(\ell)= 1,\; \ell=0,1,2\quad \mbox{ and }\quad k(\ell)= \ell,\quad \ell\ge 3.
\]
If $\nu$, $u$ or $v = 0$, $k(\ell)$ as defined above is still admissible as a lower  bound in the above sum.

\smallskip
\noindent $(b)$ {\em Induction formula}.  The coefficients $a_{k,\ell}$ still  satisfy the doubly indexed recursion~~\eqref{eq:akell2}, 
this time with $a_{1,0}= \frac{\nu}{\Gamma(\alpha+1)}$, $a_{1,1}= \frac{u}{\Gamma(\alpha)}$,  $a_{1,2} = \frac{v}{\Gamma(\alpha-1)}$ and $a_{1,\ell}= 0$, $\ell\ge 3$. Note that $a_{1,\ell}^{*2}= 0$ for every $\ell\ge0$.

\smallskip
\noindent $(c)$ {\em Closed form for the starting values}.   For every  $\ell \geq 1$, 
\begin{align}
a_{2\ell, 2\ell } &= \frac{\mu }{2\lambda} \left(\frac{2\lambda v }{\Gamma (\alpha -1)}\right)^\ell 
\prod_{j=1}^\ell 
\frac{\Gamma((2j-1)(\alpha-1))}{\Gamma((2j-1)(\alpha-1)+\alpha)}\label{eq:ellpari}\\
a_{2\ell +1, 2\ell +1} &= \frac{u }{\Gamma(\alpha)} \left( \frac{2\lambda v }{\Gamma (\alpha -1)}\right)^\ell 
\prod_{j=1}^\ell 
\frac{\Gamma(2j(\alpha-1))}{\Gamma(2j(\alpha-1)+\alpha)}\label{eq:elldispari}.
\end{align}
\end{lemma}

\begin{remark}[Case $u=v=0$] When $u=v=0$, one checks that $a_{1,1}= a_{1,2}=0$, $a_{2,2}= \mu a_{1,2}\frac{\Gamma(\alpha)}{\Gamma(2\alpha)}=0$ and, then, by induction, that $a_{\ell,\ell}=0$ for every $\ell\ge 1$. As a second step, one shows by induction  that, actually,  for every level $\ell\ge 1$, $a_{k,\ell}=0$, $k\ge 1$ so that, like in the former case,  the solution appears in the much simpler  form  
\[
\psi(t)= \psi_0(t) = \sum_{k\ge 1} a_{k,0}t^{\alpha k}.
\]
\end{remark}

\noindent {\bf Proof.} $(a)$ The fact that $k(\ell)=1$ for $\ell=0,1,2$ is obvious. For $\ell\ge 3$, it is clear by adapting the analogous proof in Step~1 of the former case $\alpha\!\in (0,1]$ that $k(\ell)$ is solution to the same recursive optimization problem
\begin{equation}\label{eq:kell2}
k(\ell) =\min\left[ \min_{\ell_1+\ell_2=\ell}\big[k(\ell_1)+k(\ell_2)\big], k(\ell-1)\right]+1
\end{equation}
with the  former initial values (keeping in mind that $a_{2,1}= a_{2,3}=0$). This time the only admissible solution is $k(\ell)=\ell$. 

\smallskip
\noindent $(b)$ This is straightforward.

\smallskip
\noindent $(c)$ 
We proceed by induction.
 First of all, it follows  from Equation ~\eqref{eq:akell2} that 
\begin{align}
a_{2, 2} & =  \frac{\mu }{2\lambda} \left(\frac{2\lambda v }{\Gamma (\alpha -1)}\right)
\frac{\Gamma(\alpha-1)}{\Gamma(\alpha)},
\end{align}
which agrees with the Formula ~\eqref{eq:ellpari}. Assume now the formula valid for $\ell $ and let us check for $\ell +1$. One checks by inspecting successively the cases $\ell=2$, $\ell=3$ and $\ell\ge 4$ that 
\[
a^{*2}_{\ell,\ell+1} = 2a_{1,2} a_{\ell-1,\ell-1},\quad \ell\ge 2,
\]
whereas $a^{*2}_{1,2}=a^{*2}_{0,1}=0$. Likewise  
\begin{align*}
a_{3,3}&=  
\lambda a^{*2}_{2,3}\frac{\Gamma\big(2(\alpha -1)\big)}{\Gamma\big(2(\alpha -1)+\alpha\big)} =\frac{2\lambda v }{\Gamma (\alpha -1)}\frac{u }{\Gamma(\alpha)}
\frac{\Gamma\big(2(\alpha -1)\big)}{\Gamma\big(2(\alpha -1)+\alpha\big)},
\end{align*}
which agrees with ~\eqref{eq:elldispari}. Therefore we get
\begin{align*}
a_{2(\ell +1),2(\ell +1)}&=  \lambda a^{*2}_{2\ell +1,2(\ell +1)}\frac{\Gamma\big((2\ell +1)(\alpha -1)\big)}{\Gamma\big((2\ell +1)(\alpha -1)+\alpha\big)} = 2\lambda a_{1,2}a_{2\ell,2\ell }\frac{\Gamma\big((2\ell +1)(\alpha -1)\big)}{\Gamma\big((2\ell +1)(\alpha -1)+\alpha\big)}\\
&=\frac{2\lambda v }{\Gamma (\alpha -1)}\frac{\mu }{2\lambda} \left(\frac{2\lambda v }{\Gamma (\alpha -1)}\right)^\ell 
\frac{\Gamma\big((2\ell +1)(\alpha -1)\big)}{\Gamma\big((2\ell +1)(\alpha -1)+\alpha\big)}
\prod_{j=1}^\ell 
\frac{\Gamma((2j-1)(\alpha-1))}{\Gamma((2j-1)(\alpha-1)+\alpha)},
\end{align*}
which also agrees with Formula ~\eqref{eq:ellpari}. One proceeds likewise for $a_{2(\ell +1)+1,2(\ell +1)+1}$.
$\qquad_{\diamondsuit}$

\bigskip Now we are in position to prove Theorem~\ref{thm:non-homog}.
Our aim in this proof is to propagate an upper-bound of the form
\begin{equation}\label{eq:propag4}
|a_{k,\ell}|\le C\theta^{\ell}\rho^k (k-k(\ell)+1)^{\frac{\alpha}{2}-1}(\ell\vee1)^{\frac{\alpha}{2}-1},\; k\ge k(\ell), \; \ell\ge 0. 
\end{equation}


\smallskip
\noindent {\sc \textbf{Step 1}} {\em Propagation for the initial value across the levels $\ell\ge 1$}.  On checks by  inspecting successively the cases $\ell=2$, $\ell=3$ and $\ell\ge 4$ that 
\[
a^{*2}_{\ell,\ell+1} = 2a_{1,2} a_{\ell-1,\ell-1},\quad \ell\ge 2,
\]
whereas $a^{*2}_{1,2}=a^{*2}_{0,1}=0$. Consequently, it follows from~~\eqref{eq:Convol2} and the fact that $a_{\ell,\ell+1}=0$, that,
for every $\ell\ge 2$, 
\begin{align*}
a_{\ell+1,\ell+1} &= \lambda a^{*2}_{\ell,\ell+1}\frac{\Gamma((\alpha-1)\ell)}{\Gamma((\alpha-1)\ell+\alpha)} =  2 \lambda  a_{1,2} a_{\ell-1,\ell-1} \frac{\Gamma((\alpha-1)\ell+1)}{\Gamma((\alpha-1)\ell+\alpha)}\frac{1}{(\alpha-1)\ell}\\
&=2 \lambda  a_{1,2} a_{\ell-1,\ell-1} \frac{\Gamma((\alpha-1)(\ell+1)+2-\alpha)}{\Gamma((\alpha-1)(\ell+1)+1)}\frac{1}{(\alpha-1)\ell} .
\end{align*}


By Kershaw's Inequality~~\eqref{eq:Kershawb} used with $x= (\alpha-1)(\ell+1)>0$ and $s= 2-\alpha\!\in (0,1)$, we deduce that
\begin{align*}
|a_{\ell+1,\ell+1} |& \le 2 \lambda  |a_{1,2} | |a_{\ell-1,\ell-1}|\frac{1}{(\alpha-1)\ell} \Big((\alpha-1)(\ell+1) +\frac{2-\alpha}{2}\Big)^{1-\alpha}\\
&= 2 \lambda  |a_{1,2} | |a_{\ell-1,\ell-1}|\frac{1}{(\alpha-1)\ell} \Big((\alpha-1)\ell +\frac{\alpha}{2}\Big)^{1-\alpha}\\
&\le  2 \lambda  |a_{1,2} | |a_{\ell-1,\ell-1}| (\alpha-1)^{-\alpha}\ell^{-\alpha}
\end{align*}
since $\alpha>1$. We assume now that~~\eqref{eq:propag4} is satisfied at levels $1\le \ell'\le \ell$, in particular for $\ell'= \ell-1$, $\ell\ge 2$. Then 
\begin{align*}
|a_{\ell+1,\ell+1} |&\le  2  \lambda |a_{1,2}|C (\rho\,\theta)^{\ell-1}(\ell-1-k(\ell-1)+1)^{\frac{\alpha}{2}-1}(\ell-1)^{\frac{\alpha}{2}-1}(\alpha-1)^{-\alpha}\ell^{-\alpha}.
\end{align*}

One checks that, for every $\ell\ge 2$,  $\ell-1-k(\ell-1)+1\ge 1=\ell+1-k(\ell+1)+1  $ and $\frac{\ell-1}{\ell+1}\ge \frac 13$. Consequently, if we set 
\[
\kappa^{(5)}_{\alpha}=3^{1-\frac{\alpha}{2}}(\alpha-1)^{-\alpha}
\]
we obtain
\[
|a_{\ell+1,\ell+1} |\le C (\rho\,\theta)^{\ell+1}(\ell+1-k(\ell+1)+1)^{\frac{\alpha}{2}-1}(\ell+1)^{\frac{\alpha}{2}-1}\frac{2\kappa^{(5)}_{\alpha}  \lambda |a_{1,2}|}{(\rho\,\theta)^2}.
\]
Consequently, keeping in mind that $a_{1,2}= \frac{v}{\Gamma(\alpha-1)}$, the propagation condition on the initial values reads
\[
|a_{1,0}|= \frac{|\nu|}{\Gamma(\alpha+1)}\le C\rho,\quad \frac{|u|}{\Gamma(\alpha)}\le \rho\,\theta \quad \mbox{ and }\quad 2\kappa^{(5)}_{\alpha}  \lambda \frac{|v|}{\Gamma(\alpha-1)}\le (\rho\,\theta)^2
\]
or, equivalently, 
\begin{equation}\label{eq:propinit2}
C\rho \ge \frac{|\nu|}{\Gamma(\alpha+1)}\quad \mbox{ and }\quad \rho\, \theta\ge C_1:=\max\left(\frac{|u|}{\Gamma(\alpha)}, \sqrt{2\kappa^{(5)}_{\alpha}  \lambda \frac{|v|}{\Gamma(\alpha-1)}}\,\right).
\end{equation}

\smallskip
\noindent {\sc \textbf{Step 2}} {\em Propagation across the levels $\ell\ge 0$}.  We assume that the  bound to be propagated holds for every  couple $(k',\ell')$ such that level $\ell'<\ell$ and  $k'\ge k(\ell')$ or $\ell'=\ell$ and  $k(\ell)\le k'\le k-1$. 

We first focus on the discrete time convolution. Let $k\ge k(\ell)+1$
\begin{align*}
|a^{*2}_{k-1,\ell} |&\le \;C^2 \sum_{\begin{smallmatrix}k_1+k_2= k-1\\ \ell_1+\ell_2=\ell\\ k_i\geq k(\ell_i), \,  \ell_i \ge 0\end{smallmatrix}}|a_{k_1, \ell_1}| |a_{k_2,\ell_2}|\\
& \le  \;C^2 \theta^{\ell} \rho^{k-1}\sum_{\begin{smallmatrix}  \ell_1+\ell_2=\ell\\   \ell_i\ge 0\end{smallmatrix}}\sum_{\begin{smallmatrix}k_1+k_2= k-1 \\ k_i\geq k(\ell_i)\end{smallmatrix}}(k_1-k(\ell_1)+1)^{\frac{\alpha}{2}-1} (\ell_1\vee 1)^{\frac{\alpha}{2}-1}(k_2-k(\ell_2)+1)^{\frac{\alpha}{2}-1}(\ell_2\vee 1)^{\frac{\alpha}{2}-1}\\
&= \; C^2 \theta^{\ell}\rho^{k-1} \sum_{\begin{smallmatrix}  \ell_1+\ell_2=\ell\\   \ell_i\ge 0\end{smallmatrix}}( \ell_1\vee 1)^{\frac{\alpha}{2}-1}(\ell_2\vee1)^{\frac{\alpha}{2}-1} \sum_{\begin{smallmatrix}k'_1+k'_2= k-(k(\ell_1)+k(\ell_2))+1\\ k'_i\ge 1\end{smallmatrix}}(k'_1)^{\frac{\alpha}{2}-1}(k'_2)^{\frac{\alpha}{2}-1}\\
&=  \; C^2 \theta^{\ell}\rho^{k-1}\sum_{\begin{smallmatrix}  \ell_1+\ell_2=\ell\\   \ell_i\ge 0\end{smallmatrix}}( \ell_1\vee 1)^{\frac{\alpha}{2}-1}(\ell_2\vee1)^{\frac{\alpha}{2}-1} B\Big(\frac{\alpha}{2},\frac{\alpha}{2}\Big)\big(k-(k(\ell_1)+k(\ell_1))\big)^{\alpha-1}\\
&\le  C^2 \theta^{\ell}\rho^{k-1} B\Big(\frac{\alpha}{2},\frac{\alpha}{2}\Big)\big(k-k(\ell)+1\big)^{\alpha-1}\sum_{\begin{smallmatrix}  \ell_1+\ell_2=\ell\\   \ell_i\ge 0\end{smallmatrix}}( \ell_1\vee 1)^{\frac{\alpha}{2}-1}(\ell_2\vee1)^{\frac{\alpha}{2}-1} 
\end{align*}
where we used Lemma~\ref{lem:f} in the penultimate line and  $k(\ell_1)+k(\ell_1)\ge k(\ell)-1$ (see~\eqref{eq:kell2}). 
Now note that, if $\ell\ge 1$, 
\begin{align*}
\sum_{\begin{smallmatrix}  \ell_1+\ell_2=\ell\\   \ell_i\ge 0\end{smallmatrix}}( \ell_1\vee 1)^{\frac{\alpha}{2}-1}(\ell_2\vee1)^{\frac{\alpha}{2}-1} &\le \ell^{\frac{\alpha}{2}-1}+\sum_{\ell_1=1}^{\ell-1}\ell_1^{\frac{\alpha}{2}-1} (\ell-\ell_1)^{\frac{\alpha}{2}-1} \le \ell^{\frac{\alpha}{2}-1}+B\Big(\frac{\alpha}{2},\frac{\alpha}{2}\Big)\ell^{\alpha-1}
\end{align*}
owing to Lemma~\ref{lem:f}. If $\ell=0$, the above inequality still holds  since $\displaystyle \sum_{\begin{smallmatrix}  \ell_1+\ell_2=\ell\\   \ell_i\ge 0\end{smallmatrix}}( \ell_1\vee 1)^{\frac{\alpha}{2}-1}(\ell_2\vee1)^{\frac{\alpha}{2}-1} =1$. Now,   combining these inequalities yields
\begin{align}
\nonumber |a^{*2}_{k-1,\ell} | & \le  C^2 \theta^{\ell}\rho^{k-1} B\Big(\frac{\alpha}{2},\frac{\alpha}{2}\Big) \big(k-k(\ell)+1\big)^{\alpha-1}(\ell\vee 1)^{\alpha-1}\left(B\Big(\frac{\alpha}{2},\frac{\alpha}{2}\Big)+(\ell\vee 1)^{-\frac{\alpha}{2}}\right)\\
\label{eq:convol4}
 &\le  C^2 \theta^{\ell} \rho^{k-1} \bar B\Big(\frac{\alpha}{2}\Big) \big(k-k(\ell)+1\big)^{\alpha-1}(\ell\vee 1)^{\alpha-1},
\end{align}
where $\bar B(\alpha/2)$ is defined in~~\eqref{eq:Bbaralpha/2}.

First note, by inspecting the four cases $\ell=0,1,2$ and $\ell\ge 3$, that 
\[
\forall\, \ell\ge 0, \quad \alpha k-\ell\ge \alpha (k(\ell)+1)-\ell>(\ell\vee 2)(\alpha-1)>0.
\]
Now, using $\Gamma(z+1)=z\Gamma(z)$ and Kershaw's Inequality~~\eqref{eq:Kershawb} with $x= \alpha k-\ell\ge \alpha (k(\ell)+1)-\ell>2(\alpha-1)$ and $s =2-\alpha\!\in [0,1)$,  we obtain
\begin{align}
\nonumber \frac{\Gamma(\alpha(k-1)-\ell+1)}{\Gamma(\alpha k-\ell+1)}& = \frac{1}{\alpha k -\ell+1-\alpha} 
 \frac{\Gamma(\alpha k-\ell+2-\alpha)}{\Gamma(\alpha k-\ell+1)}\\
\nonumber &\le \frac{1}{\alpha k -\ell+1-\alpha}
\Big(\alpha k -\ell+2-\frac{\alpha}{2}\Big)^{1-\alpha}\\
\label{eq:ratiogamma2}&\le  \frac{\alpha k -\ell+1-\frac{\alpha}{2}}{\alpha k -\ell+1-\alpha}
 \Big(\alpha k -\ell+1-\frac{\alpha}{2}\Big)^{-\alpha}.
\end{align}

As $\alpha k-\ell+1-\alpha\ge \alpha (k(\ell)+1)-\ell +1-\alpha\ge \alpha-1>0$ for every $\ell\ge 0$, we deduce that
\[
\frac{\Gamma(\alpha(k-1)-\ell+1)}{\Gamma(\alpha k-\ell+1)}\le \left(\frac{\alpha}{2(\alpha-1)}\right)  
\Big(\alpha k -\ell+1-\frac{\alpha}{2}\Big)^{-\alpha}.
\]

Now note that
\[
\alpha k -\ell+1-\frac{\alpha}{2}= \alpha(k-k(\ell))+ \alpha k(\ell)-\ell+1-\frac{\alpha}{2}
\]
and that 
$$
\alpha k(\ell)-\ell+1-\frac{\alpha}{2}\ge (\alpha-1)(\ell \vee 2)+1- \frac{\alpha}{2}\ge (\alpha-1)(\ell \vee 1) \mbox{ for every } \ell\ge 0.
$$ 
Hence, using~~\eqref{eq:elemab}, we deduce 
\[
\Big(\alpha k -\ell+1-\frac{\alpha}{2}\Big)^{-\alpha}\le \big(2\alpha(\alpha-1)\big)^{-\frac{\alpha}{2}}\big(k-k(\ell)\big)^{-\frac{\alpha}{2}}(\ell\vee1)^{-\frac{\alpha}{2}}.
\]
Finally, one notes that $\Big(\frac{k-k(\ell)+1}{k-k(\ell)}\Big)^{\frac{\alpha}{2}}\le 2^{\frac{\alpha}{2}}$ to deduce
\[
\frac{\Gamma(\alpha(k-1)-\ell+1)}{\Gamma(\alpha k-\ell+1)}\le \kappa^{(6)}_{\alpha}\Big(k-k(\ell)+1\Big)^{-\frac{\alpha}{2}}  (\ell\vee1)^{-\frac{\alpha}{2}},
\]
where 
\[
\kappa^{(6)}_{\alpha} =  \frac{\alpha}{2(\alpha-1)} \big(\alpha(\alpha-1)\big)^{-\frac{\alpha}{2}}.
\]
Plugging Inequalities~~\eqref{eq:convol4},~~\eqref{eq:ratiogamma2} and the estimate for $a_{k-1, \ell}$ into~~\eqref{eq:akell2} yields 
\begin{align*}
|a_{k,\ell}|&\le C \rho^{k-1}\theta^{\ell} (\ell\vee 1)^{\frac{\alpha}{2}-1}\big(k-k(\ell)+1\big)^{\frac{\alpha}{2}-1}2^{-\frac{\alpha}{2}} \kappa^{(6)}_{\alpha}\\
&\hskip 2cm \times \left[|\mu|\left(\frac{k-k(\ell)+1}{k-k(\ell)}\right)^{1-\frac{\alpha}{2}}+C |\lambda| \left(\frac{k-k(\ell)+1}{k-k(\ell)}\right)^{\frac{\alpha}{2}}  
  (\ell\vee 1)^{\frac{\alpha}{2}}\bar B\Big(\frac{\alpha}{2}\Big)\right]\\
&\le  C \rho^{k}\theta^{\ell} (\ell\vee 1)^{\frac{\alpha}{2}-1}\big(k-k(\ell)+1\big)^{\frac{\alpha}{2}-1}\frac{\kappa^{(6)}_{\alpha}}{\rho}\left[2|\mu|+ C|\lambda|   
\bar B\Big(\frac{\alpha}{2}\Big)\right],
\end{align*}
where we used that $\sup_{k\ge k(\ell)+1} \frac{k-k(\ell)+1}{k-k(\ell)}= 2$.
We deduce that the propagation of the bound holds as soon as
\begin{equation}\label{eq:propaglk2}
\kappa^{(6)}_{\alpha}\left[2|\mu|+ |\lambda| C  
\bar B\Big(\frac{\alpha}{2}\Big)  \right]\le \rho. 
\end{equation}

\smallskip
\noindent {\sc \textbf{Step 3}} {\em Synthesis}.  If we saturate the left-hand side of inequality~~\eqref{eq:propinit2} and plug it in~~\eqref{eq:propaglk2}, we obtain the inequality 
\[
\rho^2-2\kappa^{(6)}|\mu|\rho -\frac{|\lambda||\nu|}{\Gamma(\alpha+1)}\kappa^{(6)}\bar B\Big(\frac{\alpha}{2}\Big)\geq 0.
\]
 The minimal solution is given by
\[
\rho_*=\rho_*(\alpha,\theta)= \max\left[\kappa^{(6)}\left(|\mu|+\sqrt{|\mu|^2+\frac{|\lambda| |\nu|  \bar B(\alpha/2)}{\Gamma(\alpha+1)\kappa^{(6)} }}\,\right),\frac{C_1}{\theta}\right]
\]
where $C_1=C_1(u,v)$ is given by  the right-hand side of inequality~~\eqref{eq:propinit2} and
\[ 
C_*= C_*(\alpha,\theta)= \frac{|\nu|}{\Gamma(\alpha+1)\rho_*(\alpha,\theta)}.
\]

Now, we focus on the convergence of the series $\psi(t)= \sum_{\ell\ge 0}\psi_{\ell}(t)$ (keeping in mind that the first three levels $\ell=0,1,2$ have no  influence on  the result) so that we may use that  $k(\ell)=\ell$, $\ell\ge 3$. Set $C_*=C_*(\alpha,\theta)$ and $\rho_*=\rho_*(\alpha,\theta)$. One checks that for every $t\!\in I_*= \big (0,\rho_*^{-\frac{1}{\alpha}}\big)$,
\begin{align*}
|\psi_{\ell}(t) |\le \sum_{k\ge \ell}|a_{k,\ell}|t^{\alpha k}&\le C_*
\theta^{\ell}\ell^{\frac{\alpha}{2}-1}\sum_{k\ge \ell} (k-\ell+1)^{\frac{\alpha}{2}-1}\rho_*^{k}t^{\alpha k -\ell}\\
&=  C_*\theta^{\ell}\ell^{\frac{\alpha}{2}-1}\rho_*^{\ell-1} t^{(\alpha-1)\ell-\alpha} \widetilde \psi(t)
\end{align*} 
where $\widetilde \psi_2(t) = \displaystyle   \sum_{k\ge 1} k^{\frac{\alpha}{2}-1}\rho_*^{k}t^{\alpha k} $ is normally convergent on every compact interval $K$ of  the open interval $I_*$. Then, for every $t\!\in K$,  
\begin{align*}
\sum_{\ell\ge 3} |\psi_{\ell}(t)| &\le C_* \sup_{t\in K}\big[ t^{-\alpha}\widetilde \psi_2(t)\big]  \sum_{\ell\ge 3}(\theta\rho_*)^{\ell}\ell^{\frac{\alpha}{2}-1} t^{(\alpha-1)\ell} < C_* \big[ t^{-\alpha}\widetilde \psi_2(t)\big]  \sum_{\ell\ge 3}(\theta\rho_*)^{\ell}\ell^{\frac{\alpha}{2}-1} \rho_*^{-\frac{\alpha-1}{\alpha}\ell}\\
& = C_*  \big[ t^{-\alpha}\widetilde \psi_2(t)\big] \sum_{\ell\ge 3}(\theta\rho^{\frac{1}{\alpha}}_*)^{\ell}\ell^{\frac{\alpha}{2}-1} .
\end{align*}
Hence, the series is absolutely convergent if $  \theta \rho^{\frac{1}{\alpha}}_*(\alpha,\theta)<1$.

As $\alpha>1$, one shows that the function $\theta\mapsto \theta \rho^{\frac{1}{\alpha}}_*(\alpha,\theta)$ satisfies $\displaystyle \lim_{\theta\to 0}\theta \rho^{\frac{1}{\alpha}}_*(\alpha,\theta)=0$ so that we may set 
\[
\theta_* = \inf\big\{\theta>0: \theta \rho^{\frac{1}{\alpha}}_*(\alpha,\theta)\ge 1\big\}<+\infty\quad\mbox{which satisfies $\theta_* \rho^{\frac{1}{\alpha}}_*(\alpha,\theta_*)=1$}.
\]
Finally, one checks that the doubly indexed series $\psi$ is normally convergent on $K$. \hfill $\Box$
\end{appendix}

\section*{Acknowledgments.}
We thank Omar El Euch and Mathieu Rosenbaum for useful discussions on the first part of the paper. We are also grateful to Elia Smaniotto and Giulio Pegorer for valuable comments.

\small
\bibliography{biblio}
\bibliographystyle{apa}

\end{document}